\documentclass[11pt,final]{article}
\pdfoutput=1

\usepackage[utf8]{inputenc}
\usepackage{lmodern}
\usepackage{amsmath, amsthm, amssymb, mathtools}
\usepackage{amsfonts}
\usepackage{bm}
\usepackage{booktabs}
\usepackage{color}
\usepackage{enumitem}
\usepackage{etoolbox}
\usepackage{inconsolata}
\usepackage[left=1in,top=1in,right=1in,bottom=1in]{geometry} 
\usepackage{graphicx}
\usepackage{mleftright}\mleftright
\usepackage{caption}
\usepackage{microtype}
\usepackage[pagebackref]{hyperref}
\usepackage[backgroundcolor=todopink,linecolor=todopink,textsize=footnotesize,obeyFinal]{todonotes}
\usepackage{thmtools}
\usepackage{thm-restate}
\usepackage{url}
\usepackage{verbatim}
\usepackage[linesnumbered,ruled,procnumbered]{algorithm2e}
\usepackage[nameinlink,capitalize]{cleveref}

\hypersetup{
    pdftitle={Optimal learning of quantum Hamiltonians from high-temperature Gibbs states}, 
    pdfauthor={Jeongwan Haah, Robin Kothari, and Ewin Tang}, 
    colorlinks=true, 
    linkcolor=blue, 
    citecolor=blue, 
    urlcolor=blue 
}

\renewcommand{\backref}[1]{}

\renewcommand{\backrefalt}[4]{%
\ifcase #1 %
\or 
[p.\ #2]%
\else 
[pp.\ #2]%
\fi}

\declaretheorem[numberwithin=section]{theorem}
\declaretheorem[sibling=theorem]{lemma}
\declaretheorem[sibling=theorem]{corollary}
\declaretheorem[sibling=theorem,name=Proposition]{prop}
\theoremstyle{definition}

\declaretheorem[sibling=theorem]{definition}
\declaretheorem[sibling=theorem]{remark}

\makeatletter
\newcommand{\para}{%
  \@startsection{paragraph}{4}%
  {\z@}{2ex \@plus 3.3ex \@minus .2ex}{-1em}%
  {\normalfont\normalsize\bfseries}%
}
\makeatother

\DeclarePairedDelimiter{\abs}{\lvert}{\rvert}
\DeclarePairedDelimiter{\norm}{\lVert}{\lVert}
\DeclarePairedDelimiter{\angles}{\langle}{\rangle}
\definecolor{todopink}{HTML}{ffbae6}
\definecolor{diffcolor}{HTML}{cf327b}
\newcommand{\inorm}[1]{\norm{{#1}}_{\infty}}
\newcommand{\iinorm}[1]{\norm{{#1}}_{\infty \to \infty}}
\DeclareMathOperator{\poly}{poly}
\DeclareMathOperator{\polylog}{polylog}

\DeclareMathOperator{\Tr}{Tr}
\DeclareMathOperator{\Supp}{Supp}
\DeclareMathOperator{\Imag}{Im}
\DeclareMathOperator{\D}{D_{KL}}
\DeclareMathOperator*{\E}{E}
\newcommand{\dist}{{\mathrm{dist}}}

\DeclareMathOperator{\gr}{\mathrm{Gra}}

\newcommand{\eps}{\varepsilon}
\newcommand{\diff}{\partial}
\newcommand{\mdiff}{{\mathbf{\mathcal D}}}

\newcommand{\ZZ}{{\mathbb{Z}}}
\newcommand{\RR}{{\mathbb{R}}}
\newcommand{\CC}{{\mathbb{C}}}

\newcommand{\calZ}{{\mathcal{Z}}} 
\newcommand{\symm}{\operatorname{S}} 
\newcommand{\fun}{{\mathcal{F}}}
\newcommand{\est}{\hat{E}}
\newcommand{\rd}{{\mathrm{d}}}
\newcommand{\truncM}{{\mathfrak{m}}}
\newcommand{\textdiff}[1]{\textcolor{diffcolor}{#1}}

\newcommand{\fulldimension}{{\mathsf{D}}} 
\newcommand{\betaF}{{\mathcal{L}}} 
\newcommand{\numTerms}{M}
\newcommand{\numQubits}{N}
\newcommand{\graph}{{\mathfrak{G}}}

\newcommand{\degree}{{\mathfrak{d}}}

\newcommand{\cluster}[1]{{\mathbf{#1}}} 
\newcommand{\partition}[1]{{\mathsf{#1}}} 
\DeclareMathOperator{\marked}{Mar} 
\newcommand{\gpart}[1]{{\mathfrak{#1}}} 
\DeclareMathOperator{\gpartcon}{{PaC}} 

\newcommand{\vct}[1]{{{#1}}}

\newcommand{\multisetnumber}[2]{{\left(\!\!\middle(\genfrac{}{}{0ex}{}{#1}{#2}\middle)\!\!\right)}}

\begin{document}

\title{Optimal learning of quantum Hamiltonians\\from high-temperature Gibbs states}

\author{
    Jeongwan Haah\footnote{Microsoft Quantum and Microsoft Research, Redmond, WA, USA.}
    \and 
    Robin Kothari\footnote{Microsoft Quantum and Microsoft Research, Redmond, WA, USA.} 
    \and 
    Ewin Tang\footnote{University of Washington, Seattle, WA, USA. Some of this work was performed while E.T.\ was a research intern at Microsoft Quantum.} 
}
\maketitle

\begin{abstract}
    We study the problem of learning a Hamiltonian $H$ to precision $\eps$, supposing we are given copies of its Gibbs state $\rho=\exp(-\beta H)/\Tr(\exp(-\beta H))$ at a known inverse temperature $\beta$.
    Anshu, Arunachalam, Kuwahara, and Soleimanifar~\cite{AAKS21} recently studied the sample complexity (number of copies of $\rho$ needed) of this problem for geometrically local $\numQubits$-qubit Hamiltonians.
    In the high-temperature (low $\beta$) regime, their algorithm has sample complexity poly$(\numQubits, 1/\beta,1/\eps)$ and can be implemented with polynomial, but suboptimal, time complexity.

    In this paper, we study the same question for a more general class of Hamiltonians. We show how to learn the coefficients of a Hamiltonian to error $\eps$ with sample complexity $S = O(\log N/(\beta\eps)^{2})$ and time complexity linear in the sample size, $O(S \numQubits)$.
    Furthermore, we prove a matching lower bound showing that our algorithm's sample complexity is optimal, and hence our time complexity is also optimal.
    
    In the appendix, we show that virtually the same algorithm can be used to learn $H$ from a real-time evolution unitary $e^{-it H}$ in a small $t$ regime with similar sample and time complexity.
\end{abstract}

\tableofcontents

\listoftodos

\section{Introduction}

In this paper we study a problem that is at the intersection of quantum many-body physics and machine learning: learning the Hamiltonian of a quantum system from copies of its Gibbs state. 
This problem has recently received much attention in the quantum community~\cite{BAL19,QR19,EHF19,BGP+20}, and the classical analogue of this task, learning undirected graphical models or Markov random fields (MRFs), is well studied in the machine learning community~\cite{KS01,AKN06,sw12,bms13,Bresler2015,VMLC16,km17}.

\para{Motivation.}
This problem has a straightforward physical motivation.
The Hamiltonian $H$ of a quantum system is an operator that tells us how the constituents of the system interact with each other and how the system evolves in time, which is governed by the Schr\"odinger equation. 
The Hamiltonian also tells us what the equilibrium state of the quantum system will be if it is in contact with the environment at a particular temperature and reaches thermal equilibrium. 
This state, which is a function of the temperature and the Hamiltonian, is called the Gibbs state. Formally, for a Hamiltonian $H$ and inverse temperature $\beta$ (i.e., temperature $1/\beta$), the Gibbs state is $\rho = \exp(-\beta H)/\Tr(\exp(-\beta H))$.

In the Hamiltonian learning problem, we imagine that we have a system governed by an unknown Hamiltonian $H$ from a known class of physically reasonable Hamiltonians, such as geometrically local Hamiltonians, and we have access to copies of its Gibbs state at a known inverse temperature $\beta$.
These copies, for example, result from leaving the system to interact with the environment at a known temperature and stabilize: eventually, the system is described by the Gibbs state. 
Our goal is to learn the Hamiltonian $H$, from the assumed class of Hamiltonians, while minimizing the number of copies of $\rho$ required and the running time of the algorithm.
These are called the {sample complexity} and {time complexity} of the algorithm.

The classical analogue of Hamiltonian learning is the problem of learning undirected graphical models or Markov random fields.
This is, in fact, a special case of Hamiltonian learning where everything is classical, which means that the Hamiltonian is a diagonal operator, and consequently the Gibbs state is a diagonal density operator, which is just a sample from a classical probability distribution.
The goal is again to learn the parameters of the classical Hamiltonian from these samples.
This classical problem has been studied for over 50 years, usually in the harder setting of learning the terms and parameters of the classical Hamiltonian, starting with the work of \cite{cl68}, to more recent works that provide nearly-sample-optimal algorithms, with time-efficient implementations \cite{sw12,bms13,Bresler2015,VMLC16,km17}.
Markov random fields find applications in a variety of areas including computer graphics, vision, economics, sociology, and biology~\cite{KS80,Cli90,Lau96,JEMF06,KF09,Li09}, so studying its quantum generalization is well-motivated independently from its physical motivation.

\para{Problem statement.}
The formal statement of the Hamiltonian learning problem is as follows.
Consider a quantum system of $\numQubits$ qubits and a Hamiltonian $H=\sum_{a=1}^\numTerms \lambda_a E_a \in \CC^{2^\numQubits \times 2^\numQubits}$ consisting of $\numTerms$ terms, where the operators $E_a \in \CC^{2^\numQubits \times 2^\numQubits}$ are known, distinct, non-identity Pauli operators%
\footnote{
    It is not essential that these are Pauli operators, but it is convenient for us that the entries of Pauli matrices are small integers, which allows us compute quantities of interest exactly and not worry about numerical precision.
}
and the coefficients satisfy $\lambda_a \in [-1,1]$ for all $a \in [\numTerms] = \{1,\ldots, \numTerms\}$.
We assume the Hamiltonian has no identity term since the Gibbs state is invariant under adding multiples of the identity matrix to the Hamiltonian.
We assume terms are distinct because identical terms can be merged.
Further suppose that this Hamiltonian $H$ is low-intersection (a constraint defined below).

Given copies of the Gibbs state of this unknown low-intersection Hamiltonian $H$ and known inverse temperature $\beta$, our goal is to learn the coefficients $\lambda_a$ to additive error $\eps$, or equivalently, to learn the vector of coefficients to error $\eps$ in $\ell_\infty$ norm.
Previous work on the problem has also considered the goal of learning this vector to $\ell_2$ norm, so we study this version of the problem as well. 

We define the class of \emph{low-intersection Hamiltonians} to be the set of Hamiltonians where each operator $E_a$ is supported on a constant number of qubits, meaning that it acts as the identity operator on all but a constant number of qubits (which are its support), and for each operator $E_a$, there are only a constant number of other operators $E_b$ such that $E_a$ and $E_b$'s supports have nontrivial intersection.

Notice that this definition has no geometric constraints.
Instead, it generalizes geometrically local Hamiltonians in fixed-dimensional Euclidean spaces, which is the class of physically motivated Hamiltonians with geometric constraints considered in prior work~\cite{AAKS21}.
In such Hamiltonians, each operator $E_a$ is only supported on a constant number of qubits that are adjacent in the underlying geometry (e.g., a 2-dimensional grid).
Since the dimension is fixed and interactions must be local, each operator $E_a$ can only act nontrivially on a constant number of qubits, and furthermore each qubit can only be nontrivially involved in a constant number of operators $E_a$.
So, a geometrically local Hamiltonian in any constant-dimensional space is always low-intersection.

The converse is not true, though.
For example, if we arrange qubits on the vertices of a constant-degree expander graph, and let edges denote $2$-qubit interaction terms, such a Hamiltonian would be a low-intersection Hamiltonian, but not a geometrically local Hamiltonian in any constant-dimensional Euclidean space.
In this introduction we will assume that we have a low-intersection Hamiltonian whose degree is a constant independent of other parameters, although our general algorithm can also handle growing degree. 

\para{Prior work.}
We first discuss the complexity of the classical problem to understand the best we could do, since classical Hamiltonians, also known as Markov random fields, are a special case of quantum Hamiltonians.
The classical problem is then the parameter learning of low-intersection MRFs to $\ell_\infty$ error $\eps$.
The sample complexity and time complexity of this problem are
\begin{equation}\label{eq:classical}
    \frac{2^{O(\beta)}\log \numQubits}{\beta^2\eps^2} \quad\text{and}\quad \frac{2^{O(\beta)} \numQubits \log \numQubits}{\beta^2\eps^2},
\end{equation}
respectively.
The sample complexity is optimal up to the constant in the exponent~\cite{sw12}, so the time complexity, which is the time needed to read all of the samples, is also optimal.
This result appears to be folklore, so in \cref{sec:mrfs}, we give a simple algorithm demonstrating this result.

Most of the classical literature focuses on the harder task of \emph{structure learning}, which is learning the terms of the Hamiltonian in addition to the coefficients, for Ising models, which are classical Hamiltonians with only pairwise interactions.
For structure learning in the Ising model, the same sample complexity bound can be achieved in time only polynomially worse than the sample complexity times the size of each sample~\cite{VMLC16,km17}.

Notice that the problem becomes harder as $\beta \to 0$ and as $\beta \to \infty$. This is intuitive because the state at $\beta = 0$ is the maximally mixed state (or the uniform distribution in the classical case), which contains no information about the Hamiltonian. When $\beta$ tends to $\infty$, the state tends to the ground state of the Hamiltonian, which does not have enough information to reconstruct the entire Hamiltonian.

The quantum version of this problem for geometrically local Hamiltonians was 
recently studied.
The algorithm in \cite{BAL19} allows us to 
learn Hamiltonians from stationary states of Hamiltonian dynamics 
(which include Gibbs states)
or from the dynamics itself
by measuring local observables and solving a system of linear equations;
however, it was unclear how the algorithm would perform in the worst-case.
More recently, Anshu, Arunachalam, Kuwahara, and Soleimanifar~\cite{AAKS21} was the first to rigorously establish sample complexity upper bounds for this problem in the full range of parameters, and in particular, for all inverse temperatures $\beta$.
They showed that a geometrically local Hamiltonian in a constant-dimensional space can be learned to $\ell_\infty$ error $\eps$ using 
\begin{equation}\label{eq:AAKS}
    O\left(\frac{2^{\poly(\beta)}\numQubits^2\log \numQubits}{\beta^c\eps^2}\right)
\end{equation}
samples\footnote{
    Actually, they claim a slightly weaker statement: learning to $\ell_2$ error $\eps$ using $\numQubits$ times the expression in \cref{eq:AAKS} many samples.
    We derive the version stated here in \cref{rm:aaks-inf}.
}, for some constant $c>4$.
Note that for a geometrically local Hamiltonian, the number of terms $N = \Theta(M)$, so we have expressed the bound in terms of $\numQubits$.

This upper bound is worse than the classical sample complexity in \cref{eq:classical} in several regards.
First, it has worse dependence on $\beta$ both in the numerator and the denominator, which means it is worse in the high-temperature and low-temperature regime.
Second, the dependence on $\numQubits$ is quadratic, whereas the dependence on $\numQubits$ is logarithmic in the classical upper bound\footnote{In the $\ell_2$ error setting, though, this classical upper bound has a factor of $\numQubits$, and so the quantum bound is polynomially close to the classical bound in the high-temperature setting.}.
This leaves two natural open questions: Can we solve the quantum problem with sample complexity matching \cref{eq:classical}?
And what about time complexity?

The question of time complexity is not explicitly addressed in \cite{AAKS21}, but they note that the problem can be solved in polynomial time in the high-temperature regime, by combining their algorithm with the polynomial-time algorithm for computing partition functions at high temperatures due to~\cite{kkb20}.
We discuss this approach further in the ``Comparison with previous quantum algorithms'' section, but in brief, this approach leaves significant room for improvement in both sample complexity and time complexity.

\para{Our results.}
We study the Hamiltonian learning problem in the high-temperature regime, and we are able to obtain an algorithm with optimal sample complexity and optimal time complexity. The high-temperature regime is where we know $\beta$ is smaller than some fixed constant called the critical inverse temperature, $\beta_c$. This constant $\beta_c$ depends only on the constant in the definition of a low-intersection Hamiltonian, and not on $\numQubits$ or $\numTerms$.

Our main algorithmic result is the following. A more precise version can be found in \Cref{sec:algorithm}.

\begin{theorem}[Algorithm]\label{thm:algorithm}
    Let $H$ be a low-intersection Hamiltonian on $\numQubits$ qubits, $\eps>0$, and $\beta < \beta_c$. 
    Then we can learn the coefficients of $H$ with $\ell_\infty$ error $\eps$ and failure probability $\delta$ using $O\bigl(\frac{1}{\beta^2\eps^2}\log\frac{\numQubits}{\delta}\bigr)$ samples.
    Consequently, we can learn the coefficients of $H$ with $\ell_2$ error $\eps$ and failure probability $\delta$ using $O\bigl(\frac{\numQubits}{\beta^2\eps^2}\log\frac{\numQubits}{\delta}\bigr)$ samples. 
    In both cases, the time complexity is linear in the sample size, which is the sample complexity multiplied by $N$, the size of each sample.
\end{theorem}

Our upper bound improves on the sample complexity of \cite{AAKS21} and indeed matches the sample complexity of the classical algorithm in the high-temperature regime, where the $2^{O(\beta)}$ term can be dropped since it is constant.
Furthermore, our algorithm has optimal time complexity. 

Along the way, we show that the log-partition function is $(\frac{\beta^2}{2})$-strongly convex in the high-temperature regime; this is the main quantity bounded by \cite{AAKS21} to achieve their sample complexity result, and our analysis improves this strong convexity parameter to within a constant factor of its true value.
As observed in \cite{AAKS21}, this strong convexity bound implies a lower bound on the variance of macroscopic observables in thermal equilibrium.
Specifically, for a local operator $\sum_{a=1}^\numTerms v_a E_a$, its variance with respect to the Gibbs state is $v^\dagger (\nabla^{\otimes 2}  \betaF) v = \Omega(\beta^2 \norm{v}_2^2)$, where $\nabla^{\otimes 2}\betaF$ is the Hessian of the log-partition function, improving on the bound $\Omega(\beta^{c}\norm{v}_2^2/\numQubits)$ implied by~\cite{AAKS21}.

We also prove a matching lower bound on the sample complexity showing that our algorithm's sample complexity cannot be improved. Our lower bound significantly improves on the lower bound shown in \cite{AAKS21} (displayed below in \cref{eq:AAKSlb}), holds for the full range of $\beta$, and matches our algorithm's complexity in the high-temperature regime. A more formal version appears as \Cref{thm:lowerboundinfty} and \Cref{thm:lowerboundtwo}.

\begin{theorem}[Lower bound]\label{thm:lowerbound}
    For any $\eps \in (0,1/2]$, any $\beta>0$, and any $\numQubits$, there exists a 2-local Hamiltonian on $\numQubits$ qubits such that the sample complexity of learning its coefficients to $\ell_\infty$ error $\eps$ and failure probability $\delta$ is $\Omega\left(\frac{\exp(\beta)}{\beta^2\eps^2}\log\frac{\numQubits}{\delta}\right)$, and the sample complexity of learning its coefficients to $\ell_2$ error $\eps$ and constant failure probability is $\Omega\left(\frac{\exp(\beta)\numQubits}{\beta^2\eps^2}\right)$.
\end{theorem}

The Hamiltonians used in our lower bound are extremely simple $2$-local Hamiltonians, where each term acts nontrivially only on 2 qubits and each qubit is involved in only 1 term. This shows that although our algorithms apply to a more general class of Hamiltonians than considered by \cite{AAKS21}, restricting our attention to a simpler class of Hamiltonians will not allow us to improve on the sample complexity compared to our algorithm.

This improves significantly on the lower bound given in \cite{AAKS21}, which states that any algorithm that learns a Hamiltonian to $\ell_2$ error $\eps$ has sample complexity
\begin{equation}\label{eq:AAKSlb}
    \Omega\left(\frac{\sqrt{N}+\log(1-\delta)}{\beta\eps}\right).
\end{equation}

In addition, we observe that virtually the same algorithm can be used to learn a low-intersection Hamiltonian $H$, given black-box access to its real-time evolution unitary $e^{-itH}$, provided $t$ is known and smaller than some critical time that is a constant in the definition of a low-intersection Hamiltonian, which does not depend on $\numQubits$ or $\numTerms$.

\begin{restatable}[Real-time dynamics]{theorem}{algodynamics} \label{thm:algo-dynamics}
    Let $H$ be a low-intersection Hamiltonian on $\numQubits$ qubits
    and let $U = e^{-it H}$ be a blackbox unitary with $t < t_c$.
    Then we can learn the coefficients of $H$ to $\ell_\infty$ error $\eps$ with success probability $1-\delta$, using~$U$ $O\bigl(\frac{1}{t \eps^2}\log \frac \numQubits \delta\bigr)$ times, with time complexity $O\bigl(\frac{\numQubits}{t \eps^2}\log \frac \numQubits \delta\bigr)$.
\end{restatable}

We report this observation in \cref{sec:realtime}.
Prior work on this task \cite{BAL19,ZYLB21} uses measurements of short-time evolutions, i.e.\ time resolution $t = O(\eps)$, to estimate time derivatives, which gives a sample complexity scaling as $1/\eps^4$.
We improve this quadratically, and since we only apply $U$ for $t$ as large as constant, we improve the time resolution to constant.

\para{High-level overview of techniques.}
Our algorithm in \cref{thm:algorithm} proceeds in two steps.
First, we notice that, for sufficiently small (but constant) $\beta$, the Taylor series expansion of the expectation $\Tr(E_a \rho)$ in $\beta$ converges.
This observation follows from the cluster expansion techniques from \cite{KS19}, which essentially describes and bounds the coefficients in the Taylor series expansion of the log-partition function, $\log \Tr \exp(-\beta H)$.
We reproduce these (lengthy but elementary) calculations here, amending some minor issues in their presentation.
The log-partition function is related to the expectation $\Tr(E_a \rho)$ (in fact, $\tfrac{\diff}{\diff \lambda_a} \log \Tr \exp(-\beta H) = -\beta \Tr(E_a \rho)$), so we can use these results on convergence of the log-partition function to get convergence of the expectation.

One notable difference from prior work is our results showing how to efficiently compute the Taylor series expansion described above (\cref{statement:algo-W-deriv}).
Prior work asserted such computation was possible~\cite{kkb20}, but did not provide an explicit algorithm.
We provide an algorithm (\cref{alg:cluster-derivative}), and because the $E_a$'s are Pauli operators in our setting, the matrices are composed of small integers and so this algorithm works in exact arithmetic.

This shows that we can approximate $\Tr(E_a \rho)$, an expression that can be easily estimated from copies of $\rho$, by a polynomial in $\{\lambda_b\}$, the parameters we wish to estimate.
This polynomial is the one we get from truncating the Taylor series expansion of $\Tr(E_a \rho)$.
Once we approximate these $\Tr(E_a \rho)$'s, we are left with the task of solving the system of polynomial equations defined by these truncated Taylor series expansions.
By bounding the $\infty \to \infty$ norm of the inverse Jacobian of this system, we immediately get a bound on the sample complexity (\cref{thm:sample-complexity}).
By performing the Newton--Raphson method for root-finding, we can invert this system efficiently, only needing to compute the (first-order) Jacobian for $O(\log\frac{1}{\beta\eps})$ iterations.
In fact, the Newton--Raphson method performs so efficiently that its runtime is dominated by the runtime of simply reading in the input, making the algorithm as a whole run in linear time (\cref{thm:newton}).

For our lower bound, we use information-theoretic techniques to show that without sufficiently many Gibbs states, the coefficient vector cannot be determined to $\eps$ error.
In particular, we use Fano's lemma to establish a lower bound from a KL-divergence computation, similarly to prior classical work for lower bounds of learning undirected graphical models~\cite{sw12}.
This immediately gives the lower bound in the $\ell_\infty$ case (\cref{thm:lowerboundinfty}), and a simple argument with error correcting codes bootstraps this to an $\ell_2$ bound (\cref{thm:lowerboundtwo}).

\para{Comparison with previous quantum algorithms.}
We now provide a comparison of our algorithm's sample and time complexity bounds compared to that of prior work.
We consider the task of learning a geometrically local Hamiltonian for sufficiently small $\beta$ with success probability $0.9$.
One algorithm to compare to is a naive ``state tomography'' strategy that one can derive from \cite[Theorems~2~and~11]{kkb20}.
These results show that to estimate a Hamiltonian coefficient $\lambda_a$ to $\eps$ error, for sufficiently high temperature, it suffices to know the Gibbs state $\rho$ on a ``patch'', a ball of radius $O(\log\frac{1}{\beta\eps})$ around the support of $E_a$.
So, one can perform state tomography to learn the patch of $\rho$ in time exponential in the number of qubits in the ball, giving an algorithm for Hamiltonian learning with quasi-polynomial sample and time complexity
\begin{align}
    \tilde{O}\Big(e^{\log^c(\frac{1}{\beta \eps})}\log(\numQubits)\Big) \text{ and }
    \tilde{O}\Big(e^{\log^c(\frac{1}{\beta \eps})}\numQubits\log(\numQubits)\Big),
\end{align}
where the factor of $c$ comes from the ambient dimension, i.e.\ a ball of radius $r$ on the lattice of qubits is size $O(r)^c$.

As for \cite{AAKS21}, we already gave the sample complexity bound in \cref{eq:AAKS}.
As for time complexity, \cite{AAKS21} describe an stochastic gradient descent (SGD) approach that solves the task of Hamiltonian learning, assuming one can evaluate the log partition function.
This subroutine is hard in general, but \cite{AAKS21} acknowledge that, for the high-temperature regime, it can be done efficiently~\cite{kkb20} to get a time-efficient algorithm for Hamiltonian learning.
This resulting algorithm will inherit the sub-optimal sample complexity of \cite{AAKS21}, and though \cite{AAKS21} don't give a precise time complexity, we can conclude that at best its time complexity will be linear in the sample size, which is the sample complexity times $N$.
So, the sample and time complexity of \cite{AAKS21} in the high temperature regime is something like
\begin{align}
    O\Big(\frac{\numQubits^2\log(\numQubits)}{\beta^c\eps^2}\Big) \text{ and }
    O\Big(\frac{\numQubits^3\log(\numQubits)}{\beta^c\eps^2}\Big).
\end{align}
The prior work achieves optimal scaling in either $\numQubits$ or $1/\eps$.
We get a sample and time complexity,
\begin{align}
    O\Big(\frac{\log(\numQubits)}{\beta^2\eps^2}\Big) \text{ and }
    O\Big(\frac{\numQubits \log(\numQubits)}{\beta^2\eps^2}\Big),
\end{align}
that is simultaneously optimal in all parameters and improves at least polynomially over prior results.
In fact, in certain regimes (like when $\beta\eps \eqsim \exp(-\log^{1/c}(\numQubits))$), we give a super-polynomial improvement in sample complexity over prior work.

Our algorithm can be viewed as a refinement of the patching argument described in \cite{kkb20}; cluster expansion is the technique used to prove the results there, and we use it in a similar way to argue that it suffices to only consider $O(\frac{1}{\beta\eps})$ terms that are within $\log(\frac{1}{\beta\eps})$ of the support of $E_a$.
Our contribution is that we use this expansion algorithmically via the Newton--Raphson method to improve the quasi-polynomial time from the naive algorithm to polynomial time.

\cite{AAKS21} proceeds by establishing that the log partition function is strongly convex.
This actually immediately gives a sample complexity bound, but \cite{AAKS21} instead provide a concrete algorithm, stochastic gradient descent (SGD), that solves the task of Hamiltonian learning, assuming one can evaluate the log partition function (a hard problem in general).
Specifically, this means that they need to lower bound the smallest eigenvalue of the Hessian of the log partition function, or equivalently, upper bound the spectral norm of the inverse of this matrix.

At its core, our strategy is similar to that of \cite{AAKS21}.
Like in \cite{AAKS21}, we also work with the inverse of the Hessian of the log-partition function (or an approximation of it).
Since we want to solve the problem with $\ell_\infty$ error $\eps$, instead of upper bounding the spectral norm of this matrix, we upper bound its $\infty \to \infty$ norm.
Our bound also yields an upper bound on the spectral norm that is tighter than the bound in \cite{AAKS21}.
The improvement is due to the more precise characterization of this matrix via the series expansion described above.
To get a time-efficient algorithm, we then use the Newton--Raphson method, whose analysis also requires us to understand a higher order derivative of the log partition function than is needed for bounding the sample complexity.
Our characterization through the series expansion is able to provide this higher order information, which allows us to bound the running time of the Newton--Raphson method and show it to be time efficient.

\para{Comparision with previous classical algorithms.}
One might wonder why classical techniques for solving the Hamiltonian learning problem do not apply to quantum Hamiltonians. Our algorithm and the \cite{AAKS21} algorithm do not use strategies that are common in the classical literature. 

The reason is that classical algorithms for learning Hamiltonians rely on a property of the classical Gibbs state called the Markov property. To understand this property, partition the set of bits into $3$ disjoint parts $A$, $B$, and $C$, such that there is no term in the Hamiltonian that has a bit from $A$ and $C$. The sets $A$ and $C$ only interact through $B$. Now it is not hard to show that if we condition the classical Gibbs distribution of this Hamiltonian on the values taken by bits in $B$, the resulting distribution on $A$ and $C$ is independent. In fact, the Hammersley--Clifford theorem shows that this is not just a property of Gibbs states, but this property characterizes Gibbs states~\cite{HC71}.
This property fails to hold in general for quantum Gibbs states, although it can hold for special classes of Hamiltonians, such as commuting Hamiltonians~\cite{BP12}.
For high-temperature Gibbs states, this holds only approximately, as cluster expansion formalizes: roughly, in this setting, $A$ and $C$ can be treated as independent subsystems provided $B$ is ``wide'' enough.

Classical algorithms can efficiently perform structure learning by treating it as parameter learning on the full space of $k$-local Pauli matrices.
This would naively take exponential time, but algorithms are still able to use the low-intersection guarantee, despite not knowing anything else about the terms~\cite{km17}.
It's not clear how to show a similar statement in the quantum setting; we can apply our algorithm, but it only works for $\beta$ smaller than $1/\poly(\numQubits)$.

\section{Preliminaries}

Throughout the paper all the exponential and logarithm functions ($\exp,\log$)
are with natural base $e = \sum_{k=0}^\infty \frac{1}{k!} \approx 2.718$. For a vector $v$, $\norm{v}=\sum_i |v_i|^2$ denotes the Euclidean (or $\ell_2$) norm. For a matrix $M$, we use $\|M\|=\max_{v\neq 0} \frac{\|Mv\|}{\|v\|}$ to denote the operator norm (also known as the spectral norm, $2\to 2$ norm, or the Schatten $\infty$-norm).

\subsection{Notations and conventions} \label{sec:notation}

\begin{definition}[Hamiltonian] \label{defn:hamiltonian}
    A \emph{Hamiltonian} is a collection of tuples $(a,E_a,\lambda_a)$,
    where $a$ is an index ranging over some finite set of $\numTerms$ elements,
    which we usually take to be $[\numTerms] = \{1,2,\ldots,\numTerms\}$;
    the \emph{Hamiltonian term} $E_a \in \CC^{\fulldimension \times \fulldimension}$ is a Hermitian operator with $\norm{E_a}\leq 1$ acting on a Hilbert space of dimension $\fulldimension$;
    and the \emph{Hamiltonian term coefficient} $\lambda_a \in [-1,1]$ is a real number.
    We use the notation $\lambda = (\lambda_1,\ldots,\lambda_M)$ for the vector of coefficients.
    The associated Hamiltonian operator $H$ is defined to be $H = \sum_a \lambda_a E_a$.
\end{definition}

Our full algorithm will require that $E_a$ are distinct, non-identity Pauli matrices.
This assumption is neither essential nor too constraining.
Since a Hamiltonian is Hermitian, we definitely want $E_a$ to be Hermitian.
Requiring that $\Tr E_a = 0$ is simply a shift in the eigenspectrum of the Hamiltonian.
Since Pauli operators form an orthonormal basis for operators,
we can always write any Hamiltonian term as a sum of Pauli operators.
The dual interaction graph degree may increase in this rewriting by a factor 
that is at most the exponential of the number of qubits in the support of the Hamitonian terms;
however, in the arguably most important scenario where a term acts on a constant number of qubits,
this blowup is a constant multiplicative factor.

We always assume that some system of $\numQubits$ qubits comprises the Hilbert space,
so $\fulldimension = 2^\numQubits$ is some power of two.
Upon introducing this decomposition of the Hilbert space into qubits,
we can define the support $\Supp(P)$ of an operator $P$.
The support is the minimal set of qubits such that $P$ can be written as $P = O_{\Supp(P)} \otimes I_{\Supp(P)^c}$ for some operator $O$.
(The superscript $c$ here means the complement.)

\begin{definition}[Dual interaction graph]
    For any Hamiltonian $\{(a, E_a, \lambda_a) : a \in [\numTerms] \}$,
    there is an associated undirected \emph{dual interaction graph} $\graph$
    with vertex set $[\numTerms]$ and an edge between $a$ and $b$ if and only if $a \neq b$ and
    \begin{align}
        \Supp(E_a) \cap \Supp(E_b) \neq \varnothing.
    \end{align}
    We denote by $\degree$ the maximum degree of the graph $\graph$ over all vertices.
\end{definition}

Note that we have defined a Hamiltonian in such a way that it is possible that $E_a = E_b$ for $a \neq b$.
If this is the case, then there is an edge in $\graph$ between $a$ and $b$.
In our learning algorithm we will require that $E_a$'s are distinct nonidentity Pauli operators,
but this definition is sufficient for our series expansion of log-partition functions.

Although we do not specify how $\degree$ depends on~$\numTerms$,
we focus on the case when $\degree$ is a constant independent of~$\numTerms$.
This case encompasses most Hamiltonians classes discussed in the literature.
For example, if every Hamiltonian term $E_a$ acts on a constant number of qubits
and every qubit is involved in a constant number of terms, all with respect to $\numTerms$,
then $\degree$ will be constant as well.
More concretely, if we have a directed graph $G$ with a qubit on each vertex
and a two-qubit Hermitian operator for every edge (which may require direction on each edge),
then the vertices of $\graph$ correspond to the edges of $G$
and the dual interaction graph has $\degree \le 2(d-1)$,
where $d$ is the degree (in-degree plus out-degree) of $G$.

As another example, an important class of Hamiltonians is the class of geometrically local Hamiltonians on, say, Euclidean space $\RR^d$.
There are some constant number of qubits on each point of the lattice $\ZZ^d \subset \RR^d$,
and a Hermitian operator is defined for each unit hypercube.
Here, the dual interaction graph has $\degree \le 3^d - 1$, which is again independent of $\numTerms$.

\begin{definition}[Gibbs state]
    The \emph{Gibbs state} of the Hamiltonian $\{(a,E_a,\lambda_a)\}$ at inverse temperature $\beta > 0$ is given by
    \begin{align}
        \frac{\exp(-\beta H)}{\Tr \exp (-\beta H)} = 
        \exp\biggl(-\beta \sum_a \lambda_a E_a\biggr) \biggr/ \Tr \exp\biggl(-\beta \sum_a    \lambda_a E_a\biggr).
    \end{align}
\end{definition}

It is a trivial but important fact that
the exponential of the Hamiltonian operator always makes sense for any $\beta \in \CC$,
not just positive~$\beta$, 
since the norm of the Hamiltonian is upper bounded by~$\numTerms$.

\subsection{Time complexity}

When we discuss the time complexity of our algorithms, we will usually do so in the standard \emph{word RAM model}, where operations on \emph{words}, integers of $w$ bits with $w \geq \log_2(\numQubits + \numTerms + \eps_{\text{machine}}^{-1})$, take unit time.
This word size is defined so that an index into qubits, an index into terms, and $\beta$ can all be stored in one word.
With this model, the input to the Hamiltonian learning problem (the terms $\{E_a\}_a$ and $\beta$) can be given in $O(L\numTerms)$ words, where $L$ is the maximum support of all the terms $E_a$.
This requires representing a term $E_a$ by its support ($\leq L$ words) and the non-identity Pauli operator that $E_a$ performs on each qubit in its support ($\leq 2L$ bits).

\begin{remark}\label{rm:ham-repr}
    Our algorithms assume that the input to the Hamiltonian learning problem also contains an adjacency-list representation of the dual graph $\graph$ corresponding to the input Hamiltonian.
    That is, we want to query any node $b \in [\numTerms]$ to receive a list of its neighbors in $\graph$ in unit time, where the list is given as a random-access array.

    Producing this adjacency-list representation requires only $O(L\numTerms \degree\log\degree)$ time: first, for each qubit $i$, produce a list of the terms that have that qubit in its support, $S_i = \{a \in [M] : i \in \Supp(a)\}$; second, sort the $S_i$'s; third, for every term $a$, produce a list that is the sorted concatenation of every qubit in its support, $\cup_{i \in \Supp(a)} S_i$.
    After removing $a$ itself, this list is the set of neighbors of $a$ in $\graph$.

    The first step takes time linear in the number of edges, so $O(L\numTerms)$ time.
    The second step takes $O(\numQubits\degree\log\degree)$ time, since $\abs{S_i} \leq \degree+1$.
    The third step takes $O(L\degree\numTerms)$ time, since we can merge sorted lists in linear time, removing duplicates as we find them so that we never merge lists of length larger than $\degree + 1$.
    So, the total time complexity is $O(N\degree \log \degree + L\numTerms\degree) = O(L\numTerms \degree\log\degree)$, since $N \leq L\numTerms$.
\end{remark}

As for operations on the quantum computer, we will use the standard model of time complexity (or, rather, gate complexity).
However, since the only quantum operations our algorithm will ever perform is measuring a single-qubit Pauli operator on a qubit of an input Gibbs state, it suffices to just assume that this operation takes unit time.

Finally, in this paper we will ignore issues of numerical stability.
We can do this comfortably because the only portion of our algorithm that is not exact arithmetic is the Newton's method iterations in \cref{alg:newton}.
This algorithm accounts for per-iteration error already, so issues with numerical instability do not arise here.

\subsection{Analytic functions and series expansions}

We will extensively use infinite series expansions of analytic functions,
and here we discuss general principles of handling infinite series.
The material here is all standard in complex analysis.

A complex function $f: D \to \CC$ is defined to be (complex) \emph{analytic} at $a \in D \subseteq \CC$ 
if, for some $\eps > 0$, the function agrees with a power series
\begin{align}
    f(x) = \sum_{k=0}^{\infty} c_k (x-a)^k
\end{align}
for all $x \in D$ such that $\abs{x-a} < \eps$.
Here, the coefficients $c_k \in \CC$ and $\eps \in \RR_{>0}$ may depend on~$a$.
Since such a power series converges uniformly (at least on a small neighborhood of~$a$),
the infinite sum commutes with taking derivatives,
and therefore the coefficients $c_k$ must be those of the Taylor expansion:
\begin{align}
    c_k = \frac{(\diff_x^k f)(a)}{k!}.
\end{align}
A complex function is said to be analytic on an open set $D$ 
if it is analytic at every point in $D$.
A basic theorem in complex analysis is that 
a complex function is complex differentiable (holomorphic) at a point
if and only if it is analytic at that point~\cite[10.14, 10.16]{Rudin}.

Functions with a power series expansion are ``rigid'' in the following sense.
A power series expansion at $a \in D$ is identically zero
if and only if there is an infinite sequence of distinct points $x_1,x_2,\ldots \in D$
such that $\lim_{n \to \infty} x_n = a$ and that the power series is zero at every $x_n$.
This implies a uniqueness theorem of analytic functions~\cite[10.18]{Rudin}:
if two analytic functions $f,g$ on a connected open domain $D$ agree on a subset that has a limit point within $D$,
then $f=g$ on $D$. 
In particular,
\begin{lemma}\label{statement:analyticFunctionUnique}
    Suppose a complex function $f$ is complex differentiable on $E$, an open neighborhood of the origin.
    If its Taylor series at the origin
    \begin{align}
        \sum_{k=0}^{\infty} \frac{f^{(k)}(0)}{k!} z^k
    \end{align}
    uniformly converges on $D = \{z \in \CC: \abs{z} < r\}$ for some $r$,
    then it must converge to $f(z)$ for all $z$ on the connected component of the origin in $D \cap E$.
\end{lemma}

This basically says
that the Taylor series of a holomorphic function can only converge to its function value.
Our typical use of the lemma will be as follows.
We will show that $f$ is complex differentiable on an open set $E$ containing $\RR$
and consider its Taylor expansion at the origin and lower bound the radius of convergence.
Then, the lemma will imply that the series converges to the function value
in the real domain of convergence of the Taylor series. The way we show a function is complex differentiable is simply by noting that it is a composition of functions (such as addition, multiplication, exponentiation, etc.) that are themselves complex differentiable.

\begin{proof}
    Let $g$ be the function on $D$ defined by the series.
    By uniform convergence, the series is differentiable term by term,
    and hence $g$ is complex differentiable everywhere in $D$, 
    and hence is analytic on $D$.
    Since $f$ is analytic on $E$, 
    on some tiny open neighborhood $U \subset D \cap E$ of the origin
    it is represented by its Taylor series,
    which is the same as $g$.
    Since $U$ has a limit point within $D \cap E$, 
    we must have $f = g$ on the connected component of $D\cap E$ containing the origin.
\end{proof}

\cref{statement:analyticFunctionUnique} is not true if we only assume that $f$ is \emph{real} infinitely differentiable,
as witnessed by the well-known function
\begin{align}
    h(x) = \begin{cases}
        \exp(-\frac{1}{x^2}) & (x \neq 0)\\
        0 & (x = 0)
    \end{cases}.
\end{align}
Its Taylor series at the origin is identically zero, and hence converges everywhere on $\RR$
but the function is zero only at the origin.
Note that if we extend the domain of definition of $h$ from $\RR$ to $\CC$ by the same formula,
then $h$ is not complex differentiable at the origin;
$h$ is divergent at the origin along the imaginary axis.

\section{Series expansions of expectation values} \label{sec:cluster}

The main goal of this section is to prove the following theorem.
A direct consequence of this theorem is that the Taylor series expansion for expectation values of local operators converges when $2e^2(\degree + 1)^2\beta < 1$ (see the discussion above \cref{eq:tau}).
A reader who wishes to understand just our learning algorithm
may skip the rest of this section on their first reading, since
all notions and properties needed for our algorithm and its analysis
are contained in the statement of \cref{thm:expectation-value}.

\begin{theorem}\label{thm:expectation-value}
    Consider a Hamiltonian $\{(a,E_a,\lambda_a):a\in[\numTerms]\}$.
    Then, for every $a \in[\numTerms]$ we have a Taylor series expansion
    \begin{align}
    \frac{\Tr( E_a \exp(-\beta H))}{\Tr \exp(-\beta H)}
    =\frac{\Tr(E_a)}{\fulldimension} +
    \sum_{m = 1}^\infty \beta^mp_m(\lambda_1,\ldots,\lambda_M), \label{eq:mainseries}
    \end{align}
    where equality holds whenever the series converges absolutely.
    For any $m \in \ZZ_{>0}$, the following hold:
    \begin{enumerate}[ref=\thetheorem (\arabic*)]
        \item $p_m \in \RR[\lambda_1,\ldots,\lambda_\numTerms]$ is a degree $m$ homogeneous polynomial in the Hamiltonian term coefficients.\label{thm:expectation-value1}
        \item $p_m$ involves $\lambda_b$ only if the distance between $a$ and $b$ on $\graph$, $\dist_\graph(a,b)$, is at most $m$.\label{thm:expectation-value2}
        \item $p_m$ consists of at most $e\degree(1+e(\degree-1))^{m}$ monomials.\label{thm:expectation-value3}
        \item The coefficient in front of any monomial of $p_m$ is at most $(2e(\degree +1))^{m+1}(m+1)$ in magnitude.\label{thm:expectation-value4}
    \end{enumerate}
    Suppose further that every $E_a$ is a tensor product of Pauli matrices, supported on at most $L$ qubits.
    Then, after $O(L\numTerms\degree\log\degree)$ pre-processing time (see \cref{rm:ham-repr}), the following are true for every $m \in \ZZ_{> 0}$.
    \begin{enumerate}[label=\Alph*.,ref=\thetheorem (\Alph*)]
    \item The list of monomials that appear in $p_m$ can be enumerated 
    in time $O(m \degree C)$, where $C$ is the number of monomials (so, in particular, in time $O(m \degree^2 (1+e(\degree - 1))^m)$). \label{thm:expectation-valueA}
    \item The coefficient of any monomial in $p_m$ can be computed exactly as a rational number
    in $O(Lm^3 + 8^mm^5\log^2m) = (8^m + L)\poly(m)$ time. \label{thm:expectation-valueB}
    \end{enumerate}
\end{theorem}

\paragraph{Overview of the proof.}

The series expansion in \cref{eq:mainseries} is certainly conceivable.
When $\beta = 0$, the numerator is zero because $\Tr P = 0$.
The first order term in $\beta$ comes from the first order term in $\beta$ of the numerator
and the zeroth order term in the denominator.
One can keep finding terms order by order in $\beta$,
but this calculation soon becomes too complicated to be useful in a proof.

Following the wisdom of statistical mechanics,
as done explicitly in e.g.~\cite{kkb20,WA22},
we examine the logarithmic partition function $\log \Tr \exp(-\beta H)$,
and take differentials to arrive at \cref{eq:mainseries}.
This basic connection is established in \cref{sec:logZ}.

To put our series expansion on a rigorous foundation,
we will use complex differentiability (holomorphicity)
of the function $\beta \mapsto \log \Tr \exp(-\beta H)$.
Although we are only interested in the regime where $\beta$ is positive real,
it is quite useful for us to observe that 
this function behaves nicely in a sufficiently large domain
in the complex plane of $\beta$.
Armed with \cref{statement:analyticFunctionUnique},
we fearlessly write infinite Taylor series and resummations thereof,
to derive a series expansion of the logarithmic partition function
as a multivariate function of $\lambda_1,\ldots,\lambda_\numTerms$.
This leads to the concept of cluster expansion
\begin{align}
    \sum_{m \ge 0} \sum_{\cluster V} C(m,\cluster V),
\end{align}
which is an infinite sum of finite sums over clusters $\cluster V$ (see \cref{sec:multiTaylor}). In this section we use the convention that boldface uppercase letters refer to clusters.
Each $C(m,\cluster V)$ has a $\beta^m$ factor,
so in order for the infinite sum $\sum_{m \ge 0}$ to converge 
for some small enough $\beta$, 
each $\sum_{\cluster V} C(m,\cluster V)$ has to be at most exponentially large in~$m$.
This exponential bound will occupy us for most of the proof,
regarding which we follow many elements from~\cite{KS19} and~\cite{WA22}.
We count the number of summands of $\sum_{\cluster V}$ purely combinatorially
in~\cref{sec:countingClusters},
and, separately, bound the magnitude of $C(m,\cluster V)$ for each $\cluster V$.
The second step follows the approach in~\cite{WA22}.
The result will be that 
$\abs{\sum_{\cluster V} C(m,\cluster V) } \le \poly(\beta,\degree)O(\degree^2\beta)^{m}$.
In contrast, assertions in~\cite{KS19} imply $\abs{\sum_{\cluster V} C(m,\cluster V) } \le \poly(\beta,\degree) O(\degree \beta)^m$,
which is quadratically stronger in~$\degree$.
This difference is because we do not use~\cite[App.~D of arXiv-v2]{KS19},
in which the argument appears to have a mathematical gap in~\cite[(D.10) of arXiv-v2]{KS19}.
Recently, another mathematical gap in~\cite[App. C of arXiv-v2]{KS19} was pointed out  by~\cite{WA22}.

Finally, we take some care to give an algorithm to compute the Taylor series.
This is often elided as it is fairly standard, but for completeness, we show how to do this with symbolic computation to get the coefficients with no error.
First, we are able to enumerate the list of summands in the series in \cref{sec:enumeratingClusters}.
This is a breadth-first search, with some care to avoid duplicating work.
Second, we compute the coefficient for each summand in \cref{sec:computingClusterDerivatives}.
These coefficients are derivatives of the log-partition function $\log \Tr \exp(-\beta H)$ at the origin.
With some simple observations (truncating Taylor series and using the definition of a derivative), we conclude that these derivatives are equal to the trace of a matrix polynomial in the terms of the Hamiltonian at zero~\cref{eq:reduced-to-polynomial}.
Since we assume that these are tensor products of Pauli matrices, these polynomials can be evaluated efficiently, where the final runtime is exponential in the order of the derivative, as one would expect.

This gives the formal guarantees that this computation is bounded in complexity.
Practically, one can use any method at hand, not necessarily relying on the specific algorithm we propose.
Indeed, there is a large body of classic literature
on high temperature expansions (see e.g.\ the book~\cite{Domb})
which does not always discuss formal convergence.

\subsection{The logarithmic partition function}\label{sec:logZ}

Given a Hamiltonian $\{(a, E_a, \lambda_a) :a\in[\numTerms]\}$,
our primary object of study is
\begin{align}
    \betaF = \log \Tr \exp (-\beta H) = \log \Tr \exp\biggl( -\beta \sum_{a \in [\numTerms]} \lambda_a E_a \biggr).
\end{align}
The argument of the log function is called the \emph{partition function}
in statistical mechanics.
Hence, we will refer to this expression $\betaF$ as the \emph{logarithmic partition function} or \emph{log-partition function} of the Hamiltonian.
The quantity $\beta^{-1} \betaF$ is called (Helmholtz) free energy in statistical mechanics, but we will not use this terminology.

The connection of the logarithmic partition function to \cref{thm:expectation-value}
is given by the following.
\begin{prop}\label{statement:expectation-as-betaF-derivative}
    For any Hamiltonian $\{(a, E_a, \lambda_a) :a\in[\numTerms]\}$, $a \in [\numTerms]$, and nonzero $\beta \in \CC$,
    \begin{align}
        \frac{\Tr( E_a \exp(-\beta H))}{\Tr \exp(-\beta H)}
        = -\frac 1 \beta \frac{\partial}{\partial \lambda_a} \log \Tr \exp( -\beta H ).
    \end{align}
\end{prop}
\begin{proof}
    Since $H$ and $E_a$ have finite norm,
    the Taylor expansion of $\Tr \exp$ converges absolutely.
    The claim is proved by
    \begin{align}
    \frac{\diff}{\diff \lambda_a} \Tr \exp( -\beta H )
    &= \sum_{m=0}^\infty \frac{1}{m!} \Tr\left[ \frac{\diff}{\diff \lambda_a} (-\beta H)^m \right] \nonumber\\ 
    &= \sum_{m=1}^\infty \frac{1}{m!} \sum_{k=1}^m \Tr[ (-\beta H)^{k-1}(-\beta E_a) (-\beta H)^{m-k} ] \nonumber\\
    &= -\beta \sum_{m=1}^\infty \frac{1}{m!} \sum_{k=1}^m \Tr[ E_a (-\beta H)^{m-1} ] \nonumber\\
    &= - \beta \Tr[ E_a \exp(- \beta H)],
    \end{align}
using linearity and the cyclic property of $\Tr$.
Finally,
\begin{equation}
    \frac{\partial}{\partial \lambda_a} \log \Tr \exp( -\beta H ) = \frac{1}{\Tr \exp( -\beta H )} \frac{\partial}{\partial \lambda_a} \Tr \exp( -\beta H ) = -\beta \frac{\Tr( E_a \exp(-\beta H))}{\Tr \exp(-\beta H)},
\end{equation}
which completes the proof.
\end{proof}

If we understood the series expansion of the logarithmic partition function well enough, we could prove \cref{thm:expectation-value} easily by way of \cref{statement:expectation-as-betaF-derivative}.
There will be an important advantage (\cref{statement:log-derivative-connected} below)
in considering the logarithmic partition function,
rather than the ratio of two traces as in \cref{thm:expectation-value}.
So, we will study the series expansion of the logarithmic partition function.

\subsection{Deriving multivariate Taylor series expansions}\label{sec:multiTaylor}

In this section, we prove that a series expansion for the logarithmic partition function like the one in \cref{eq:mainseries} converges in some open neighborhood around the origin.

Though the logarithmic partition function is a complex-valued function of $\beta,\lambda_1,\ldots,\lambda_\numTerms$, depending on context, we'll think of it as either a function of a single variable $\beta$ for a fixed choice of $\lambda_a$'s,
or a function of $(\lambda_1,\ldots,\lambda_\numTerms)$ for a fixed $\beta$.

Let us first take the first perspective: Fix%
\footnote{We restrict $\lambda_j$ to the open interval $(-1,1)$ to avoid the inconvenience of discussing derivatives at the boundary of the domain of $\betaF$.}
$\lambda_a \in (-1,1)$
 for all $a \in [\numTerms]$ and consider the map $\beta \mapsto \betaF$.
By our convention that Hamiltonian terms $E_a$ are Hermitian with $\norm{E_a}\leq 1$, the spectrum of the Hamiltonian operator $H = \sum_a \lambda_a E_a$ is contained in the real interval $(-M, M)$, and further, for any $\beta \in \CC$, the spectrum of the operator $-\beta H$
is contained in the complex disk $\{z \mid z \in \CC,\,\abs{z} \leq \abs{\beta} \numTerms\}$.
Hence, for $\beta \in E$ where
\begin{align}
E = \left\{ x + iy ~\middle|~x,y \in \RR,\, \abs{y} < \frac{\pi}{2\numTerms} \right\},
\end{align}
if $h \in (-\numTerms,\numTerms)$ is an eigenvalue of $H$, then the complex number $e^{-\beta h} = e^{-x h} e^{-i y h}$ has positive real part.
Therefore, $\Tr \exp(-\beta H)$ is in the right half-plane of the complex plane
and the function $\beta \mapsto \log \Tr \exp(-\beta H)$ is complex differentiable on $E$,
using that $\log$ is complex differentiable on the right half-plane and that complex differentiability is closed under composition.
So by \Cref{statement:analyticFunctionUnique}, we are guaranteed that the function $\beta \mapsto \log \Tr \exp(-\beta H)$ has
a Taylor series representation in some open neighborhood of the origin in the complex plane of $\beta$,
although we do not yet know how large the open neighborhood can be.

Note that the same argument shows that the multivariate function
\begin{align}
    \CC^{\numTerms+1} \ni (\beta,\lambda_1,\ldots,\lambda_\numTerms) \mapsto \log\Tr\exp(-\beta H)
\end{align}
is complex differentiable in each variable on
\begin{align}
    \tilde E = \left\{(\beta,\lambda_1,\ldots,\lambda_\numTerms) \in \CC^{\numTerms+1} ~\Big\vert~ 
    |\Imag(\beta\lambda_a)| < \frac{\pi}{2M} \right\}.
\end{align}
Here, for any $z \in \CC$, $\Imag(z)$ denotes the imaginary part of $z$.
This set $\tilde E$ is an open neighborhood of the real line of $\beta$ times the real box $(-1,1)^M$ of $\lambda_a$'s.

The Taylor expansion at the origin is straightforward to write 
as $\sum_{m \ge 0} \frac{1}{m!}\beta^m (\diff_\beta^m \betaF|_{\beta=0})$, 
but this is not enlightening.
Let us make some observations first.
We interpret the logarithmic partition function as a function of $z = (z_1,\ldots,z_\numTerms) \in \CC^\numTerms$:
\begin{align}
    \betaF = \log \Tr \exp(-\sum_a z_a E_a) \text{ where } z_a = \beta\lambda_a.
\end{align}
It follows that
\begin{equation}
    \frac{\diff \betaF}{\diff \beta} 
    = \sum_a \frac{\diff z_a}{\diff \beta} \frac{\diff \betaF}{\diff z_a}
    = \sum_a \lambda_a \frac{\diff \betaF}{\diff z_a},
\end{equation}
so
\begin{align}
    \betaF 
    &= \sum_{m \ge 0} \frac{\beta^m}{m!}\Biggl(\frac{\diff^m \betaF}{\diff \beta^m}\biggr|_{\beta = 0}\Biggr) \nonumber\\
    &= \sum_{m \ge 0} \frac{\beta^m}{m!} \sum_{a_1,a_2,\ldots,a_m} \lambda_{a_1} \cdots \lambda_{a_m} 
        \Biggl(\frac{\diff^m \betaF}{\diff z_{a_1} \cdots \diff z_{a_m}}\biggr|_{z = (0,\ldots,0)}\Biggr). \label{eq:betaf-ordered-series}
\end{align}
Since $\betaF$ is complex differentiable in any variable (at least on $\tilde{E}$),
it is infinitely differentiable, and hence any two differentiations commute.
So, instead of summing over ordered tuples $(a_1,\ldots,a_m)$, we can sum over \emph{multisets}, which are unordered tuples.
This particular class of multisets will be used frequently, so let's give a proper definition.

\begin{definition} \label{def:cluster}
    A \emph{cluster} $\cluster V$ is a set of tuples $\{(a,\mu(a)) \mid a \in [\numTerms]\}$ where the function $\mu: [\numTerms] \to \ZZ_{\geq 0}$ maps $a$ to the multiplicity of $a$.
    The total weight, denoted $\abs{\cluster V}$, of $\cluster V$ is $\sum_a \mu(a)$. 
    We will write $a \in \cluster V$ if $\mu(a)$ is nonzero,
    and the \emph{support} of $\cluster V$ is defined to be $\Supp \cluster V = \{ a \in [\numTerms] : \mu(a) \ge 1 \}$.
    We also introduce a combinatorial factor $\cluster V !$ to mean $\prod_a \mu(a)!$.
\end{definition}
\noindent
One may think of a cluster as a function $a \mapsto \mu(a)$
or a monomial in indeterminates $\lambda_1,\ldots,\lambda_\numTerms$.
Returning to \cref{eq:betaf-ordered-series}, we have
\begin{align}
    \betaF 
    &= \sum_{m \ge 0} \frac{\beta^m}{m!} \sum_{a_1,a_2,\ldots,a_m} \lambda_{a_1} \cdots \lambda_{a_m} 
        \Biggl(\frac{\diff^m \betaF}{\diff z_{a_1} \cdots \diff z_{a_m}}\biggr|_{z = (0,\ldots,0)}\Biggr) \nonumber\\
    &= \sum_{m \ge 0}  \beta^m \sum_{\cluster V : \abs{\cluster V} = m} \frac{1}{\cluster V !} \prod_{a \in \Supp \cluster V}\lambda_a^{\mu(a)} \left( \prod_{a \in \Supp \cluster V} \frac{\diff^{\mu(a)}}{\diff z_{a}^{\mu(a)}} \right) \betaF \biggr|_{z = (0,\ldots,0)} \nonumber\\
    &= \sum_{m \ge 0}   \sum_{\cluster V : \abs{\cluster V} = m}
    \frac{1}{\cluster V !} \underbrace{\prod_{a \in \Supp \cluster V}\lambda_a^{\mu(a)}}_{\lambda^{\cluster V}} \underbrace{\left( \prod_{a \in \Supp  \cluster V} \frac{\diff^{\mu(a)}}{\diff \lambda_{a}^{\mu(a)}} \right)\biggr|_{\lambda = (0,\ldots,0)}}_{\mdiff_{\cluster V}} \betaF \nonumber\\
    &= \sum_{m \ge 0}  \sum_{\cluster V : \abs{\cluster V} = m} \frac{\lambda^{\cluster V}}{\cluster V !} \mdiff_{\cluster V} \betaF. \label{eq:multivariateTaylor}
\end{align}
Note that we have introduced the cluster notations
\begin{align}
    \lambda^{\cluster V} = \prod_{a \in \Supp \cluster V}\lambda_a^{\mu(a)} \text{ and }
    \mdiff_{\cluster V} = \prod_{a \in \Supp \cluster V} \frac{\diff^{\mu(a)}}{\diff \lambda_{a}^{\mu(a)}} \biggr|_{\lambda_1 = \cdots = \lambda_\numTerms = 0} .
    \label{eq:mdiff-def}
\end{align}
\cref{eq:multivariateTaylor} is the series expansion of $\betaF$ that we are going to investigate.
It is nothing but the Taylor expansion of $\betaF$, treating it as
a multivariate function \mbox{$(\lambda_1,\ldots,\lambda_\numTerms) \mapsto \betaF$}.
Though we could have guessed this expansion from the outset, this derivation is necessary to show that the series converges: we start with a $\beta$-series whose validity is guaranteed,
albeit on an unspecified small domain,
by the complex differentiability with respect to $\beta$.

The number of all clusters of weight $m$ is at least $\binom{\numTerms}{m}$.
This is much larger than what we claim in \cref{thm:expectation-value3},
where the bounds are independent of $\numTerms$.
The special structure of the logarithmic partition function will help us show the improved bound.

\subsection{Counting connected clusters}\label{sec:countingClusters}

The main point of considering the logarithmic partition function
is that a cluster has nonzero coefficient in \cref{eq:multivariateTaylor} only if it is \emph{connected}.

\begin{definition}\label{defn:ConnectedCluster}
    A cluster $\cluster W = \{ (a,\mu(a)) \}$ is \emph{connected} 
    if the subgraph of $\graph$ induced by the support of $\cluster W$
    is connected.
\end{definition}

\begin{prop}\label{statement:log-derivative-connected}
    Define $\calZ = \frac 1 \fulldimension \Tr \exp( - \beta H)$.
    If $\cluster W'$ and $\cluster W''$ are both nonempty such that no edge of~$\graph$ connects $\Supp \cluster W'$ and $\Supp \cluster W''$,
    then $\mdiff_{\cluster W' \cup \cluster W''} \calZ = (\mdiff_{\cluster W'} \calZ)(\mdiff_{\cluster W''} \calZ)$.
    In particular, if a cluster $\cluster W$ is not connected, then $\mdiff_{\cluster W} \betaF = 0$.
\end{prop}
\begin{proof}
    The two operators $\sum_{a \in \Supp \cluster W'} \lambda_a E_a$ and $\sum_{b \in \Supp \cluster W''} \lambda_b E_b$ 
    commute with each other since the operators' supports do not overlap.
    Define $\calZ|_{\cluster W} = \frac 1 \fulldimension \Tr \exp(-\beta \sum_{a \in \Supp \cluster W} \lambda_a E_a )$
    and similarly $\betaF|_{\cluster W} = \log (\fulldimension \calZ|_{\cluster W})$.
    Then, $\calZ|_{\cluster W} = \calZ|_{\cluster W'} \calZ|_{\cluster W''}$ and
    $\betaF|_{\cluster W} = \betaF|_{\cluster W'} + \betaF|_{\cluster W''} - \log \fulldimension$.
    The first claim immediately follows.
    Since $\mdiff_{\cluster W}$ evaluates the derivative at the origin of the Hamiltonian coefficient space,
    we see
    \begin{align}
    \mdiff_{\cluster W} \betaF 
    &= 
    \mdiff_{\cluster W} (\betaF|_{\cluster W})
    =
    \prod_{a \in \Supp \cluster W'} \frac{\diff^{\mu(a)}}{\diff \lambda_{a}^{\mu(a)}} \prod_{b \in \Supp \cluster W''} \frac{\diff^{\mu(b)}}{\diff \lambda_{b}^{\mu(b)}}
    (\betaF|_{\cluster W'} + \betaF|_{\cluster W''} - \log \fulldimension)\biggr|_{\lambda = (0,\ldots,0)}
    = 0. 
    \end{align}
    In other words, $\betaF$ decomposes into a sum of $\betaF|_{\cluster W'}$ and $\betaF|_{\cluster W''}$, which are each a function of a strict subset of the variables $\lambda_a$ that appear in $W$, but they are each being differentiated with respect to all $\lambda_a$ that appear in $W$, and hence must evaluate to zero.
\end{proof}

\noindent Now we bound the number of all connected clusters of a given total weight $w$.
What matters most for us is that this bound is just exponential in the total weight, $\exp(O(w))$ instead of, say, the more naive bound $O(w!)$.
We optimize the base of the exponent in our bound, since this affects the eventual algorithm's runtime.

\begin{prop}\label{statement:count-clusters}
    Let $\graph$ be any graph with maximum degree $\degree \geq 2$.
    Given any node $a$ of $\graph$ and any weight $w \in \ZZ_{>0}$,
    the number of all connected clusters $\cluster W$ such that 
    $a \in \cluster W$ and $\abs{\cluster W} = w$
    is at most $e\degree(1+e(\degree-1))^{w-1}$.
    If $\degree =1$, then we have an upper bound of $w$.
\end{prop}

\begin{proof}
    For a fixed degree $\degree$, the number of clusters is maximized when $\graph$ is an infinite $\degree$-regular tree.
    Since this graph is self-similar, without loss of generality we can think of $a \in \graph$ as being the root of the tree.
    So, this question reduces to upper-bounding the number of connected rooted subtrees of the infinite $\degree$-regular tree, where nodes of the subtrees are allowed to have multiplicity.
    If every node in the tree must have multiplicity one (that is, if we disallow multiplicity), we have the following.
    \begin{lemma}\label{statement:counting-subtrees}
        For $n \in \ZZ_{\geq 0}$, let $D_n$ be the number of all connected rooted subtrees with $n$ nodes in the infinite $\degree$-regular tree.
        Then
        \begin{align}
            D_n = \binom{n(\degree-1)+1}{n-1}\frac{\degree}{n(\degree-1) + 1} \le e\degree (e(\degree -1))^{n-1}.
        \end{align}
    \end{lemma}    
    To count subtrees \emph{with} multiplicity,
    we must count the number of ways to assign a positive integer to every node of a subtree.
    If the subtree has $k$ nodes,
    there are $\binom{(w-k)+(k-1)}{k-1}$ ways to assign multiplicities to these nodes such that the multiplicities sum to $w$.
    Hence, the number of weight-$w$ connected rooted clusters of the infinite $\degree$-regular tree is
    \begin{align}
           \sum_{k = 1}^w D_k \binom{w-1}{k-1}
        &\leq \sum_{k = 1}^w e \degree (e(\degree-1))^{k-1}\binom{w-1}{k-1} \\
        &= e\degree (1+e(\degree-1))^{w-1}.\qedhere
    \end{align}
\end{proof}

\begin{proof}[Proof of \cref{statement:counting-subtrees}]    
    This can be done with standard manipulations of generating functions, which we detail below.
    As a reminder, if we have a sequence $\{a_0, a_1, a_2, \ldots\}$ of integers, 
    then the generating function corresponding to it is $A(z) = \sum_{i\geq 0} a_i z^i$.
    We use the notation $[z^i]A(z)$ to refer to the coefficient $a_i$ of $A(z)$.

    Let $E(z)$ be the generating function counting subtrees of the infinite $(\degree-1)$-ary tree, 
    the tree where every node has $\degree-1$ children.%
    \footnote{
            For combinatorialists, this is also the generating function for the Fuss-Catalan numbers: 
            $E(z) + 1 = \sum_{n \geq 0} \frac{1}{n(\degree-1) + 1}\binom{n(\degree-1) + 1}{n}z^n$ \cite[7.5 Example~5]{GKP94}.
        }
    Then the following recursion holds.
    \begin{align}
        [z^n](E(z)) = \sum_{\substack{(i_1,\ldots, i_{\degree-1}) \in \ZZ_{\geq 0}^{\degree - 1} \\ i_1 + \cdots + i_{\degree - 1} = n-1}} \prod_{j=1}^{\degree - 1} [z^{i_j}](E(z)) \label{eq:e-gf}
    \end{align}
    In words, this describes an $n$-node rooted subtree of the infinite $(\degree-1)$-ary tree 
    as $\degree-1$ (possibly empty) subtrees corresponding to each child of the root, 
    where the number of nodes in each subtree sum to $n-1$.
    These subtrees are also rooted subtrees of the infinite $(\degree-1)$-ary tree, allowing the expression above to recurse as stated.
    One can verify that \cref{eq:e-gf} is equivalent to the equation
    \begin{align}
        E(z) = z(1 + E(z))^{\degree-1}.
    \end{align}
    
    Let $D(z)$ be the generating function counting subtrees of the infinite $\degree$-regular tree, 
    the tree where every node has $\degree$ neighbors.
    In particular, it only differs from the $(\degree-1)$-ary tree in the root node, 
    where there are $\degree$ instead of $\degree - 1$ many options.
    Using a similar argument as with $E(z)$, we can conclude that
    \begin{align}
            D(z) = z(1 + E(z))^{\degree} = E(z)(1+E(z)).
    \end{align}
    As an aside, the number of rooted clusters on the $\degree$-regular infinite tree, $C_n = \sum_{k \ge 1} \binom{n-1}{k-1} D_k$, corresponds to the generating function definition $C(z) = D(\frac{z}{1-z})$.
    Since $\binom{n-1}{k-1} = [z^{n}]( (z+z^2+\cdots)^{k} )$, we have
    \begin{align}
        [z^n](C(z)) &= \sum_{k \geq 0} \binom{n-1}{k-1} [z^k]D(z)
        = [z^n]D(z + z^2 + \cdots) = [z^n]D(\tfrac{z}{1-z}).
    \end{align}
    Returning to the proof, we use the Lagrange--Bürmann formula to get the series expansion of $D(z)$ from the inverse of $E(z)$.
    In particular, we use the formulation common in combinatorics~\cite[Thm A.2 (14)]{FS09},
    \begin{align}
        [z^n]H(y(z)) = \frac{1}{n}[u^{n-1}](H'(u)\phi(u)^n),
    \end{align}
    where $H$ is an arbitrary function and $y(z) = z\phi(y(z))$.
    We set $H(z) = z(1+z)$, $y(z) = E(z)$, and $\phi(z) = (1+z)^{\degree -1}$.
    So
    \begin{align}\label{eq:d-gf}
        [z^n]D(z) &= \frac{1}{n}[u^{n-1}]((2u+1)(1+u)^{n(\degree-1)}) \\
        &= \frac{1}{n}\Biggl(2\binom{n(\degree-1)}{n-2} + \binom{n(\degree-1)}{n-1}\Biggr) \nonumber\\
        &= \binom{n(\degree-1)+1}{n-1}\frac{\degree}{n(\degree-1) + 1}.\nonumber
    \end{align}
    \cref{eq:d-gf} is the desired equality in the lemma statement.
    As for the inequality, clearly, $[z^n]D(z) \leq e\degree(e(\degree-1))^{n-1}$ for $n =0,1$.
    For $n \ge 2$, we see
    \begin{align}
            \binom{n(\degree-1)+1}{n-1}\frac{\degree}{n(\degree-1) + 1}
            &\le
            \degree\binom{n(\degree-1)}{n-1} \\
            &\le
            \degree\Big(\frac{en(\degree -1)}{n-1}\Big)^{n-1} \nonumber\\
            &=
            \degree(e(\degree-1))^{n-1} \big( \frac{n}{n-1} \big)^{n-1}\nonumber\\
            &\le 
            e\degree(e(\degree-1))^{n-1}. \qedhere
    \end{align}
\end{proof}

\subsection{Enumerating connected clusters} \label{sec:enumeratingClusters}

Enumerating clusters of total weight $m$ at a node $a \in \graph$ can be done with the following algorithm.
First, perform a breadth-first search starting at $a$ to produce the disjoint sets $V_i$ for \mbox{$i \in \{0,1,\ldots,m-1\}$}, where $V_i$ is the set of nodes exactly $i$ away from $a$ in graph distance.
Note that this induces a directed tree $G = (V,E)$ with vertices $V = V_0 \sqcup \cdots \sqcup V_{m-1}$ and a directed edge $(u,v)$ occurring if it is an edge in $\graph$ and $v$ is at a lower level than $u$ (so $u \in V_i$ and $v \in V_{i+1}$).
Since $G$ is directed, all the neighbors of $u \in V_i$ are in $V_{i+1}$.
For $S \subset V$, we denote $\Gamma(S)$ to be the neighborhood of $S$ in $G$.
If~$\cluster S$ is a multiset, 
$\Gamma(\cluster S)$ is defined to be the $G$-neighborhood of the support, $\Gamma(\Supp \cluster S)$.

Every cluster can be represented uniquely as a collection of multisets $\cluster S_i$ of $V_i$ for $i \in \{0,1,\ldots,m-1\}$ 
satisfying that every node in $\cluster S_{i+1}$ has a parent in $\cluster S_{i}$ (or equivalently, satisfying that $\Supp \cluster{S}_0 \cup \cdots \cup \Supp \cluster{S}_{m-1}$ is connected in $\graph$).

Because of this characterization, we can enumerate clusters through a recursive function that, given the first $i$ layers of a cluster $\cluster S_0, \ldots, \cluster S_{i-1}$, outputs a list of all possible ways to complete the cluster $\cluster S_i,\ldots,\cluster S_{m-1}$ such that the total weight of the cluster is $m$.
We describe this function in \cref{alg:cluster-enumeration}; it only requires three parameters, the recursion level $i$, the remaining weight $m_i$, and the neighborhood $C_i$, which correspond to $i$, $m - (\abs{\cluster{S}_0} + \cdots + \abs{\cluster{S}_{i-1}})$, and $\Gamma(\cluster{S}_{i-1})$ in the above description.
To find all the clusters of weight $m$, run $\operatorname{tails}(0, m, \{a\})$.
The function proceeds as follows: at recursion level $i \ge 0$, we loop over all possible nonempty multisets $\cluster S_i$ of $C_i \subset V_i$ of weight $\leq m_i$.
For each such multiset $\cluster S_i$ of $C_i$, we call the recursive function to enumerate all of its possible continuations,
with parameters $i+1$, $m_{i+1} = m_i - \abs{\cluster S_i}$, and $C_{i+1} = \Gamma(\cluster S_i)$.
Once it returns the possible continuations of this cluster, we add $\cluster S_i$ to every continuation (to make them continuations of $\cluster S_{i-1}$).
Upon enumerating all possible $\cluster S_i$'s, return the resulting (now complete) list of continuations of $\cluster S_{i-1}$ as output.

\begin{algorithm} 
    \caption{$\operatorname{tails}(i, m_i, C_i)$ [Cluster enumeration recursion]}\label{alg:cluster-enumeration}
    \KwData{depth $i$, size $m_i$, base $C_i \subseteq V_i$}
    \KwResult{$\mathcal C_i$, the collection of all connected multisets of size $m_i$ supported on $C_i \sqcup V_{i+1} \sqcup \cdots \sqcup V_{m-1}$}
    \If{$m_i = 0$}{
        Return $\{\varnothing\}$\;
    }
    Let $\operatorname{out} \leftarrow \varnothing$\;
    \For{$\cluster S \in \multisetnumber{C_i}{1} \cup \cdots \cup \multisetnumber{C_i}{m_i}$, 
    		where $\multisetnumber{C_i}{j}$ denotes the set of all multisets of weight $j$ supported on $C_i$ 
    }{
        Recurse $\operatorname{continuations}(\cluster S) \leftarrow \operatorname{tails}(i+1,~ m_i - \abs{\cluster{S}}, ~ \Gamma(\cluster S))$\;
        Append $\operatorname{out} \leftarrow \operatorname{out} \cup \{\cluster S \cup \cluster T \mid \cluster T \in \operatorname{continuations}(\cluster S) \}$\;
    }
    Return $\operatorname{out}$\;
\end{algorithm}

Given the dual interaction graph as a random-access dictionary, (i.e., one can query an arbitrary node and receives its neighbors), the runtime of this algorithm is $O(m \degree \mathcal{C})$, where $\mathcal{C}$ is the number of clusters output.
The main cost is computing $\Gamma(\cluster S)$ from a given multiset $\cluster S$ on $V_{i}$:
this takes time $\sum_{v \in \Supp \cluster S} \abs{\Gamma(v)} \leq \degree \abs{\cluster S}$,
where the inequality uses that $G$ is degree $\leq \degree$.
For every cluster $\cluster{S}_0 \sqcup \cdots \sqcup \cluster{S}_{m-1}$, this computation occurs once for each of the $\cluster{S}_i$'s.
Since every cluster has a total weight of $m$, this gives an upper bound of $O(m\degree \mathcal{C})$ for all such computations.

\subsection{Estimating cluster derivatives}\label{sec:clusterDerivatives}

The goal of this subsection is to prove the following bound on cluster derivatives.
\begin{prop}\label{statement:cluster-deriv-bound}
    Consider a Hamiltonian $\{(a, E_a, \lambda_a)\}$.
    Let $\cluster W = \{(a, \mu(a))\}$ be a cluster of the associated dual interaction graph $\graph$ with total weight $m+1 \ge 1$. 
    Then
    \begin{align}
    \abs[\Big]{\frac{1}{\cluster W !} \mdiff_{\cluster W} \betaF}
    \le
    (2 e (\degree+1) \beta)^{m+1} .
    \end{align}
\end{prop}

\noindent 
To prove this, we will bound $\abs{\mdiff_{\cluster w} \betaF}$ by a quantity 
that depends on a simple graph constructed from~$\cluster W$ and the dual interaction graph~$\graph$,
which we now define.

Recall from \cref{def:cluster} that given a set~$S$, a multiset~$S^\mu$ of elements of~$S$ 
is a set $\{(s,\mu(s))~|~ s \in S\}$ where $\mu(s)\in \ZZ_{\ge 0}$ is the multiplicity of~$s$.
We also write $S^\mu = \{ (s_i,\mu(s_i)) \} = \{\{ s_1,\ldots, s_2,\ldots\}\}$ where $s_i$ is repeated exactly $\mu(s_i)$ times.
The size of~$S^\mu$ is $\abs{S^\mu} = \sum_{s \in S} \mu(s)$.
The support of~$S^\mu$ is $\Supp S^\mu = \{ s \in S | \mu(s) > 0 \}$.
We write $S^\mu !$ to mean $\prod_{s \in S} (\mu(s)!)$.
If $S$ is the node set of a simple graph~$F$,
we define a simple graph~$\gr( S^\mu )$ as follows.
The set of nodes are
\begin{equation}
    \marked(S^\mu) := \{(s,i) \in (\Supp S^\mu) \times \ZZ_{>0} |  1 \le i \le \mu(s)\};
\end{equation}
in other words, there are exactly $\mu(s)$ nodes corresponding to~$s$ for each~$s \in S$,
so there are $\abs{S^\mu}$ nodes in total.
In~$\gr(S^\mu)$, an edge between $(s,i)$ and $(s',i')$ exists 
if and only if either $s = s'$ or $(s,s')$ is an edge of~$F$.
In particular, the induced subgraph of~$\{(s,i)\}_i$ for any given~$s$ is a clique.

Since the dual interaction graph~$\graph$ serves as the underlying graph for Hamiltonian terms,
for any cluster~$\cluster W$ of Hamiltonian terms we have a corresponding $\gr(\cluster W)$.
In $\gr(\cluster W)$ there are $\abs{\cluster W}$ nodes in total,
and an edge between two nodes~$(a,i), (a',i')$ exists iff 
the Hamiltonian terms~$E_a$ and $E_{a'}$ have overlapping supports.
Note that if all multiplicities of~$\cluster W$ are either~$0$ or~$1$,
then $\gr(\cluster W)$ is an induced subgraph of~$\graph$,
but is not otherwise.
Most of this section will be devoted to proving the following lemma, which implies \cref{statement:cluster-deriv-bound}.

\begin{lemma}[\cite{WA22}] \label{statement:cluster-derivative-bound-no-multiplicity}
Denote by $\deg(v)$ the number of neighbors of any node~$v$ in~$\gr(\cluster W)$.
Then,
\begin{align}
    \abs[\Big]{\mdiff_{\cluster W} \betaF} \le \abs{\beta}^{\abs{\cluster W}}
    \prod_{v \in \marked \cluster W} (2\deg(v)).
\end{align}
\end{lemma}

\begin{proof}[Proof of \cref{statement:cluster-deriv-bound}]
    It follows from the definition of~$\gr(\cluster W)$ that
    \begin{align} \label{eq:degree-count}
        \deg((b,i))= \mu(b) - 1 + \sum_{a \in \Gamma(b)} \mu(a)
    \end{align}
    for any~$b \in \Supp \cluster W$,
    where $\Gamma(b)$ is the set of all neighbors of~$b$ in~$\graph$ that appear in~$\cluster W$.
    Further note that
    \begin{multline} \label{eq:degree-sum-count}
        \sum_{b\in \Supp \cluster W} \deg((b,1))
        =\sum_{b \in \Supp \cluster W} \Big(\mu(b) - 1 + \sum_{a \in \Gamma(b)} \mu(a) \Big)
        \le m+1 + \sum_{b \in \Supp \cluster W} \sum_{a \in \Gamma(b)} \mu(a)
        \le m + \degree (m+1).
    \end{multline}
    because $\mu(a)$ appears in $\sum_b \sum_{a\in \Gamma(b)}$ at most $\degree$ times.
    We can now apply \cref{statement:cluster-derivative-bound-no-multiplicity}.
    \begin{align}
        \frac 1 {\cluster W !} \abs{\mdiff_{\cluster W} \betaF}
        &\le
        \frac {(2\beta)^{m+1}} {\cluster W !} \prod_{b\in \Supp \cluster W} \prod_{i=1}^{\mu(b)} \deg((b,i)) &\text{by \Cref{statement:cluster-derivative-bound-no-multiplicity}} \nonumber\\
        &=
        (2\beta)^{m+1} 
            \prod_{b\in \Supp \cluster W} \frac{1}{\mu(b)!}
            \Bigl( \mu(b) - 1 + \sum_{a \in \Gamma(b)} \mu(a) \Bigr)^{\mu(b)} &\text{by \Cref{eq:degree-count}}
            \nonumber\\
        &\le
        (2e\beta )^{m+1}\prod_{b\in \Supp \cluster W} \Biggl( \frac{\mu(b) - 1 + \sum_{a \in \Gamma(b)} \mu(a)}{\mu(b)} \Biggr)^{\mu(b)} &\text{because }u! \ge u^u e^{-u} \nonumber\\
        &\le
        (2e\beta)^{m+1} \left( \frac{(1+\degree) (m+1)}{m+1} \right)^{m+1}. &
    \end{align}
    The last inequality uses \cref{statement:munu} below and \Cref{eq:degree-sum-count}.
\end{proof}

\begin{lemma}\label{statement:munu}
    Let $\mu_1,\ldots,\mu_n > 0$ be real numbers, and
    $y_1,\ldots,y_n \ge 0$ be real numbers.
    Then
    \begin{align}
        \left(\frac{y_1}{\mu_1}\right)^{\mu_1}
        \cdots
        \left(\frac{y_n}{\mu_n}\right)^{\mu_n}
        \le
        \left(\frac{y_1 + \cdots + y_n}{\mu_1 + \cdots + \mu_n} \right)^{\mu_1 + \cdots + \mu_n}  \label{eq:munu},
    \end{align}
    where the equality holds when $y_j / \mu_j = (\sum_i y_i) / (\sum_i \mu_i)$ for all $j$.
\end{lemma}
\begin{proof}
    If any of $y_i$ is zero, the inequality is trivial.
    Assume $y _i > 0$ for all $i$.
    Taking log of both sides and dividing by $\sum_i \mu_i$, we have
    \begin{align}
        \sum_{i=1}^n \frac{\mu_i}{\sum_j \mu_j} \log\left(\frac{y_i}{\mu_i}\right)
        &\le \log\left(\frac{y_1 + \cdots + y_n}{\mu_1 + \cdots + \mu_n} \right).
    \end{align}
    This is Jensen's inequality applied to a concave function $\log$.
\end{proof}

\subsubsection*{Proof of \cref{statement:cluster-derivative-bound-no-multiplicity}.}

We will adopt the approach in~\cite{WA22}, making some short-cuts.\footnote{
    We thank the authors of~\cite{WA22} for pointing out 
    a problem in our earlier version of this proof,
    which traces back to~\cite[App.~C]{KS19}.
}
Our combinatorics will be self-contained; 
we do not assume any prior knowledge of Tutte polynomials or chromatic polynomials,
which were used in~\cite{WA22}.
However, in essence, the proof here is due to~\cite{WA22}.

We will use multisets of clusters.
All the general remarks above on multisets continue to apply.
Consider a multiset of clusters~$\cluster W_i$, which we denote $\partition P = \{\{ \cluster W_1, \ldots \}\} = \{ (\cluster W, \mu(\cluster W) ) \}$.
We write $\abs{\partition P} = \sum_{\cluster W} \mu(\cluster W)$,
$\Supp \partition P = \{ \cluster W ~|~ \mu(\cluster W) > 0 \}$,
and $\partition P! = \prod_{\cluster W \in \Supp \partition P} (\mu(\cluster W)!)$.
The set (not multiset) of all \emph{connected} clusters defines 
a simple graph where there is an edge between $\cluster W$ and $\cluster W'$ if and only if
their multiset union $\cluster W \cup \cluster W'$ (obtained by summing the multiplicities) is a connected cluster.
With this, a multiset~$\partition P$ of connected clusters defines a simple graph~$\gr(\partition P)$:
there are $\abs{\partition P}$ nodes in total,
and an edge between two nodes corresponding to~$\cluster W$ and~$\cluster W'$ exists if and only if
$\cluster W \cup \cluster W'$ is connected.

\paragraph{Counting partitions of a cluster}

Consider a cluster $\cluster W$.
From an (unordered) partition $\{W_1, \ldots, W_p\}$ of the \emph{graph node set} $\marked(\cluster W)$, we can get an (unordered) partition of the \emph{cluster} $\cluster W$ simply by forgetting all of the labels on the terms in $\cluster W$.
By a ``partition'' of a cluster, we mean a multiset $\partition P = \{\{\cluster W_1, \ldots, \cluster W_p\}\}$ such that their multiset union is $\cluster W$.
Conversely, given a cluster partition $\partition P$, the number of graph partitions that get mapped to $\partition P$ by ``forgetting'' is
\begin{equation} \label{eq:cluster-to-graph-count}
    \frac{\cluster W!}{\partition P! \prod_{\ell=1}^p(\cluster W_\ell!)}.
\end{equation}
Here, $\cluster W!$ is the number of ways to assign labels to $\partition P$ if we give an arbitrary ordering to both the clusters and the terms within the clusters; $\prod_{\ell=1}^p(\cluster W_\ell!)$ addresses the overcounting from ordering each term $\cluster W_\ell$, since swapping labels within a cluster doesn't change the cluster; and $\partition P!$ addresses the overcounting from ordering $\partition P$, since swapping the labels across two identical clusters doesn't change the cluster partition.
For example, consider a Hamiltonian with two terms $E_a, E_b$ with overlapping support.
Then the cluster partition $\partition P = \{\{a,b\},\{a,b\},\{b,b\}\}$ of the connected cluster $\cluster W = \{(a, 2), (b, 4)\}$ corresponds to 12 graph partitions, using the reasoning above.
\begin{align*}
    &(((a,f), (b,h)), ((a,g), (b,i)), ((b,j), (b,k)))
    &\{f,g\} = [2], \{h,i,j,k\} = [4],\, 2!4! = 48& \text{ choices} \\
    \mapsto&(\{(a,f), (b,h)\}, \{(a,g), (b,i)\}, \{(b,j), (b,k)\})
    &\text{overcounts by a } (1!1!)(1!1!)(2!) \text{ factor, } 24& \text{ choices} \\
    \mapsto&\{\{(a,f), (b,h)\}, \{(a,g), (b,i)\}, \{(b,j), (b,k)\}\}
    &\text{overcounts by a } 2!1! \text{ factor, } 12& \text{ choices}
\end{align*}

We write $\gpartcon(F)$ for a graph~$F$ (not necessarily simple) to mean 
the collection of all graph partitions of~$F$ into \emph{connected} induced subgraphs.
For any integer $n \ge 1$,
let $\chi^*(n, F)$ denote the number of all node colorings (two end nodes of an edge having different colors) 
on~$F$ using exactly $n$ colors.

\begin{lemma}\label{lem:ToColoring}
For any nonempty cluster~$\cluster W$ we have
\begin{equation}
\abs{\mdiff_{\cluster W} \betaF}
\le
\beta^{\abs{\cluster W}}
\sum_{\gpart P \in \gpartcon(\gr(\cluster W))}
\abs*{
    \sum_{n = 1}^{\abs{\gpart P}} \frac{(-1)^{n-1}}{n} \chi^*(n, \gr(\gpart P)) 
}.
\end{equation}
\end{lemma}

This is a repackaging of~\cite[App.~B.1]{WA22}.

\begin{proof}
Recall that $ \calZ = \frac 1 \fulldimension \Tr \exp( - \beta H )$.
There is a formal cluster expansion for~$\calZ$ (which is meaningful in view of~\cref{statement:analyticFunctionUnique}):
\begin{equation}
    \calZ = 1 + \sum_{k \ge 1} 
    \sum_{\cluster W: \text{ weight } k} \frac{\lambda^{\cluster W}}{\cluster W!} \mdiff_{\cluster W} \calZ
\end{equation}
where $\cluster W$ is not always connected.
The cluster derivative~$\mdiff_{\cluster W} \calZ$ factorizes 
if $\cluster W$ is not connected (\cref{statement:log-derivative-connected}).
Let $\partition P_{\max} (\cluster W)$ be the set of maximal connected subclusters of~$\cluster W$.
Then, using $\log(1+x) = \sum_{n = 1}^\infty \frac{(-x)^{n-1}}{n}$ for small~$x$
(which is again meaningful in view of~\cref{statement:analyticFunctionUnique}),
\begin{align}
    \calZ &= 1 + \sum_{k \ge 1}
    \sum_{\substack{\cluster W: \\ \abs{\cluster W} = k}} \prod_{\cluster V \in \Supp \partition P_{\max}(\cluster W)} \frac{\lambda^{\cluster V}}{\cluster V !} \mdiff_{\cluster V} \calZ ,\\
    \log \calZ &=
    \sum_{n=1}^\infty \frac{(-1)^{n-1}}{n}
    \sum_{k_1,\ldots,k_n \ge 1}
    \sum_{\substack{\cluster W_1,\ldots, \cluster W_n: \\ \abs{\cluster W_i} = k_i}}
    \prod_{i=1}^n \left( 
    \prod_{\cluster V_i \in \Supp \partition P_{\max}(\cluster W_i)} \frac{\lambda^{\cluster V_i}}{\cluster V_i !} \mdiff_{\cluster V_i} \calZ
    \right)
    \nonumber
\end{align}
Now, we rearrange this sum as a $\beta$-power series;
$\mdiff_{\cluster V} \calZ$ carries $\beta^{\abs{\cluster V}}$.
The double product can be written as a product over a 
multiset~$\partition P = \{\{ \cluster V_{1,1}, \cluster V_{1,2},\ldots, \cluster V_{n,1}, \cluster V_{n,2},\ldots \}\}$ of connected clusters.
This multiset~$\partition P$ of connected clusters 
obeys a special property that each cluster $\cluster V_{i,j}$ is assigned a label~$i$,
and among those of a given label~$i$ no two clusters become connected 
by taking the multiset union.
That is, the labels give a unique node coloring on $\gr(\partition P)$ with exactly $n$ colors.

Conversely, with a node coloring on $\gr(\partition P)$ with exactly $n$ colors,
we see that the collection of nodes of a given color~$i$
defines~$\cluster W_i$.
But this converse direction is many-to-one:
a cluster in~$\partition P$ with multiplicity $\ge 2$
gives two or more nodes in~$\gr(\partition P)$ that have all different colors,
and permuting colors among these ``duplicate'' nodes of~$\gr(\partition P)$
gives the same~$\cluster W_i$'s.
We see that precisely $\partition P !$ different colorings give the same~$\cluster W_i$'s.
The multiset~$\partition P$ is a cluster partition of the multiset union~$\cluster W = \bigcup_i \bigcup_j \cluster V_{i,j}$,
and $\abs{\cluster W} = k = k_1 + \cdots + k_n$ is the order in~$\beta$ of the double product.
Hence, letting $\ell$ assume all tuples $(i,j)$ that index~$\cluster V_{i,j}$, we see that
\begin{equation}
    \log \calZ = 
\sum_{k = 1}^\infty 
\sum_{\substack{\cluster W: \\ \abs{\cluster W} = k,\\ \cluster W \text{ connected} }}
\sum_{\substack{\partition P = \{\{ \cluster V_\ell \}\}: \\ \bigcup_\ell \cluster V_\ell = \cluster W, \\ \cluster V_\ell \text{ connected} }}
\left(\sum_{n = 1}^{\abs{\partition P}} \frac{(-1)^{n-1}}{n} \frac{\chi^*(n, \gr(\partition P))}{\partition P!}\right)
\prod_{\ell = 1}^{\abs{\partition P}} \frac{\lambda^{\cluster V_\ell}}{ \cluster V_\ell !} \mdiff_{\cluster V_\ell} \calZ .
\end{equation}

Next, we rewrite the sum over~$\partition P$ as a sum over graph partitions~$\gpart P$
of~$\gr(\cluster W)$ into connected induced subgraphs.
By~\cref{eq:cluster-to-graph-count},
for each~$\partition P$ there are exactly
$\frac{\cluster W!}{\partition P ! \prod_{\ell=1}^{\abs{\partition P}} (\cluster V_\ell !) }$
different graph partitions~$\gpart P$ that give~$\partition P$.
Each induced subgraph in~$\gpart P$ is connected iff the corresponding cluster is connected.
Therefore,
\begin{equation}
       \log \calZ = 
\sum_{\substack{\cluster W :\\ \text{connected}\\ \text{nonempty} }}
\frac{1}{\cluster W !}
\sum_{\gpart P \in \gpartcon(\gr(\cluster W))}
\left(\sum_{n = 1}^{\abs{\partition P}} \frac{(-1)^{n-1}}{n} \chi^*(n, \gr(\partition P)) \right)
\prod_{\ell = 1}^{\abs{\partition P}} \lambda^{\cluster V_\ell} \mdiff_{\cluster V_\ell} \calZ .
\end{equation}
where $\gpart P \mapsto \partition P = \{\{ \cluster V_\ell \}\}$ is implicit.
Now we can read off 
\begin{align}
\mdiff_{\cluster W} \betaF = 
\mdiff_{\cluster W} \log \calZ 
&= 
\sum_{\gpart P \in \gpartcon(\gr(\cluster W))}
\left(\sum_{n = 1}^{\abs{\partition P}} \frac{(-1)^{n-1}}{n} \chi^*(n, \gr(\partition P)) \right)
\prod_{\ell = 1}^{\abs{\partition P}} \mdiff_{\cluster V_\ell} \calZ \\
&=
\sum_{\gpart P \in \gpartcon(\gr(\cluster W))}
\left(\sum_{n = 1}^{\abs{\gpart P}} \frac{(-1)^{n-1}}{n} \chi^*(n, \gr(\gpart P)) \right)
\prod_{\ell = 1}^{\abs{\gpart P}} \mdiff_{\cluster V_\ell} \calZ  . \nonumber
\end{align}
The proof is completed by bounding~$\mdiff_{\cluster V} \calZ$ as follows.
If $\cluster V = \{\{ a_1, \ldots, a_k \}\}$, then
\begin{align}
    \mdiff_{\cluster V} \calZ 
    &= 
    \mdiff_{\cluster V} \frac 1 \fulldimension \Tr( e^{-\beta H} )
    =
    \mdiff_{\cluster V}  \sum_{n=0}^\infty \frac{(-\beta)^n}{n!} \frac 1 \fulldimension \Tr(H^n) \nonumber\\
    &=
    \mdiff_{\cluster V} \frac{(-\beta)^k}{k!} \frac 1 \fulldimension \Tr(H^k) 
    =
    \frac{(-\beta)^k}{k!} \sum_{\sigma \in \symm_k} \frac 1 \fulldimension \Tr(E_{a_{\sigma(1)}} \cdots E_{a_{\sigma(k)}}).
\end{align}
Since $\norm{E_a} \le 1$, we see $\abs{\mdiff_{\cluster V} \calZ } \le \abs{\beta}^k$.
\end{proof}

\Cref{statement:cluster-derivative-bound-no-multiplicity} 
is proved by~\cref{lem:ToColoring} and a combinatorial estimate in~\cref{lem:tutte-bound} below.
We recall some elements of graph combinatorics.

Let $F$ be a possibly nonsimple graph with self-loops and multiple edges.
Given an edge~$e$ (not a self-loop) of a graph~$F$
the \emph{contraction} of the edge~$e$, denoted by~$F / e$,
is the graph obtained by removing the edge~$e$ and merging the two end points of~$e$ into one vertex;
and the \emph{deletion} of~$e$, denoted by~$F\setminus e$, 
is the graph obtained by removing the edge~$e$ from~$F$.
The following identities are standard.
For any graph~$F$ (not necessarily simple),
\begin{align}
    \chi^*(k, F \setminus e) &= \chi^*(k,F) + \chi^*(k, F/e) , \label{eq:chi-deletion-contraction}\\
    \tau(F) &= \tau(F \setminus e) + \tau(F / e) \label{eq:tau-deletion-contraction}
\end{align}
where $\tau(F)$ is the number of all spanning trees of~$F$.
If $F$ is disconnected, $\tau(F) = 0$.
For~\cref{eq:chi-deletion-contraction},
a coloring of $F \setminus e$ either colors the endpoints of $e$ the same or different:
the colorings where they are colored differently are exactly the set of colorings of $F$,
and the colorings where they are colored the same correspond to the set of colorings of $F / e$.
For~\cref{eq:tau-deletion-contraction},
a spanning tree of~$\tau(F)$ either contains the edge~$e$ or it does not:
the spanning trees that do not contain~$e$ are exactly the spanning trees of $\tau(F\setminus e)$,
and the spanning trees that do contain~$e$ correspond to spanning trees of $F/e$.

\begin{lemma} \label{lem:tutte-bound}
    Let $G$ be a nonempty connected graph with $n$ vertices.
    Then,
    \begin{equation}
        \sum_{\gpart P \in \gpartcon(G)}
    \abs*{
        \sum_{k = 1}^{\abs{\gpart P}} \frac{(-1)^{k-1}}{k} \chi^*(k, \gr(\gpart P))
    }
    \le
    2^{n-1} \tau(G).
    \end{equation}
\end{lemma}

This can be proved via perhaps more canonical approach using Tutte polynomials and chromatic polynomials~\cite{WA22},
but we directly use the founding principle of these polynomials---the deletion-contraction recurrence.

\begin{proof}
For any nonempty (not necessarily simple) graph~$F$ on $n$ vertices,
we define a rational number~$\eta(F)$, which will turn out to be a nonnegative integer, by
\begin{equation}
    \eta(F) = \sum_{k=1}^{n} \frac{(-1)^{n - k}}{k} \chi^*(k,F)
\end{equation}
Now, consider an edge $e = (a, b)$ where $a \neq b$.
Since $F / e$ has $n - 1$ vertices, \cref{eq:chi-deletion-contraction} gives
\begin{align}
    \eta(F/e) + \eta(F\setminus e)
    &=
    \sum_{k=1}^{n -1}  \frac{(-1)^{n - 1 - k}}{k} \chi^*(k,F/e)
    +
    \sum_{k=1}^{n}  \frac{(-1)^{n - k}}{k} \chi^*(k,F\setminus e) \nonumber\\
    &=
    -\sum_{k=1}^{n}  \frac{(-1)^{n - k}}{k} \chi^*(k,F/e)
    +
    \sum_{k=1}^{n}  \frac{(-1)^{n - k}}{k} \chi^*(k,F\setminus e)\nonumber \\
    &=
    \eta(F)  \label{eq:U-recursion}
\end{align}
where the second equality is because $\chi^*(n, F/e) = 0$.
Notice that this is the same recursive formula as that of $\tau$ in \cref{eq:tau-deletion-contraction}.
Our goal will be to show $\eta(F) \leq \tau(F)$: 
by applying \cref{eq:U-recursion} and \cref{eq:tau-deletion-contraction} to reduce to cases with fewer edges via deletion and contraction, 
it suffices to show this for $F$ with no edges, and only self-loops.

We first show that, for $S$ the graph of $n$ isolated vertices with no edges, $\eta(S) = \delta_{n,1}$ where $\delta$ is the Kronecker delta.
The number of all proper colorings of~$S$ using $k$ or fewer colors is $\chi(k,S) = k^n = (x \diff_x)^n x^k |_{x=1}$,
where $x$ is an indeterminant.
By inclusion-exclusion, we have $\chi^*(k,T) = \sum_{j=0}^k \binom{k}{j} (-1)^{k-j} \chi(j,T)$.
Hence,
\begin{align}
    \eta(T) &= \sum_{k=1}^{n} \frac{(-1)^{n - k}}{k} \sum_{j=0}^k \binom{k}{j} (-1)^{k-j} (x \diff_x)^{n} x^j \Big|_{x=1} \nonumber\\
    &= (-x \diff_x)^n \Big|_{x=1} \sum_{k=1}^{n} \frac{(-1)^{k}}{k} (x-1)^k \\
    &= (-x \diff_x)^{n-1} \Big|_{x=1} x\sum_{k=1}^{\infty} (-1)^{k-1} (x-1)^{k-1} \nonumber\\
    &= (-x \diff_x)^{n-1} \Big|_{x=1} 1 = \delta_{n,1}. \nonumber
\end{align}
So, $\eta(S) = \tau(S) = \delta_{n,1}$ for graphs $S$ without any edges or loops.
For graphs~$F$ with some loops but without any edges,
$\eta(F) = 0$ since a self-loop prohibits any proper coloring,
but $\tau(F)$ may be positive if $n = 1$.
Hence, $\eta(F) \leq \tau(F)$ for $F$ without any edges.

We conclude that $0 \le \eta(F) \le \tau(F)$ for any nonempty graph~$F$.%
\footnote{
    \cite{WA22} shows that $\eta(F)$ equals the value of the Tutte polynomial~$T_F$ at~$(1,0)$,
    and quotes the facts that $T_F(1,0) \le T_F(1,1)$ and that $T_F(1,1) = \tau(F)$.
}
It follows that
\begin{equation}
    \abs[\Bigg]{
        \sum_{k = 1}^{\abs{\gpart P}} \frac{(-1)^{k-1}}{k} \chi^*(k, \gr(\gpart P))
    } = \eta(\gr(\gpart P)) \le \tau(\gr(\gpart P)).
\end{equation}

The proof is completed by observing the following~\cite[Lem.~21]{WA22}.
Consider a spanning tree~$T$ of~$G$ and a subset~$E_0$ of some edges of~$T$.
Let $E_1$ be the set of all edges of~$T$ not in~$E_0$; 
$E_0 \sqcup E_1$ is the total edge set of~$T$.
We obtain a graph partition~$\gpart P \in \gpartcon(G)$ of~$G$ into connected subgraphs,
namely the connected components of $E_0$-deleted subgraph of~$T$.
We also obtain a spanning tree of $\gr(\partition P)$,
obtained by contracting all edges of~$E_1$ from~$T$.
(Any contraction on a tree is a tree.)
Hence, given $G$, 
we have a map from pairs $(T,E_0)$ of a spanning tree and its subset of edges
to pairs $(\gpart P \in \gpartcon(G), T')$ of a graph partition~$\gpart P$ 
and a spanning tree of~$\gr(\gpart P)$.
This map is surjective: 
by choosing a spanning tree in each party of~$\gpart P$ we have~$E_0$,
and by choosing some edges, 
one among those that would merge to an edge in~$T'$ upon contraction of the spanning trees of 
the parties of~$\gpart P$, we construct a spanning tree~$T$ of~$G$ whose edge set contains~$E_0$.
This surjection gives
\begin{equation}
     \sum_{\gpart P \in \gpartcon(G)} \tau(\gr(\gpart P)) \le 2^{|G|-1} \tau(G).
\end{equation}
where $\abs G - 1$ is the number of all edges in any spanning tree~$T$ of~$G$.
\end{proof}

Combining \cref{lem:ToColoring,lem:tutte-bound}, we have
\begin{equation}
    \abs{\mdiff_{\cluster E} \betaF} \le 
    \sum_{\gpart P \in \gpartcon(G)}
    \abs*{
        \sum_{k = 1}^{\abs{\gpart P}} \frac{(-1)^{k-1}}{k} \chi^*(k, \gr(\gpart P))
    }
    \le
    2^{n-1} \tau(G).
\end{equation}
where $n = \abs{\cluster E} = \abs{G}$.
It remains to show that $\tau(G)$, the number of all spanning trees of~$G$, is at most the product of degrees of nodes.
Fix a root, an arbitrary node of~$G$.
Given a spanning tree of~$G$,
we choose a unique edge attached to each node different from the root
along which the unique shortest path to the root from the node traverses.
Each nonroot node~$a$ have $\deg(a)$, the degree of~$a$, choices at most,
implying $\tau(G) \le \prod_{a \in G} \deg(a)$.
This complete the proof 
of~\cref{statement:cluster-derivative-bound-no-multiplicity}.


\subsection{Computing cluster derivatives}\label{sec:computingClusterDerivatives}

\begin{prop}\label{statement:algo-W-deriv}
    Consider a Hamiltonian $\{(a,E_a,x_a)\}$ where every term $E_a$ is a tensor product of Pauli matrices 
    and is supported on at most $L$ qubits.
    Then, there is a deterministic algorithm running in time $(8^m + L) \poly(m)$
    such that for every cluster~$\cluster W$ of total weight $m+1$ it outputs
    $
    \frac{1}{\beta^{m+1} \cluster W !} \mdiff_{\cluster W} \betaF
    $
    exactly as a rational number.
\end{prop}
\begin{proof}
    Let $\cluster W = \{ (a, \mu_a) \}$, and suppose there are $n$ distinct $a$ in $\cluster W$.
    By assumption, $n \le m+1$.
    Without loss of generality, we relabel the indices so that $a$ iterates over $[n]$: $\cluster W = \{ (1,\mu_1), \ldots, (n,\mu_n)\}$.
    The expression $\mdiff_{\cluster W} \betaF$ depends only on $\{x_1,\ldots, x_n\}$, since it evaluates the derivative at the origin.
    Hence, to evaluate it,
    we may assume that our Hamiltonian is simply $H = \sum_{a=1}^n x_a E_a$.
    This restricted Hamiltonian has operator norm at most $n$.

    We can further simplify because $\mdiff_{\cluster W} \betaF$ depends only on the $\beta^{m+1}$ term of the Taylor expansion of $\betaF$.
    So, we can freely truncate and avoid worrying about higher-order terms.
    In particular, we will truncate the Taylor expansion of the exponential in $\betaF$ to get
    \begin{align}
    \mdiff_{\cluster W} \betaF = \beta^{m+1}\mdiff_{\cluster W}\log \Tr \exp\Big(-\sum_{j=1}^n x_jE_j\Big) = \beta^{m+1} \mdiff_{\cluster W} \log r.
    \end{align}
    for $r$ the polynomial function\footnote{
        In this proof of \cref{statement:algo-W-deriv},
        \cref{eq:rat} is the only place we use the fact that $E_a$ are Pauli operators.
        If we used an arithmetic model of computation, the operators $E_a$ can be more general.
    }
    \begin{align}
    r(x_1,\ldots,x_n) &= \sum_{k=0}^{m+1} \frac{ (-1)^k }{k!} \frac{1}{\fulldimension} \Tr\Bigl[ \Bigl(\sum_{j=1}^n x_j E_{j}\Bigr)^k \Bigr] \in \CC[x_1,\ldots,x_n], \label{eq:rat}\\
    r(0,\ldots,0) &= 1.
    \end{align}
    This expression normalizes by $\frac{1}{\fulldimension}$, which vanishes after taking $\mdiff_{\cluster W} \log$.
    To compute the derivative of $\log r$, we first make a general observation, whose proof will be given later.
    \begin{lemma}\label{statement:limit-finite-difference}
        Let $f : \RR^n \to \RR$ 
        be a smooth (infinitely differentiable in any variable in any order) function.
        Then, $f(x)$ for $x = (x_1,\ldots,x_n)$ satisfies
        \begin{align}
                \frac{\diff^{\mu} f}{\diff {x_1}^{\mu_1} \cdots \diff {x_n}^{\mu_n}}
            =
                \lim_{\alpha \to 0}
            \frac 1 {\alpha^{\mu}} \sum_{z_1=0}^{\mu_1}\cdots\sum_{z_n=0}^{\mu_n} (-1)^{|z|+\mu}\binom{\mu_1}{z_1}\cdots\binom{\mu_n}{z_n} f( x + \alpha z)
        \end{align}
        where $\mu = \mu_1 + \cdots + \mu_n$, $z = (z_1,\ldots,z_n) \in \ZZ_{\ge 0}^n$ and $\abs{z} = z_1 + \cdots + z_n$.        
    \end{lemma}
    \Cref{statement:limit-finite-difference} implies that there are some integers $c_z$ depending on $z$ such that we can write
    \begin{gather} \label{eq:alphalimit}
    \mdiff_{\cluster W} \log r = \lim_{\alpha \to 0} \frac{1}{\alpha^{m+1}} \sum_{\substack{z \in \ZZ_{\geq 0}^n \\ z_j \leq \mu_j\, \forall j \in [n]}} c_z \log g(\alpha; z) \\
    \text{where } g(\alpha; z) = r( z_1 \alpha , \ldots , z_n \alpha) \in \CC[\alpha].
    \end{gather}
    We treat $g(\alpha; z)$ as a univariate function parametrized by an integer vector $z \in \ZZ^n$.
    There are at most $2^{m+1}$ summands in \cref{eq:alphalimit}.
    Further, the limit in \cref{eq:alphalimit} can be computed by L'Hospital's rule:
    \begin{align}
        \mdiff_{\cluster W} \log r &= \frac{1}{(m+1)!}\sum_{\substack{z \in \ZZ_{\geq 0}^n \\ z_j \leq \mu_j\, \forall j \in [n]}} c_z \Big[\diff_\alpha^{m+1} \log g(\alpha; z)\Big]_{\alpha = 0} \\
        &= \frac{1}{(m+1)!}\sum_{\substack{z \in \ZZ_{\geq 0}^n \\ z_j \leq \mu_j\, \forall j \in [n]}} c_z h_m(\alpha;z)|_{\alpha = 0}
    \end{align}
    where we define $h_m(\alpha; z) := g(\alpha; z)^{m+1}\diff_\alpha^{m+1} \log g(\alpha; z)$.
    The above equality holds because \mbox{$g(0; z) = 1$}.
    The function $h_m(\alpha; z)$ is in fact a polynomial and satisfies a straightforward recursion:
    \begin{align}
        h_0 &= g \,\diff_\alpha \log g \nonumber\\
        &= \diff_a g \\
        h_t &= g^{t+1}\diff_\alpha^{t+1} \log g \nonumber\\
        &= \diff_\alpha(g^t\diff_\alpha^t\log g)g - t(g^t\diff_\alpha^t\log g)\diff_\alpha g \nonumber\\
        &= \diff_\alpha(h_{t-1}) g - t h_{t-1} \diff_\alpha g \label{eq:ht}
    \end{align}
    In summary, we have reduced the problem to finding the constant term of $h_m$, a recursively-defined polynomial, for various choices of $z$:
    \begin{align} \label{eq:reduced-to-polynomial}
    \frac{1}{\beta^{m+1} \cluster W !} \mdiff_{\cluster W} \betaF 
    = \frac{1}{\cluster W ! (m+1)!} \sum_z c_z \frac{h_m(0;z)}{g(0;z)^{m+1}}
    = \frac{1}{\cluster W ! (m+1)!} \sum_z c_z h_m(0;z).
    \end{align}
    For a fixed~$z$, the function $g(\alpha; z)$ is a polynomial in~$\alpha$.
    Using the notation that $[\alpha^k]g(\alpha; z)$ is the coefficient of $\alpha^k$ in $g(\alpha; z)$, we have that
    \begin{align}
    g(\alpha;z) &= 1 + [\alpha]g(\alpha;z)\cdot\alpha + \cdots + [\alpha^{m+1}]g(\alpha;z)\cdot\alpha^{m+1} \\
    [\alpha^k]g(\alpha;z) &= \frac{(-1)^k}{k!} \frac 1 \fulldimension \Tr \Biggl[ \Bigl(\sum_{j=1}^n z_j E_{j}\Bigr)^k \Biggr] \label{eq:trPauli}
    \end{align}

    \paragraph{Time complexity.}
    Now that we've established the form of the expression that we will compute, we will now discuss the time complexity necessary to compute it.
    First, we consider computing the coefficients $[\alpha^k]g(\alpha;z)$ as seen in \cref{eq:trPauli}.

    A binary representation of the $n$ Pauli operators $E_a$ and Gauss elimination
    reveals a minimal set (multiplicative basis) of Pauli operators 
    which can generate all the $n$ operators by multiplications together with phase factors $\pm 1,\pm i$~\cite{ag04}.
    Since the set of all those $n$ Pauli operators are supported on at most $nL$ qubits,
    It takes time $O(n^3 L)$ to find a multiplicative basis.
    Once we have a multiplicative basis,
    we can find another set of Pauli operators $\tilde E_1,\ldots, \tilde E_n$ on $n$ qubits,
    preserving all the pairwise commutation relations and the multiplicative independence.
    The procedure is simple:
    For the first basis element, choose $\tilde E_1 = Z_1$.
    If the second basis element commutes with the first, choose $\tilde E_2 = Z_2$,
    or otherwise, choose $\tilde E_2 = X_1 Z_2$.
    Inductively, for $t$-th basis element we choose $\tilde E_t$ to be $Z_t$ 
    multiplied by an appropriate Pauli operator to preserve all the commutation relations 
    with $\tilde E_{t'}$ where $t' < t$.
    Thus, it takes time $O(n^3 L)$ to find a faithful representation $\tilde E_a$ 
    of $n$ Pauli operators $E_a$.
    
    Equipped with a faithful representation $\tilde E_a$, \cref{eq:trPauli} is evaluated
    by powering a matrix $P = \sum_j z_j \tilde E_a$ of dimension $2^n$.
    The normalization constant also changes from $\frac{1}{\fulldimension}$ to $\frac{1}{2^n}$.
    The matrix $P$ has $O(n)$ entries in each column and in each row, and they are in $\ZZ[\sqrt{-1}]$ since the $E_a$'s are Pauli operators.
    Hence, multiplying a $2^n$-dimensional matrix by $P$ takes $O(4^n n)$ integer arithmetic operations.
    Since we raise $P$ to $k$-th power,
    we can compute $\Tr(P^k)$ in $O(4^nnk)$ integer arithmetic operations, and we can compute the trace for all $k \in [m+1]$ in $O(4^nnm)$ operations.
    This is the numerator of $g(;z)_k$, and we maintain $g(;z)_k$ as a rational number by maintaining its numerator and denominator.
    By multiplying both by $(k+1)(k+2)\cdots(m+1)$, we can standardize these rational number representations of $g(;z)_k$ to have be some integer in the numerator and $(m+1)!2^n$ in the denominator.

    Once we have representations for $g(;z)_k$, we can compute $h_m(0; z)$, the constant term of $h_m(\alpha; z)$, via the recursive formula \cref{eq:ht}.
    This takes $O(m^3)$ integer arithmetic operations, since $[\alpha^k]h_t$, the coefficient of $\alpha^k$ in $h_t$, satisfies
    \begin{align}
        [\alpha^k]h_0 &= (k+1)[\alpha^{k+1}]g \\
        [\alpha^k]h_t &= \sum_{j=0}^k \Bigl[[\alpha^j]h_{t-1}' [\alpha^{k-j}]g - t\cdot [\alpha^j]h_{t-1} [\alpha^{k-j}]g'\Bigr] \nonumber\\
        &= \sum_{j=0}^k\Bigl[(j+1)[\alpha^{j+1}]h_{t-1} [\alpha^{k-j}]g - t (k-j+1)[\alpha^j]h_{t-1} [\alpha^{k-j+1}]g\Bigr].
        \label{eq:h-coeff-recursion}
    \end{align}
    Because our goal is to compute $[\alpha^0]h_m$, we only need to compute the $[\alpha^k]h_t$'s for $t$ from $0$ to $m$ and $k$ from $0$ to $m-t$.
    Since we can compute $[\alpha^k]h_t$ with $O(k)$ integer arithmetic operations, we can compute $[\alpha^0]h_m$ with $O(m^3)$ integer operations.
    Since there are at most $2^{m+1}$ summands in \cref{eq:alphalimit}, computing $\frac{1}{\beta^{m+1}\cluster{W}!}\mdiff_{\cluster W}\betaF$ requires $O(2^m(4^n nm + m^3)) = O(8^m m^3)$ integer operations.
    Since the integer operations in question (addition, subtraction, and multiplication) can be performed in $O(b\log b)$ time~\cite{hh21}, where $b$ is the length of the integer in bits, it suffices to show that, throughout this procedure, we always work with integers that are $\poly(m)$ bits long.

    When computing coefficients of $g$, note that the magnitude of the integers in $P^k$ is bounded by $\norm*{P^k} = O(m^k) = O(m^m)$, so the trace (and consequently, the numerator of $[\alpha^k]g$ can be represented with $O(m\log m)$ bits.
    The denominator is $m!2^n$, which can also be represented in $O(m\log m)$ bits.
    
    As for $h_t$, we proceed by giving upper bounds on the coefficients.
    We claim that for any $k \le m+1$ and any $t \le m$
    \begin{align}
        \abs*{[\alpha^k]h_t} \le 2^t (m+1)^{3t + 1} e^{(t+1)(m+1)}.
    \end{align}
    The case of $t = 0$ is shown below.
    \begin{align}
    \abs*{[\alpha^k]g} 
    &\le \frac{1}{k!} \Bigl(\sum_{j=1}^n z_j \Bigr)^k
    \le \frac 1 {k!} \Bigl(\sum_{j=1}^n \mu_j \Bigr)^k = 
    \frac 1 {k!} (m+1)^k \le \sum_{s=0}^\infty \frac{(m+1)^s}{s!} = e^{m+1}\label{eq:gbound} \\
    \abs*{[\alpha^k]h_0} &= \abs*{[\alpha^k]g'} \le \frac{(m+1)^{k+1}}{k!} \le (m+1)e^{m+1}. \label{eq:gprimebound}
    \end{align}
    Since taking derivative brings a factor at most $(m+1)$ 
    to the polynomial coefficient of $k$-th order term where $k \le m+1$,
    the induction hypothesis implies that
    \begin{align}
        \abs*{[\alpha^k]h_t} 
        &= \abs[\Bigg]{\sum_{j=0}^k\Bigl[(j+1)[\alpha^{j+1}]h_{t-1} [\alpha^{k-j}]g - t (k-j+1)[\alpha^j]h_{t-1} [\alpha^{k-j+1}]g\Bigr]} \\
        &\le (m+1) \left[ (m+1) 2^{t-1} (m+1)^{3t-2} e^{t(m+1)} e^{m+1} \right.\nonumber \\
        &\qquad + \left. (m+1)^2 2^{t-1} (m+1)^{3t-2} e^{t(m+1)} e^{m+1}\right]\nonumber\\
        &= 2^{t-1} (m+1)^{3t} e^{(t+1)(m+1)}(m+2)\nonumber\\
        &\le 2^{t} (m+1)^{3t+1} e^{(t+1)(m+1)} .\nonumber
    \end{align}
    This implies that every coefficient of $h_t(\alpha;z)$ up to the $(m+1)$-th order
    is at most $\exp(O(m^2))$ in magnitude.%
    \footnote{
        If this bound were $\exp(O(m))$,
        we would have implied \Cref{statement:cluster-deriv-bound}.
    }
    Because all coefficients of $g$ are represented with a denominator of $(m+1)!2^n$, when computing $[\alpha^k]h_t$, the denominator is always $((m+1)!)^{t+1}2^{n(t+1)} = \exp(O(m^2\log m))$.
    So, along with the magnitude bound, the numerator of these coefficients is always $\exp(O(m^2\log m))$.
    Thus, we are always working with integers of $O(m^2\log m)$ digits, as desired.

    Altogether, this makes the time complexity of computing a coefficient
    \begin{equation}
        O(n^3L + 8^mm^3(m^2\log m(\log(m^2\log m))))
        = O(Lm^3 + 8^mm^5\log^2m).
    \end{equation}
\end{proof}

We summarize the algorithm in \cref{alg:cluster-derivative}.
\begin{algorithm}
    \caption{Evaluating a cluster derivative}\label{alg:cluster-derivative}
    \KwData{Cluster $\cluster W = \{(a,\mu(a)) : a = 1,2,\ldots,n \}$ of total weight $\sum_a \mu(a) = m+1$ 
    and associated Pauli operator $E_a$ with $a \in \cluster W$}
    \KwResult{A rational number $
    \frac{1}{\beta^{m+1} \cluster W !} \mdiff_{\cluster W} \betaF
    $}
    Let $\operatorname{out} \leftarrow 0$\;
    Let $\tilde E_a$ be a faithful representation on $m+1$ qubits for $E_a$\;
    \For{ $z_1 \in \{0,1,\ldots,\mu(1)\}, \ldots, z_n \in \{0,1,\ldots,\mu(n)\} $ }{
        Compute coefficients $[\alpha^k]g(\alpha;z)$ of \cref{eq:trPauli} using $\tilde E_a$\;
        \For{$t \in \{0,1,\ldots,m\}$}{
            Compute $[\alpha^k]h_t(\alpha; z)$ for $k$ from 0 to $m-t$ by \cref{eq:h-coeff-recursion}\;
        }
        $\operatorname{out} \leftarrow \operatorname{out} + (-1)^{z_1+\cdots+z_n + m+1} \binom{\mu(1)}{z_1} \cdots \binom{\mu(n)}{z_n} [\alpha^0]h_m(\alpha; z)$\;
    }
    Return $\frac{1}{\cluster W ! (m+1)!} \operatorname{out}$\;
\end{algorithm}

\begin{proof}[Proof of \cref{statement:limit-finite-difference}]
    Without loss of generality, we assume that our derivative is taken at the origin in the domain.
    Define functions $e_k$ for $k=0,1,2,\ldots,n$ recursively as
    \begin{align}
    e_0 &= f,\\
    e_k(x_1,\ldots,x_n) 
    &= e_{k-1}(x_1,\ldots,x_{k-1},x_k,x_{k+1},\ldots,x_n) - e_{k-1}(x_1,\ldots,x_{k-1},0,x_{k+1},\ldots,x_n).
    \end{align}
    The function $e_n$ is a sum of $2^n$ terms.
    The mean value theorem implies that
    for any $\alpha \neq 0$
    there exist $\theta_n,\theta_{n-1},\ldots, \theta_1 \in (0,1)$ such that
    \begin{align}
    e_n(\alpha,\alpha,\ldots,\alpha)
    &= \alpha (\diff_n e_{n-1})(\underbrace{\alpha,\ldots,\alpha}_{n-1}, \theta_n \alpha)\nonumber\\
    &= \alpha (\diff_n \alpha \diff_{n-1} e_{n-2}) (\underbrace{\alpha,\ldots,\alpha}_{n-2},\theta_{n-1} \alpha, \theta_n \alpha)\nonumber\\
    &= \alpha^n (\diff_n \diff_{n-1} \cdots \diff_1 e_0)(\theta_1 \alpha, \theta_2 \alpha, \ldots, \theta_n \alpha) .
    \end{align}
    This means that
    \begin{align}
    \frac{\diff^n f}{\diff x_1 \cdots \diff x_n} \biggr|_{\vct x = 0} = \lim_{\alpha \to 0} \frac{e_n(\alpha,\ldots,\alpha)}{\alpha^n}
    = 
    \lim_{\alpha \to 0}
    \frac 1 {\alpha^{n}} \sum_{y \in \{0,1\}^n} (-1)^{|y|+n} f( x + \alpha y),     
    \end{align}
    where $\abs{y}$ is the sum of components of $y$.
    This completes the proof of the lemma in the case where the derivative is first order in each variable.
    
    Higher order cases are proved by considering the composition $h$ of $f$ and a linear function
    \begin{align}
    h :\quad &(x_{i,j} ~|~j = 1,\ldots, \mu_i,\quad i = 1,2,\ldots,n )\\
    &\mapsto
    \left(x_i = \sum_j x_{i,j} ~\biggr|~ i = 1,2,\ldots, n \right)\nonumber\\
    &\mapsto
    f(x_1,\ldots,x_n)\nonumber
    \end{align}
    We see that $\diff_1^{\mu_1} \cdots \diff_n^{\mu_n} f = ( \prod_{i=1}^n \prod_{j=1}^{\mu_i} \diff_{x_{i,j}} ) h$.
    Let $\mu = \mu_1 + \mu_2 + \cdots + \mu_n$.
    For any $y \in \{0,1\}^\mu$ 
    we define $z(y) \in \ZZ_{\ge 0}^n$ to be a vector whose component $z(y)_k$ is
    \begin{align}
            z(y)_k = y_{\mu_1 + \cdots + \mu_{k-1} + 1} + y_{\mu_1 + \cdots + \mu_{k-1} + 2} + \cdots + y_{\mu_1 + \cdots + \mu_{k}}. \label{eq:zy}
    \end{align}
    In other words, we put the components of $y$ into $n$ bins of sizes $\mu_k$ and
    sum the numbers in each bin to make $z(y)$.
    Then, an arbitrary mixed derivative is expressed as
    \begin{align}
        \frac{\diff^{\mu} f}{\diff {x_1}^{\mu_1} \cdots \diff {x_n}^{\mu_n}} \biggr|_{x}
        =
        \lim_{\alpha \to 0}
        \frac 1 {\alpha^{\mu}} \sum_{y \in \{0,1\}^\mu} (-1)^{|y|+\mu} f( x + \alpha z(y)) .
    \end{align}
    Expressing the summation over $y$ as a summation over $z_k$, we complete the proof.
\end{proof}
    
\subsection{Proof of \texorpdfstring{\cref{thm:expectation-value}}{Theorem \ref*{thm:expectation-value}}}

To be clear where we are evaluating derivatives,
in this proof we let $\betaF = \betaF(\vct \lambda)$ be the function of variables $\lambda_a$.
Let $\diff_a$ denote the derivative $\diff / \diff \lambda_a$ 
at $\vct \lambda = \vct \xi$.
\Cref{statement:expectation-as-betaF-derivative} says that
the expectation value $\Tr(E_a e^{-\beta H}) / \Tr e^{-\beta H}$ of $E_a$
is given by the first derivative of the logarithmic partition function:
\begin{align}
    \frac{\Tr(E_a e^{-\beta H})}{\Tr e^{-\beta H}} \biggr|_{\vct \lambda = \vct \xi}
    &= -\frac{1}{\beta} \diff_a \log \Tr \exp( - \beta H) |_{\vct \lambda = \vct \xi}\label{eq:pexpec} \\
    &=
    -\frac{1}{\beta} \diff_{a}|_{\vct \lambda = \vct \xi}
    \sum_{m=0}^\infty \sum_{\cluster V : \abs{\cluster V} = m} \frac{\lambda^{\cluster V}}{\cluster V !} 
    \mdiff_{\cluster V} \betaF
    & \text{by \cref{eq:multivariateTaylor}}
\intertext{
Recall that $\mdiff_{\cluster V} \betaF$ is a constant in $\lambda$, since $\mdiff_{\cluster V}$ is a derivative evaluated at $\vct \lambda = 0$ (\cref{eq:mdiff-def}).
This multivariate Taylor series is in fact a (disguised) power series in $\beta$.
Since all the functions here
are complex differentiable on an open set that contains $\RR$
in the $\beta$-complex plane,
\Cref{statement:analyticFunctionUnique} implies that the equality holds
whenever the series is absolutely convergent and $\beta \in \RR$,
in which case the infinite sum over $m$ can be interchanged with $\diff_{a}$.
The term with $m=0$ is eliminated by $\diff_a$,
so we shift the dummy variable~$m$ by one.
Hence,
}
    &= -\frac{1}{\beta} 
    \sum_{m=0}^\infty \sum_{\cluster V : \abs{\cluster V} = m+1} \frac{\diff_{a} \lambda^{\cluster V}|_{\vct \lambda = \vct \xi}}{\cluster V !} 
    \mdiff_{\cluster V} \betaF \\
    &= -\frac{1}{\beta} 
    \sum_{m=0}^\infty \sum_{\substack{\cluster V : \abs{\cluster V} = m+1,\\a \in \cluster V}} \frac{\diff_{a} \lambda^{\cluster V}|_{\vct \lambda = \vct \xi}}{\cluster V !} 
    \mdiff_{\cluster V} \betaF.
    & \text{since $\diff_{a} \lambda^{\cluster V} = 0$ if $a \not\in \cluster V$}
\end{align}
For a cluster $\cluster V$ of total weight $m+1$, the expression $\mdiff_{\cluster V}\betaF$ has a factor of $\beta^{m+1}$.
The overall $\frac 1 \beta$ factor reduces the exponent of $\beta$ by one, so, if we group summands by degree of $\beta$, we have
\begin{gather}
    \frac{\Tr(E_a e^{-\beta H})}{\Tr e^{-\beta H}} \biggr|_{\vct \lambda = \vct \xi} 
    = p_0 + \beta p_1 + \beta^2 p_2 + \cdots, \\
    p_m = (-1)^{m+1} \sum_{\substack{\cluster V : \abs{\cluster V} = m+1,\\a \in \cluster V}} (\diff_{a} \lambda^{\cluster V} |_{\vct \lambda = \vct \xi})
    \frac{\mdiff_{\cluster V} \betaF}{\beta^{m+1} \cluster V !}.
\end{gather}
Note that $p_0$ is proportional to \cref{eq:pexpec} with $\vct \lambda = 0$.
Since $\Tr E_a = 0$, we have $p_0 = 0$.
This proves \cref{thm:expectation-value1}
that $p_m$ is a homogeneous polynomial in $\lambda_a$ of total degree $m$.
\Cref{statement:log-derivative-connected} says that $\cluster V$ has to be connected and includes $a$,
implying that $\cluster V$ of total weight $m+1$ includes nodes within $\graph$-distance $m$ from $a$.
This implies \cref{thm:expectation-value2}.
\Cref{statement:count-clusters} bounds the number of clusters to be summed over.
This is \cref{thm:expectation-value3}.
\Cref{statement:cluster-deriv-bound} bounds the magnitude of $\mdiff_{\cluster V} \betaF / \cluster V !$.
The derivative $\diff_a$ may put an additional factor at most $m+1$.
This proves \cref{thm:expectation-value4}.

We have considered algorithms to enumerate clusters, proving \cref{thm:expectation-valueA}.
\Cref{statement:algo-W-deriv} shows \cref{thm:expectation-valueB}.
This completes the proof of \cref{thm:expectation-value}.

\section{Learning algorithm}
\label{sec:algorithm}

In this section, we describe our algorithm for learning the coefficients of a Hamiltonian given copies of its Gibbs state. This section relies on the results of the previous section only through \Cref{thm:expectation-value}. 

Unlike \cref{sec:cluster}, we only consider Hamiltonians $\{(a,E_a, \lambda_a) : a = 1,2,\ldots, \numTerms \}$ (\cref{defn:hamiltonian}) where the $E_a$'s are distinct non-identity tensor products of Pauli matrices, so that they are orthonormal 
with respect to the normalized Hilbert--Schmidt inner product.
That is,
\begin{align}
\forall a,b \in [\numTerms]: \quad \Tr(E_a E_b) = \fulldimension \delta_{ab}, \label{eq:orthonormality}
\end{align}
where $\delta_{ab}$ is the Kronecker delta function.

Our overall strategy for the learning algorithm can be broken down into the following two steps.
Let $\rho(\lambda)$ be the Gibbs state with coefficients $\lambda$.
\begin{enumerate}
    \item Find estimates $\est_a$ for all of the expectation values $\angles{E_a}(\lambda) = \Tr(E_a \rho(\lambda))$ that satisfy $\abs*{\est_a - \angles{E_a}(\lambda)} \leq \beta\eps$ for all $a \in [\numTerms]$.
    \item Then (approximately) invert the function $\vct{x} \mapsto \angles{E_a}(\vct{x})$ on these estimates to find an estimate of the coefficients $\hat{\vct{x}}$.
\end{enumerate}

Step~1 of this plan is the easier step and not too hard to establish.

\begin{lemma}\label{lem:stepone}
    Consider a Hamiltonian $\{(a,E_a,\lambda_a): a \in [\numTerms]\}$ on $\numQubits$ qubits.
    We can find estimates $\est_a$ such that $\abs*{\est_a - \angles{E_a}(\lambda)} \leq \beta\eps$ for all $a \in [\numTerms]$, with probability at least $1-\delta$, using only $O(\frac{\degree}{\beta^2\eps^2}\log(\frac{\numTerms}{\delta}))$ copies of the Gibbs state and with time complexity $O(\frac{\numQubits\degree}{\beta^2\eps^2}\log(\frac{\numTerms}{\delta}))$.
\end{lemma}

\begin{proof}
Recall the problem of estimating $\Tr(E\rho)$ for $\norm{E}\leq 1$ and a quantum state $\rho$.
If we want to estimate this to accuracy $\eps$ with success probability at least $1-\delta$, 
it is a standard result that this can be done with $O(\log(1/\delta)/\eps^2)$ copies of $\rho$.
Indeed, we measure $\rho$ in the eigenbasis of $E$ and output the corresponding eigenvalue of $E$ on getting that outcome. 
This is a random variable with expected value $\Tr(E\rho)$.
Since $\norm{E}\leq 1$, this is a random variable in $[-1,1]$. 
Hence by the Chernoff bound we can estimate it to additive error $\eps$ with probability at least $1-\delta$ 
using $O(\log(1/\delta)/\eps^2)$ copies of $\rho$.

Now we want to measure all the observables $E_a$.
But not all of these have overlapping support, 
and we can measure a large number of them simultaneously.
Imagine we color the vertices of $\graph$ using $\degree +1$ colors 
such that no neighboring pair of nodes have the same color;
a greedy coloring algorithm can be used.
By definition of the dual interaction graph,
all the $E_a$'s of a particular color act on separate qubits.
So we can estimate all of the $E_a$'s of a particular color 
using only $O(\log(1/\delta') / \eps'^2)$ Gibbs state
where $\delta'$ is the probability that one of estimates has error larger than $\eps' = \eps \beta$.

Since $E_a$ is a Pauli operator (a tensor product of single-qubit Paulis),
it suffices to measure individual qubits in some Pauli basis and multiply them (each of which is $\pm 1$) 
to infer the eigenvalue of $E_a$.
Hence, for a particular color, the time complexity is $O(\numQubits \log(1/\delta') / \eps'^2)$.
We repeat this for each color, resulting in $\degree +1$ rounds.

Since we want all $\numTerms$ estimates to be correct with probability at least $1-\delta$, 
it suffices to set $\delta' = \delta / \numTerms$ to apply the union bound.
\end{proof}

The remainder of this section is devoted to implementing Step~2 of the above plan. 
We start by upper bounding the sample complexity, 
and then move on to bounding the time complexity of our algorithm.

\subsection{Definitions and a sample complexity upper bound}

Recall that \Cref{thm:expectation-value} implies that we can expand $\angles{E_a}$ into a Taylor series
\begin{align}
    \angles{E_a}(\vct{x}) = \beta p_1^{(a)}(\vct x) + \beta^2 p_2^{(a)}(\vct{x})+ \beta^3 p_3^{(a)}(\vct{x}) + \cdots,
\end{align}
where the sum of the absolute values of the coefficients of $p_m$ is bounded by a universal constant that depends only on $\degree$ and $m$.
We call this constant $c_m \in \RR_{>0}$, and from \Cref{thm:expectation-value3,thm:expectation-value4} we have 
\begin{align}
    c_m &= e\degree(1+e(\degree-1))^m(2e(\degree + 1))^{m+1}(m+1)\nonumber \\ 
    &= 2e^2\degree(\degree+1) \tau^{m}(m+1),\label{eq:cm}
\end{align}
where
\begin{equation}
   \tau = (1+e(\degree - 1))(2e(\degree + 1)) \le 2e^2 (\degree + 1)^2. \label{eq:tau}
\end{equation}
Further, $p_k^{(a)}(\vct{x})$ only depends on the entries of $\vct{x}$ 
whose operators are within $\graph$-distance $k$ from $a$.
The first term $p_1^{(a)}$ can be determined more explicitly by 
\begin{equation}
p_1^{(a)}(\vct x) 
= \frac{\diff}{\diff \beta} \frac{\Tr (E_a \exp(-\beta H))}{\Tr \exp(-\beta H)}\biggr|_{\beta = 0, \vct \lambda = \vct x}
= \frac{1}{\fulldimension} \Tr( E_a (-H))|_{\vct \lambda = \vct x}
= - x_a,
\end{equation}
where we used \cref{eq:orthonormality} in the last equality.

Let $\fun : [-1,1]^\numTerms \to \RR^\numTerms$ be $\angles{E_a}(\vct{x})$, 
truncated to order $\truncM$ terms ($\truncM \ge 1$) and shifted by our known estimates $\est_a$ of $\angles{E_a}(\vct{\lambda})$ from \Cref{lem:stepone}, which satisfy $\abs{\est_a - \angles{E_a}(\lambda)} \leq \beta\eps$.  Thus we have
\begin{align}
    \fun_a( \vct x ) := \fun_a( x_1, \ldots, x_\numTerms ) = \sum_{k = 0}^{\truncM} \beta^k p_k^{(a)}(\vct{x}) = - \est_a - \beta x_a + \beta^2 p_2^{(a)}(\vct x) + \cdots + \beta^\truncM p_\truncM^{(a)}(\vct x),
\end{align}
where we defined $p_0^{(a)} = -\est_a$.
Our goal is to find an $\vct{x}$ such that $\mathcal{F}(\vct{x})$ is small, since, as we argue below, such an $\vct{x}$ will be close to the true coefficient vector $\lambda$. 

As a warmup for the time complexity upper bound proved in the next section, we will show a sample complexity upper bound.
The fundamental idea in both upper bounds is the same: Find an $\vct{x}$ such that $\inorm{\fun(\vct{x})} = O(\beta\eps)$.
\begin{theorem} \label{thm:sample-complexity}
    Consider a Hamiltonian $\{(a,E_a,\lambda_a): a \in [\numTerms]\}$ such that $E_a$ are traceless and orthonormal with respect to the Hilbert-Schmidt inner product.
    Then, for any $\beta$ such that
    \begin{equation}
     100 e^6 (\degree+1)^8 \beta \le 1, \label{eq:degree8beta}
    \end{equation}
    we can find $\vct{x} \in [-1,1]^\numTerms$, such that $\|x - \lambda\|_\infty \leq \eps$ with probability $\geq 1-\delta$ using only 
    \begin{equation}
        O\left(\frac{\degree}{\beta^2\eps^2}\log\frac{\numTerms}{\delta} \right)
    \end{equation}
    copies of the Gibbs state.
\end{theorem}
\begin{proof}
From \Cref{lem:stepone} we know that $O(\frac{\degree}{\beta^2\eps^2}\log(\frac{\numTerms}{\delta}))$ Gibbs states suffice to estimate $\angles{E_a}$ to $\beta\eps$ accuracy for all $a$ with probability $\geq 1-\delta$.

Next, consider $\fun$ for $\truncM = \infty$, so $\fun_a(\vct{x}) = \angles{E_a}(\vct{x}) - \est_a$.
Notice that this means that $\|\fun(\lambda)\|_\infty \leq \beta\eps$, by our assumption about the accuracy of the estimates $\est_a$.
Our algorithm will be to find and output any $\vct{x} \in [-1,1]^\numTerms$ satisfying $\|\fun(\vct{x})\|_\infty \leq \beta\eps$.
We know one such $x$ must exist, since $\lambda$ satisfies this equation. It remains to be shown that any such $x$ is also close to $\lambda$.

Let $\diff_b$ denote the derivative with respect to $x_b$ and let $J = \rd \fun$ be the Jacobian of $\fun$, so $J_{ab} := \diff_b \fun_a$.
Then, for each $a$, by the multivariate mean value theorem, there exists $\vct y(a) \in (-1,1)^\numTerms$ such that
\begin{equation}
\fun_a(\vct x) = \fun_a(\lambda) + (J|_{\vct y(a)} (\vct x - \lambda))_a.
\end{equation}
This implies that
\begin{align}
\abs{\vct x_a - \lambda_a}
&= \abs{\sum_b (J|_{\vct{y}(a)}^{-1})_{ab}(\fun_b(\vct{x}) - \fun_b(\lambda))} \nonumber\\
&\le \iinorm{J|_{\vct y(a)}^{-1}} (\inorm{\fun(\vct x)} + \inorm{\fun(\lambda)}) 
\le (2 \beta^{-1})(2\beta\eps) = 4\eps, \label{eq:samp-final}
\end{align}
where the final inequality uses \cref{lem:jac-bound} below, which holds when $\beta$ is bounded as in \cref{eq:degree8beta}.
Rescaling $\eps \to \frac{1}{4}\eps$ completes the proof.
\end{proof}

\begin{lemma} \label{lem:jac-bound}
    For Hamiltonians as in \cref{thm:sample-complexity}, 
    if \cref{eq:degree8beta} holds, 
    then for any $\vct x \in [-1,1]^\numTerms$, 
    we have $\iinorm{I + \beta^{-1}J(\vct x)} \leq \frac12$ 
    and $\iinorm{J(\vct x)^{-1}} \leq 2\beta^{-1}$ for any $\truncM \ge 1$.
\end{lemma}
In particular, the lemma is true when $\truncM = \infty$ and 
consequently $J$ is also the Jacobian of the function
$\RR^\numTerms \ni x \mapsto (\angles{E_a}(\vct{x})) \in \RR^\numTerms$.
The proof implicitly uses a band-diagonal property of $J$: if $b$ and $a$ are distance $k$ apart, then $J_{ab}$ scales as $\beta^{k+1}$.
\begin{proof}
In this proof we suppress the argument $\vct x$ in $J(\vct x)$.
If $\iinorm{I + \beta^{-1}J} \le \frac 1 2$, then
\begin{align}
J^{-1} &= -\frac{1}{\beta} \frac{I}{I - (I + \beta^{-1}J)} = -\frac{1}{\beta} \sum_{k=0}^\infty (I + \beta^{-1} J)^k \label{eq:Jinv}\\
\iinorm{J^{-1}} &\le \beta^{-1} \sum_{k=0}^\infty \iinorm{I + \beta^{-1} J}^k \le 2\beta^{-1} .\label{eq:JinvBound}
\end{align}
Hence, we have to show that $\iinorm{\beta I + J} \le \frac \beta 2$ in the stated range of $\beta$
to complete the proof.
The leading order term of $J$ is $-\beta I$,
\begin{align}
    J_{ab} = \diff_b \fun_a = -\beta \delta_{ab} + O(\beta^2);
\end{align}
we will bound the rest of $J$ to show that $J$ is close to $-\beta I$.
Let $u = (u_1,u_2,\ldots, u_\numTerms)$ be such that $\abs{u_b} \le 1$ for all $b$.
\begin{align}
    ((J + \beta I)u)_a
    &= \sum_b (J + \beta I)_{ab} u_b \nonumber\\
    &=\sum_b u_b \left( \beta^2 \diff_b p_2^{(a)}(\vct x) +\cdots + \beta^\truncM \diff_b p_\truncM^{(a)}(\vct x) \right) \nonumber\\
    &=\sum_{k=2}^\truncM \beta^k \sum_{b:\dist(a,b) \le k} u_b \diff_b p_k^{(a)}(\vct x) & \text{by \cref{thm:expectation-value2}.} \label{eq:JbetaI}
\intertext{
For each~$k$ in the last sum, the index~$b$ ranges over
at most $1 + \degree + \cdots + \degree^k \leq (\degree+1)^{k}$ nodes of $\graph$.
Further, \cref{thm:expectation-value1} says that $p_k^{(a)}$ is a homogeneous polynomial of degree $k$ 
and the sum of the absolute value of its coefficients is bounded by $c_k$ of \cref{eq:cm}.
As a result, $\abs{\diff_b p_k^{(a)}} \leq k c_k$ everywhere in the domain of $\fun$.
}
    \abs{((J + \beta I)u)_a}
    &\leq\sum_{k=2}^\infty \beta^k \cdot (\degree+1)^{k} \cdot k c_k \label{eq:Jbetainf-a}\\
    &\le 2e^2(\degree+1)^2 (\beta (\degree+1) \tau)^2 \sum_{k=2}^\infty (\beta (\degree+1) \tau)^{k-2} \cdot k(k+1) \label{eq:Jbetainf-b}\\
    &= 2e^2(\degree+1)^4 \beta^2 \tau^2 \biggl(\frac{6-6r+2r^2}{(1-r)^3}\Bigr|_{r=\beta(\degree+1)\tau}\biggr) & \text{if $\beta(\degree+1)\tau < 1$} \nonumber\\
    &\le 2e^2(\degree+1)^4 \beta^2 \tau^2 \cdot \frac{25}{4} & \text{if $\beta(\degree+1)\tau \le \frac 1 {100}$.} \label{eq:Jbetainf}
\end{align}
Since $u \in [-1,1]^\numTerms$ is arbitrary, 
the last quantity is an upper bound on $\iinorm{J + \beta I}$.
The bound on $\tau$, \cref{eq:tau}, and the bound on $\beta$, \cref{eq:degree8beta}, together imply that it is $\leq \frac{\beta}{2}$.
\end{proof}

With this analysis, we can also deduce a bound on the strong convexity of the log-partition function, as analyzed by \cite{AAKS21}, that is optimal up to constants.
This is simply a matter of bounding $J^{-1}$ in the usual operator norm, $\|\cdot\|_{2\to 2}$, rather than the (in this case, larger) $\|\cdot\|_{\infty \to \infty}$ norm.

\begin{corollary}
    For Hamiltonians as in \cref{thm:sample-complexity},
    if \cref{eq:degree8beta} holds, 
    then $\betaF$ is $(\frac{\beta^2}{2})$-strongly convex,
    i.e.,
	$\nabla^{\otimes 2}  \betaF - \frac{\beta^2}{2}I$ (whose $(a,b)$-component is $\diff_a\diff_b \betaF - \beta^2 \delta_{ab}/2$) 
	is positive semidefinite.
    The strong convexity constant is only a constant factor off from optimal.
\end{corollary}
\begin{proof}
In this proof we suppress the argument $\vct x$ in $J(\vct x)$, and similar arguments.
By \cref{statement:expectation-as-betaF-derivative}, $-\beta J$ is the Hessian of $\betaF$, taking $\truncM = \infty$.
Since it comes from a Hessian, $J$ is Hermitian.
So, it suffices to show that $\norm{I + \beta^{-1}J} \le \frac 1 2$, since
\begin{align}
    \norm{I + \beta^{-1}J} &\le \frac 1 2 \implies I + \beta^{-1}J \preceq I/2 \implies \beta^{-1}J \preceq -I/2 \implies \nabla^{\otimes 2} \betaF \succeq \beta^2 I/2, \\
    \norm{I + \beta^{-1}J} &\le \frac 1 2 \implies -I - \beta^{-1}J \preceq I/2 \implies -\beta^{-1}J \preceq 3I/2 \implies \nabla^{\otimes 2} \betaF \preceq \beta^2 3I/2.
\end{align}
The second equation above proves optimality, up to a factor of 3.
The bound we need follows immediately from \cref{lem:jac-bound}, since for a Hermitian matrix $X$, it holds that $\|X\| \leq \|X\|_{\infty \to \infty}$.
(For an eigenvector $v$ achieving $Xv = \mu v$ with $\abs{\mu} = \|X\|$, we see $\|X\|_{\infty \to \infty} \geq \|Xv\|_{\infty}/\|v\|_{\infty} = \|X\|$.)
So,
\begin{align}
    \norm{I + \beta^{-1}J} \le \iinorm{I + \beta^{-1}J} \le \frac 1 2,
\end{align}
as desired.
\end{proof}

\begin{remark}\label{rm:aaks-inf}
    In this remark, we show how to tweak the result in \cite{AAKS21} to get a slightly improved version shown in \cref{eq:AAKS}.
    We assume knowledge of \cite{AAKS21}.
    First, if we do not perform the final bound in \cite[Proof of Theorem 28, p.28]{AAKS21}, we have that $v^\dagger \nabla^{\otimes 2} \betaF v \geq C\|v\|_\infty^2$ for $C = e^{-O(\beta^c)}\beta^{c'}$.
    Using \cref{statement:expectation-as-betaF-derivative}, we have that $-\beta (v^\dagger J v) \geq C\|v\|_\infty^2$.
    Consider taking the $v$ that achieves $\|J^{-1}v\|_\infty = \|J^{-1}\|_{\infty\to\infty}\|v\|_{\infty}$.
    Then, using that $\|X\|_{2 \to 2} \leq \|X\|_{\infty \to \infty}$ for Hermitian $X$,
    \begin{equation}
        C\|J^{-1}\|_{\infty \to \infty}^2\|v\|_{\infty}^2 = C\|(-J)^{-1} v\|_\infty^2 \leq \beta (v^\dagger J^{-1} v)
        \leq \beta \|J^{-1}\| \|v\|_2^2 \leq \beta\|J^{-1}\|_{\infty \to \infty} \numTerms\|v\|_\infty^2.
    \end{equation}
    So, $\|J^{-1}\|_{\infty\to\infty} \leq \frac{\beta\numTerms}{C}$.
    This can be plugged in directly into, say, \cref{eq:samp-final} to see that, using this bound, we would need to estimate the marginals to $\eps\frac{C}{\beta\numTerms}$ error, giving the bound.
    Note that the assumption that $\beta = O(1)$ is not needed to achieve this sample complexity bound.
\end{remark}

\subsection{Time complexity and analysis of the Newton--Raphson method} \label{sec:newton-raphson}
The goal of this section is to prove the following theorem, which when combined with \Cref{lem:stepone} to get the assumed estimates, gives us the main result (\Cref{thm:algorithm}).
\begin{theorem} \label{thm:newton}
    Consider a Hamiltonian $\{(a,E_a,\lambda_a): a \in [\numTerms]\}$ such that $E_a$ are traceless and orthonormal with respect to the Hilbert-Schmidt inner product.
    Suppose $\beta > 0 $ satisfies
    \begin{align}
    25 e^6(\degree+1)^{10} \beta \le 1. \label{eq:beta-bound-for-newton-convergence}
    \end{align}
    Suppose we know estimates $\est \in [-1,1]^\numTerms$ such that $\abs*{\est_a - \angles{E_a}} \leq \beta\eps$ for all $a \in [\numTerms]$.
    Then we can find an $\vct{x}$ such that $\|x - \lambda\|_\infty \leq 18\eps$
    in time  $O\Bigl(\frac{\numTerms L}{\eps} \poly(\degree, \log\frac{1}{\beta\eps})\Bigr)$.
\end{theorem}
Recall that we defined $L$ as the maximum number of qubits that a Hamiltonian term acts on in \Cref{thm:expectation-value}. If $L$ and $\degree$ are constant (as in our definition of a low-intersection Hamiltonian), then our time complexity has linear dependence in $\numTerms$, which is optimal since our output consists of $M$ numbers.
In addition, our $\eps$-dependence is better than the $\eps^{-2}$ dependence in the sample complexity.
There is very mild $\beta$-dependence since $\fun$ becomes simpler for smaller $\beta$.
The rest of this section constitutes the proof of this theorem.

\begin{algorithm}
    \caption{The Newton--Raphson method}\label{alg:newton}
    \KwData{$\beta$ satisfying \cref{eq:beta-bound-for-newton-convergence}, 
    and estimates $\{\est_a\}_{a \in [\numTerms]}$ such that $\abs*{\est_a - \angles{E_a}} \leq \beta\eps$ for all $a \in [\numTerms]$ with $\eps < \frac{1}{18}$.}
    \KwResult{Estimates $\hat{\lambda} \in [-1,1]^{\numTerms}$ such that $\abs*{\hat{\lambda}_a - \lambda_a} \leq 18\eps$ for all $a \in [\numTerms]$}
    Define $T = \Theta(\log(\frac{1}{\beta \eps \degree}))$ (see \cref{eq:Tbound} for a precise expression) and 
    $K = \lceil \log(\frac{3}{\beta\eps})\rceil$\;
    Initialize $x^{(0)} = \vec{0}$ for $x^{(0)} \in \mathbb{R}^{\numTerms}$\; \label{line:init}
    Compute all of the coefficients in the polynomials $\fun_a(\cdot)$ for all $a \in [M]$ via \cref{alg:cluster-derivative}\; \label{line:coeffs}
    \For{$t=0,1,2,\ldots,T-1$}{
        Compute $\fun(x^{(t)})$ and (the nonzero entries of) $J(x^{(t)})$\; \label{line:fun-eval}
        Compute $x^{(t+1)}$ by \cref{eq:newton}, $\vct{x}^{(t+1)} = \operatorname{Proj}_{[-1,1]^\numTerms}\Big[\vct{x}^{(t)} + \beta^{-1}\displaystyle\sum_{k=0}^{K-1} (I + \beta^{-1}J(\vct{x}^{(t)}))^k\fun(\vct{x}^{(t)})\Big]$\; \label{line:newton-iteration}
    }
    Return $\hat{\lambda} \leftarrow x^{(T)}$\;
\end{algorithm}

From this point on, we will fix the point where we truncate $\fun$ to be a particular value
\begin{equation} \label{eqn:truncM-val}
    \truncM = \Bigg\lceil \frac{e}{e-1} \frac{1}{\ln \frac 1 {\beta \tau}}\ln\Bigg(\frac{12 e^2(\degree+1)^2}{\beta \eps \ln\frac{1}{\beta \tau}}\Bigg)\Bigg\rceil,
\end{equation}
a choice that is explained in \cref{eqn:truncM-choice}.

To perform the task in the theorem statement, we use \cref{alg:newton}.
Our analysis only applies when $\eps \leq \frac{1}{12}$, but when $\eps \geq \frac{1}{18}$, we can simply output $\hat{\lambda} = \vec{0}$ as a sufficient approximation.
As in the previous section, the main idea is to find an $\vct{x} \in [-1,1]^\numTerms$ such that $\inorm{\fun(\vct{x})} = O(\beta\eps)$.
We will do this with a version of the \emph{Newton--Raphson method}.
Typically, the Newton--Raphson method performs the iteration $\vct{x}^{(t+1)} = \vct{x}^{(t)} - (J^{-1}\fun)(\vct{x}^{(t)})$ until convergence.
However, we want to avoid computing the inverse of $J$ explicitly, so we will perform the iteration
\begin{align}
    \vct{x}^{(0)} = \vec{0} \qquad \vct{x}^{(t+1)} = \operatorname{Proj}_{[-1,1]^\numTerms}\Big[\vct{x}^{(t)} + \beta^{-1}\sum_{k=0}^{K-1} (I + \beta^{-1}J(\vct{x}^{(t)}))^k\fun(\vct{x}^{(t)})\Big].\label{eq:newton}
\end{align}
This uses the Taylor series approximation for $J^{-1}$ from \cref{eq:Jinv}.
We also perform a projection to remain inside our parameter space $[-1,1]^M$,
where $\operatorname{Proj}_{[-1,1]^\numTerms}$ is the coordinate-wise application of
\begin{align}
    \operatorname{Proj}_{[-1,1]}(u) = 
    \begin{cases} 1 & \text{ if } u \in (1,\infty) \\ u & \text{ if } u \in [-1,1] \\ -1 &\text{ if } u \in (-\infty,-1) \end{cases}.
\end{align}
In \cref{alg:newton} it might seem counterintuitive that $T = \Theta(\log(\frac{1}{\beta \eps \degree}))$
decreases as $\degree$ increases when $\beta$ and $\eps$ are held constant;
however, due to \cref{eq:beta-bound-for-newton-convergence}
our algorithm is not guaranteed to work for arbitrarily large $\degree$ with $\beta$ and $\eps$ fixed. 

\paragraph{Time complexity.}
First, we will show that \cref{alg:newton} has the time complexity claimed in \Cref{thm:newton}. There are several parameters that appear in the algorithm, and it will be helpful to upper bound them with simpler expressions now. Note that $T$, the number of iterations of Newton--Raphson, and $K$, the number of terms used in the approximation of the inverse of $J$, are both clearly $O(\ln(\frac{1}{\beta\eps}))$. The other parameter, which is implicit in the definition of $\fun$ is $\truncM$, which is also $O(\ln(\frac{1}{\beta\eps}))$ due to \cref{eq:truncM-bound}.

Now let us bound the time complexity of the algorithm line by line. The first line of the algorithm with a nontrivial contribution to time complexity is \Cref{line:coeffs}. We need to compute all the coefficients in the polynomials representing $\fun_a$ for $a \in [M]$ up to truncation order $\truncM$. By \Cref{thm:expectation-value3}, we know that each polynomial $p_m$ has at most $e\degree(1+e(\degree-1))^m$ monomials, and hence the total number of monomials in $\fun_a$ is at most $C$, where $C \leq \sum_{m=1}^\truncM e\degree(1+e(\degree-1))^m \leq (e\degree)^{\truncM+1}$.
By \Cref{thm:expectation-valueA}, we can enumerate these coefficients in time $O(\degree \numTerms C)$. Then by \Cref{thm:expectation-valueB}, each coefficient can be computed exactly in time $D=(8^\truncM + L)\poly(\truncM)$. Finally there are $M$ different $\fun_a$ to be computed, and hence we can write down all of $\fun$ in $O(\degree \truncM C \numTerms D)$ time.

Then in \Cref{line:fun-eval}, we can perform evaluations of $\fun(\vct{x})$ in $O(C\numTerms\truncM)$ time, since there are $C$ monomials in each $\fun_a$, and each has up to $\truncM$ variables. Now recall that $J_{ab}(\vct{x})=\partial_b \fun_a (\vct{x})$ is a sparse matrix with at most $\degree^\truncM$ nonzero entries per row or column due to \Cref{thm:expectation-value2}. We start by setting all these entries to $0$. Then we fill out the nonzero entries of column $b$ of the matrix by enumerating the monomials of $\fun_a$, and for those monomials that contain $x_b$ (and hence will contribute to $J_{ab}(\vct{x})$), adding the contribution due to this monomial to the memory location for $J_{ab}(\vct{x})$. For a given $b\in[M]$, this takes time $O(C\truncM)$, and so we can compute $J(\vct{x})$ in $O(C\numTerms\truncM)$ time. 

Finally, in \Cref{line:newton-iteration}, we need to compute the power $(I + \frac{1}{\beta}J)^k\fun$, which can be done by starting with $\fun$ and multiplying by $I+\frac{1}{\beta}J$ $k$ times, where each matrix--vector product takes time linear in the number of nonzero entries in $I + \frac{1}{\beta}J$, which is $O(\numTerms\degree^\truncM)$.
So, the total runtime is
\begin{equation}
    O\Bigg( 
     \underbrace{\degree \truncM C\numTerms D}_{\text{\cref{line:coeffs}}} + T\Big(\underbrace{C\numTerms \truncM}_{\text{\cref{line:fun-eval}}} + \underbrace{K \numTerms\degree^\truncM}_{\text{\cref{line:newton-iteration}}}\Big)\Bigg)
    = O( 
            \numTerms \degree^2 (e\degree)^{\truncM} (8^\truncM + L)\polylog(\tfrac{1}{\beta\eps})
    ).
\end{equation}
Let us examine $\truncM$ more carefully.
Let $d := \degree +1$. 
There are two asymptotically small parameters $\eps$ and $\beta$,
and one large parameter $d$.
The inverse temperature $\beta$ is at most $\beta_c = (25 e^6 d^{10})^{-1}$ by \cref{eq:beta-bound-for-newton-convergence}.
Pulling $\truncM$ from \cref{eqn:truncM-val} (and recalling that $\tau = (1+e(d-2))2ed$ from \cref{eq:tau}), we have
\begin{align}
\truncM - 1 
&= \Bigg\lceil \frac{e}{e-1}\Big(\frac{\ln( \frac{12 e^2 d^2}{\beta \eps} \ln (\frac{1}{\beta \tau}))}{\ln(\frac{1}{\beta\tau})}\Big)\Bigg\rceil - 1\nonumber\\
&\le
\frac{e}{e-1}\Big(\frac{\ln( \frac{12 e^2 d^2}{\beta \eps} \ln (\frac{1}{\beta \tau}))}{\ln(\frac{1}{\beta\tau})}\Big) \nonumber \\
&=
\frac{e}{e-1}\Big(1 + \frac{\ln(d^2\tau) + \ln(\frac{1}{\eps}) }{\ln(\frac{1}{\beta\tau})} + \frac{\ln(12 e^2 \ln(\frac{1}{\beta\tau}))}{\ln(\frac{1}{\beta\tau})}\Big)\nonumber\\
&\leq
\frac{e}{e-1}\Big(1 + \frac{\ln(d^2\tau) + \ln(\frac{1}{\eps}) }{\ln(\frac{1}{\beta_c\tau})} + \frac{\ln(12 e^2 \ln(\frac{1}{\beta_c\tau}))}{\ln(\frac{1}{\beta_c\tau})}\Big)\nonumber\\
&=
\frac{e}{e-1}\Big(\frac{\ln(\frac{d^2}{\beta_c})}{\ln(\frac{1}{\beta_c \tau})} + \frac{\ln(\frac{1}{\eps})}{\ln(\frac{1}{\beta_c \tau})}\Big) + O\Bigl(\frac{\ln\ln d}{\ln d}\Bigr). \label{eq:truncM-bound}
\end{align}
So far, we have refrained from bounding the leading-order term (apart from taking $\beta \leq \beta_c$).
We do this now to bound the runtime.
We use that $\tau \leq 2e^2d^2$ by \cref{eq:tau}, so $\frac{1}{\beta \tau} \geq \frac{1}{\beta_c\tau} \geq 12e^4d^8$.
\begin{align}
    (8ed)^{\truncM} &= \exp\Big(\ln(8ed)\Big(1 + \frac{e}{e-1}\Big(\frac{\ln(\frac{d^2}{\beta_c})}{\ln(\frac{1}{\beta_c \tau})} + \frac{\ln(\frac{1}{\eps})}{\ln(\frac{1}{\beta_c \tau})}\Big)\Big) + O(\ln\ln d)\Big) \nonumber\\
    &\leq \exp\Big(\ln(8ed)\Big(1 + \frac{e}{e-1}\Big(\frac{\ln(25e^6d^{12})}{\ln(12e^4d^8)} + \frac{\ln(\frac{1}{\eps})}{\ln(12e^4d^8)}\Big)\Big) + O(\ln\ln d)\Big) \nonumber\\
    &\leq \exp\Big(\ln(8ed)\Big(1+\frac{e}{e-1}\frac{3}{2}\Big) + \frac{e\ln(8ed)}{(e-1)\ln(12e^4d^8)}\ln\frac{1}{\eps} + O(\ln\ln d)\Big)
\intertext{One can verify that $1 + \frac{e}{e-1}\frac{3}{2} < 3.5$ and $\frac{e}{e-1}\frac{\ln(8ed)}{\ln(12e^4d^8)} \leq \frac{e\ln(16e)}{(e-1)\ln(3072e^5)} < 0.5$ for all $d \ge 2$.
Hence,}
(8ed)^\truncM &= (8ed)^{1+\frac{3e}{2(e-1)}} e^{O(\ln\ln(d))} (1/\eps)^{0.5} = O\left( d^{3.5} (1/\eps)^{0.75} \right) = O\left(\frac{\poly(\degree)}{\eps}\right).\label{eq:8edm}
\end{align}
This leads to an upper bound on the time complexity $O\Bigl(\frac{\numTerms L}{\eps} \poly(\degree, \log\frac{1}{\beta\eps})\Bigr)$
as promised in \cref{thm:newton}.

\paragraph{Correctness and error analysis.}
We begin by explaining the choice of $\truncM$ that we stated above in \cref{eqn:truncM-val}.
We want to choose a large enough $\truncM$ so that the magnitude $\abs{\fun_a( \lambda )}$ will be small (say, at most $2\beta\eps$). 
The convergence of the $\beta$-series by \cref{thm:expectation-value} implies that for all $a$,
\begin{align}
    \abs{\fun_a( \lambda )} &\le \abs{-\est_a + \angles{E_a}} + \abs{-\angles{E_a} -\beta \lambda_a + \beta^2 p_2^{(a)}(\lambda) + \cdots + \beta^\truncM p_\truncM^{(a)}(\lambda)} \nonumber\\
    &\le \beta\eps + \sum_{m > \truncM} \beta^m c_m \nonumber\\
    &= \beta\eps + 2e^2\degree(\degree+1)\frac{(\beta \tau)^{\truncM}}{(1 - \beta \tau)^2}(\truncM(1 - \beta \tau) + 1)\nonumber\\
    &\le \beta\eps + 12 e^2 (\degree+1)^2 (\beta \tau)^{\truncM}\truncM &\text{if $\beta \tau \le \tfrac 1 2$}.\label{eq:funBound}
\end{align}
To obtain an $\truncM$ such that $\|\fun(\lambda)\|_\infty \leq 2\beta\eps$, we require
\begin{align} \label{eqn:truncM-choice}
e^{-\truncM'}\truncM' \leq \frac{\beta \eps}{12 e^2(\degree+1)^2}\ln\frac{1}{\beta \tau} \text{ where } \truncM' = \truncM\log\frac{1}{\beta \tau}.
\end{align}
Using the fact that, for $0 < b < 1$, $x = \frac{e}{e-1}\ln\frac{1}{b}$ is a solution to $xe^{-x} \leq b$,
it is enough to have $\truncM$ chosen as in \cref{eqn:truncM-val}.

Then recall that in our algorithm we wanted to apply $J^{-1}$, but settled for an approximation to make it more time efficient. There is a deviation incurred from this approximation of $J^{-1}$ in each time step $t$:
\begin{align}
    \vct{e}^{(t)} &:= \Big(J(\vct{x}^{(t)})^{-1} + \frac{1}{\beta}\sum_{k=0}^{K-1} (I + \beta^{-1}J(\vct{x}^{(t)}))^k\Big)\fun(\vct{x}^{(t)}) \nonumber\\
    &= -\frac{1}{\beta}\sum_{k=K}^\infty (I + \beta^{-1}J(\vct{x}^{(t)}))^k\fun(\vct{x}^{(t)}) \\
    &= J^{-1}(\vct{x})(I + \beta^{-1}J(\vct{x}^{(t)}))^K\fun(\vct{x}^{(t)}) \nonumber
\end{align}
From \cref{lem:jac-bound}, this error decays exponentially with $K$.

We can now begin analyzing the convergence of the Newton--Raphson method.
Consider $\fun_a(s): [0,1] \to \RR$, 
where $\fun_a(s) := \fun_a(\vct{x} + s(\lambda - \vct{x}))$, 
which is coordinate $a$ of $\fun$ along the straight-line path between $\vct{x}$ and $\lambda$.
Then using Taylor's theorem, which gives us a form for the remainder term in a Taylor series expansion, there is some $s' \in [0,1]$ such that
\begin{equation}
    \fun_a(1) = \fun_a(0) + (\diff_s \fun_a)(0) + \frac{1}{2}(\diff_s^2 \fun_a)(s').
\end{equation}
Now, we use that $\diff_s = \sum_b (\lambda_b - x_b) \diff_b$ and substitute our previous definition of $\fun_a$ to get that, for $\vct{y}^{(a)} := s'\lambda + (1-s')\vct{x}$,
\begin{align}
    \fun_a(\lambda) &= \fun_a (\vct x) + \sum_b (\lambda_b - x_b) (\underbrace{\diff_b \fun_a}_{J_{ab}})(\vct x) + \frac{1}{2} \sum_{b,c} (\lambda_b-x_b)(\lambda_c - x_c) (\diff_b \diff_c \fun_a)(\vct{y}^{(a)}).
\end{align}
Using this, we will analyze how a Newton--Raphson method iteration decreases the distance to the solution~$\lambda$.
Let $\vct{x} := \vct{x}^{(t)}$, $\vct{x}' := \vct{x}^{(t+1)}$, $\vct{e}:= \vct{e}^{(t)}$, $\Delta := \vct{x} - \lambda$, and $\Delta' := \vct{x}' - \lambda$.
\begin{align}
    \abs{\Delta_d'} &= \abs{\operatorname{Proj}_{[-1,1]}[(\vct{x} - (J^{-1}\fun)(\vct{x}) + \vct{e})_d] - \lambda_d} 
    \label{eq:delta-explicit}\\
    &\leq \abs*{(\vct{x} - (J^{-1}\fun)(\vct{x}) + \vct{e})_d - \lambda_d} 
    \nonumber\\
    &= \abs*{e_d + \Delta_d - \sum_a (J(\vct{x})^{-1})_{da}(\fun(\vct{x}))_a} 
    \nonumber\\
    &= \abs*{e_d + \Delta_d - \sum_a J(\vct{x})_{da}^{-1}\Big(\fun_a(\lambda) - \sum_b (\lambda_b - x_b) J(\vct{x})_{ab} - \frac{1}{2} \sum_{b,c} (\lambda_b-x_b)(\lambda_c - x_c) [\diff_b \diff_c \fun_a](\vct{y}^{(a)})\Big)} 
    \nonumber\\
    &= \abs*{\Big[\vct{e} + \Delta -J(\vct{x})^{-1}\fun(\lambda) - J(\vct{x})^{-1}J(\vct{x})\Delta\Big]_d 
        + \frac{1}{2} \sum_{a,b,c} J(\vct{x})_{da}^{-1}\Delta_b\Delta_c [\diff_b \diff_c \fun_a](\vct{y}^{(a)})} 
    \nonumber\\
    &= \abs*{\Big[J(\vct{x})^{-1}\Big((I + \beta^{-1}J(\vct{x}))^K\fun(\vct{x}) - \fun(\lambda)\Big)\Big]_d 
        + \frac{1}{2} \sum_{a,b,c} J(\vct{x})_{da}^{-1}\Delta_b\Delta_c [\diff_b \diff_c \fun_a](\vct{y}^{(a)})} 
    \nonumber
\end{align}
We will bound each expression above in turn.
We can bound the first expression using \cref{lem:jac-bound,eq:funBound}, and that $K = \lceil\log_2(\frac{3}{\beta\eps})\rceil$:
\begin{align}
    &\abs*{\Big[J(\vct{x})^{-1}\Big((I + \beta^{-1}J(\vct{x}))^K\fun(\vct{x}) - \fun(\lambda)\Big)\Big]_d} \nonumber \\
    &\leq \iinorm{J(\vct{x})^{-1}}\Big(\iinorm{I + \beta^{-1}J(\vct{x})}^K\inorm{\fun(\vct{x})} + \inorm{\fun(\lambda)}\Big) \nonumber \\
    &\leq 2\beta^{-1}\Big(2^{-K}(2+\beta \eps) + 2\beta\eps\Big) 
    \leq 6\eps. 
\end{align}
The second expression can be bounded through an argument similar to that of \cref{lem:jac-bound},
in particular, that $\fun_a(\vct y)$ decomposes into degree-$k$ polynomials~\mbox{$p_k^{(a)}(\vct{y})$} 
that depend only on~$y_b$ where~$b$ are within $\graph$-distance~$k$ from~$a$ 
and that have a bound on the magnitude of the coefficients (given by $c_k$ defined in \cref{eq:cm}).
We have that for all~$d$,
\begin{align}
\abs*{\frac{1}{2} \sum_{a,b,c} J(\vct{x})_{da}^{-1}\Delta_b\Delta_c [\diff_b \diff_c \fun_a](\vct{y}^{(a)})}
&\leq \frac{1}{2} \iinorm{J(\vct{x})^{-1}}\max_a\abs{\sum_{b,c} \Delta_b\Delta_c [\diff_b \diff_c \fun_a](\vct{y}^{(a)})} 
\\
&\leq \frac1\beta\max_a\sum_{k \geq 0}\sum_{b,c} \abs{\Delta_b \Delta_c} \cdot \beta^k \cdot \abs{\diff_b \diff_c p_k^{(a)}(\vct y)} 
\nonumber\\
&\leq \frac1\beta\max_a \sum_{k\ge 0} \sum_{\substack{b,c : \\ \dist(b,a) \le k \\ \dist(c,a) \le k}}  \inorm{\Delta}^2 \cdot \beta^k \cdot k (k-1)c_k 
\nonumber\\
&\leq \frac 1\beta  \sum_{k\ge 0} (\degree+1)^{2k} \inorm{\Delta}^2 \cdot \beta^k \cdot k (k-1) c_k 
\nonumber \\
&= \frac {12 e^2}{\beta} \inorm{\Delta}^2 (\degree+1)^2 \frac{(\beta(\degree+1)^2 \tau)^2}{(1-\beta (\degree+1)^2 \tau)^4} \tag*{if $\beta(\degree+1)^2 \tau < 1$}
\nonumber\\
&\leq 12.5 e^2 \beta (\degree+1)^6 \tau^2 \inorm{\Delta}^2 \tag*{if $\beta\degree^2 \tau \le 1 - \sqrt[4]{\frac{12}{12.5}}$.}
\nonumber
\end{align}
These two computations, together with \cref{eq:delta-explicit}, gives us our bound on $\inorm{\Delta'}$.
\begin{align}
    \inorm{\Delta'}
    &\leq 6\eps + 12.5 e^2 \beta (\degree+1)^6 \tau^2 \inorm{\Delta}^2
\end{align}
To summarize, we have just shown that for the Newton--Raphson method iteration shown in \cref{eq:newton}, the error decays as
\begin{align}
    \inorm{\Delta_{t+1}} &\leq 6\eps + 12.5 e^2 \beta (\degree+1)^6\tau^2\inorm{\Delta_t}^2. \label{eq:ddFBound}
\end{align}
We can solve this recursion: By \cref{lem:iteration} below, 
provided that 
$75 \eps e^2\beta(\degree+1)^6 \tau^2 \leq \frac14$ and 
$\inorm{\Delta_0} \leq \frac{1}{25 e^2\beta (\degree+1)^6 \tau^2}$, 
we have that $\inorm{\vct{x}_t - \lambda} \le 18\eps$ 
after $T$ iterations where
\begin{align}
    T &= \lceil-\log_2(75e^2 (\degree+1)^6 \tau^2 \beta\eps )\rceil \nonumber\\
    &\le \lceil-\log_2(300e^6(\degree+1)^{10} \beta\eps )\rceil .\label{eq:Tbound}
\end{align}
Since $\inorm{\Delta_0} \le 1$, 
the condition is satisfied when $\beta \leq (25e^2(\degree+1)^6 \tau^2)^{-1}$ and $\eps \leq \frac{1}{12}$.
This completes the proof of \cref{thm:newton}.

\begin{lemma}\label{lem:iteration}
    Let $c,d \in \RR_{>0}$ be such that $cd \leq \frac 1 4$.
    Consider a sequence of positive real numbers $z_0, z_1, z_2,\ldots $ that satisfy for all $n \ge 0$,
\begin{equation}
    z_0 \le \frac{1}{2d}  \quad \text{and}\quad  z_{n+1} \le c + d z_n^2.
\end{equation}
Then, for all $n \ge \log_2 \frac 1 {cd} - 1$ it holds that $z_n \le 3c$.
\end{lemma}
\begin{proof}
With $y_n := dz_n$, the recursion is
\begin{align}
    y_0 \le \frac{1}{2} \quad \text{and}\quad y_{n+1} \le cd + y_n^2.
\end{align}
Note that by induction, $y_n \leq \frac{1}{2}$ for all $n \geq 0$.
So, $\{y_n\}$ also satisfies the inequality $y_{n+1} \le cd + \frac{1}{2}y_n$.
Unrolling the iteration, we get that
\begin{align}
    dz_n = y_n \leq cd\Big(1 + \frac{1}{2} + \cdots + \frac{1}{2^{n-1}}\Big) + \frac{1}{2^n}y_0 \leq 2cd + \frac{1}{2^{n+1}}.
\end{align}
So, $z_n \leq 3c$ when $n \geq \log_2\frac{1}{cd}-1$.
\end{proof}

\section{Lower bounds}

In this section we establish the lower bounds claimed in \Cref{thm:lowerbound}, starting with the lower bound for Hamiltonian learning with $\ell_\infty$ error $\eps$. We  then build on that argument to obtain the lower bound with $\ell_2$ error $\eps$. 

\subsection{Warmup for constant \texorpdfstring{$N$}{N}}\label{sec:warmup}

As a warmup, let's establish a lower bound for $\ell_\infty$ error for Hamiltonians
on a constant number of qubits. In this case we want to show a lower bound of $\Omega(\exp(2\beta)/\beta^2\eps^2)$ samples for any $\beta$ and $\eps \in (0,1/2]$.

Consider two diagonal Hamiltonians $H_0$ and $H_1$ for a 2-qubit system (or a single qudit with local dimension $4$) expressed in terms of the Pauli matrices 
\begin{equation}
    Z \otimes I=\left(\begin{smallmatrix} +1 \\ & +1\\ & & -1\\ & & & -1\end{smallmatrix}\right), \quad
    I \otimes Z=\left(\begin{smallmatrix} +1 \\ & -1\\ & & +1\\ & & & -1\end{smallmatrix}\right), \quad       \text{and} \quad
    Z \otimes Z=\left(\begin{smallmatrix} +1 \\ & -1\\ & & -1\\ & & & +1\end{smallmatrix}\right). 
\end{equation}
For any $\eps\in(0,1/2]$, we define
\begin{align}
    H_0 = \left(-1\right) Z \otimes I + \left(-\frac{1}{2}\right) I \otimes Z + \left(-\frac{1}{2}\right) Z \otimes Z 
    &= \begin{pmatrix}
        -2 \\
        & 0 \\
        & & 1\\
        & & & 1
    \end{pmatrix}, \quad \text{and}\\ \label{eq:Hone}
    H_1 = \left(-1\right) Z \otimes I + \left(-\frac{1}{2}+{\eps}\right) I \otimes Z + \left(-\frac{1}{2}-{\eps}\right) Z \otimes Z
    &= \begin{pmatrix}
        -2 \\
        & 0 \\
        & & 1+2\eps \\
        & & & 1-2\eps
    \end{pmatrix}.
\end{align}
The coefficients of these Hamiltonians lie in $[-1,1]$, and if we learn an unknown Hamiltonian to error $<\eps/2$, then we can distinguish these two Hamiltonians. We now show that distinguishing the Gibbs states of these Hamiltonians needs $\Omega(\exp(2\beta)/\beta^2\eps^2)$ samples. 

Since the Hamiltonians are diagonal, their Gibbs states are also diagonal and are simply probability distributions. The two Gibbs states $\rho_0$ and $\rho_1$ are 
\begin{align}\label{eq:rho}
    \rho_0 = \frac{1}{Z_0}\begin{pmatrix}
        e^{2\beta} \\
        & 1 \\
        & & e^{-\beta} \\
        & & & e^{-\beta} \\
    \end{pmatrix} \quad \text{and} \quad
    \rho_1 = \frac{1}{Z_1} \begin{pmatrix}
        e^{2\beta} \\
        & 1 \\
        & & e^{-\beta+2\beta\eps} \\
        & & & e^{-\beta-2\beta\eps} \\
    \end{pmatrix},
\end{align}
where $Z_0$ and $Z_1$ are the respective partition functions (or normalization constants).

We want to lower bound the number of samples needed to distinguish the two probability distributions corresponding to $\rho_0$ and $\rho_1$, which we can call $q_0$ and $q_1$. 
The problem of distinguishing probability distributions given samples is called hypothesis testing, and its complexity is well understood.

One way to lower bound the number of samples needed is via the KL divergence between these distributions, which is defined as follows for two distributions $p$ and $q$:
\begin{equation}
    \D(p\;\|\;q)=\sum_j p_j \log \left(\frac{p_j}{q_j}\right).
\end{equation}

\begin{lemma}\label{lem:DKL}
    For the probability distributions $q_0$ and $q_1$ corresponding to $\rho_0$ and $\rho_1$ in \cref{eq:rho}, we have $\D(q_1 \;\|\; q_0) \leq 8\beta^2 \eps^2e^{-3\beta+2\beta\eps}$. When $\eps \leq 1/2$, we have $\D(q_1 \;\|\; q_0) \leq 8\beta^2 \eps^2 e^{-2\beta}$.
\end{lemma}
\begin{proof}This follows from a straightforward calculation.
    \begin{align}
        \D(q_1 \;\|\; q_0) 
        &= \frac{e^{2\beta}\log\frac{e^{2\beta}}{e^{2\beta}} + 1\log\frac{1}{1} + e^{-\beta+2\beta\eps}\log\frac{e^{-\beta+2\beta\eps}}{e^{-\beta}} + e^{-\beta-2\beta\eps}\log\frac{e^{-\beta-2\beta\eps}}{e^{-\beta}}}{Z_1} + \log\frac{Z_0}{Z_1} \nonumber\\
        &= \frac{e^{-\beta+2\beta\eps}2\beta\eps - e^{-\beta-2\beta\eps}2\beta\eps}{Z_1} + \log\frac{Z_0}{Z_1} \nonumber\\
        &= \frac{2\beta\eps e^{-\beta+2\beta\eps}(1 - e^{-4\beta\eps})}{e^{2\beta} + 1 + e^{-\beta+2\beta\eps} + e^{-\beta-2\beta\eps}} - \log\frac{Z_1}{Z_0} \nonumber\\
        &= \frac{2\beta\eps e^{-\beta+2\beta\eps}(1 - e^{-4\beta\eps})}{e^{2\beta} + 1 + e^{-\beta+2\beta\eps} + e^{-\beta-2\beta\eps}} 
        - \log\frac{e^{2\beta} + 1 + e^{-\beta-2\beta\eps} + e^{-\beta + 2\beta\eps}}{e^{2\beta} + 1 + 2e^{-\beta}}.
    \intertext{Now using the inequality $\log(x) \geq 1- 1/x = (x-1)/x$, which holds for $x>0$, we get}
        &\leq \frac{2\beta\eps e^{-\beta+2\beta\eps}(1 - e^{-4\beta\eps})}{e^{2\beta} + 1 + e^{-\beta+2\beta\eps} + e^{-\beta-2\beta\eps}} 
        - \frac{e^{-\beta-2\beta\eps} + e^{-\beta + 2\beta\eps} - 2e^{-\beta}}{e^{2\beta} + 1 + e^{-\beta-2\beta\eps} + e^{-\beta + 2\beta\eps}}. \\
    \intertext{The denominators in this expression are $\geq e^{2\beta}$, so we can continue}
        &\leq e^{-2\beta}\Big(2\beta\eps e^{-\beta+2\beta\eps}(1 - e^{-4\beta\eps}) - (e^{-\beta-2\beta\eps} + e^{-\beta + 2\beta\eps} - 2e^{-\beta})\Big) \nonumber\\
        &= e^{-3\beta+2\beta\eps}\Big(2\beta\eps(1 - e^{-4\beta\eps}) - (1 - e^{-2\beta\eps})^2\Big) \nonumber\\
        &\leq e^{-3\beta+2\beta\eps}2\beta\eps(1 - e^{-4\beta\eps}) \nonumber\\
        &= e^{-3\beta+2\beta\eps}2\beta\eps(1-e^{-2\beta\eps})(1 + e^{-2\beta\eps}) \nonumber\\
        &\leq e^{-3\beta+2\beta\eps}4\beta\eps(1-e^{-2\beta\eps})
        \leq  e^{-3\beta+2\beta\eps}(4\beta\eps)(2\beta\eps) \nonumber\\
        &= 8\beta^2 \eps^2e^{-3\beta+2\beta\eps},
    \end{align}
    where the last inequality used $1-e^{-x} \leq x$, which holds for $x>0$.
\end{proof}

The number of samples needed to distinguish the two probability distributions is lower bounded by the inverse of the KL divergence between the two, as we make precise in the next section, which gives us the desired lower bound for constant $N$.

\subsection{Lower bound for \texorpdfstring{$\ell_\infty$}{ell infty} error}

To prove the general lower bound for non-constant $N$, we will need Fano's lemma, and specifically we use the version in \cite[Cor.~2.6]{Tsy09}:
\begin{lemma}[Fano's lemma]
    For any $N \geq 2$, let $P_0,P_1,\ldots,P_N$ be probability distributions that satisfy
    \begin{align}
        \frac{1}{N+1}\sum_{j=1}^N \D(P_j\;\|\;P_0) \leq \alpha
    \end{align}
    for some $\alpha \in (0, \log N)$.
    Then if $p_\mathrm{error}$ denotes the minimax error of the hypothesis testing problem, or the worst-case error of distinguishing the different distributions by the best strategy, we have 
    \begin{align}
        p_\mathrm{error}
        \geq \frac{\log(N+1) - \log(2) - \alpha}{\log(N)}
        \geq 1 - \frac{\log(2) + \alpha}{\log(N)}.
    \end{align}
\end{lemma}

We're now ready to establish a more precise version of \Cref{thm:lowerbound} for $\ell_\infty$ error $\eps$.

\begin{theorem}\label{thm:lowerboundinfty}
    For any $\eps \in (0,1/2]$, $\beta>0$, $\delta>0$, and $\numQubits$, there exists a 2-local Hamiltonian on $2\numQubits$ qubits such that the sample complexity of learning its coefficients to $\ell_\infty$ error $\eps$ with probability at least $1-\delta$ is $\Omega \left( \frac{\exp(2\beta)}{\beta^2\eps^2} \log \left(\frac{\numQubits}{\delta}\right) \right)$.
\end{theorem}

\begin{proof}
We divide our $2N$ qubits into $N$ pairs and consider Hamiltonians that are either $H_0$ or $H_1$, as defined in \Cref{sec:warmup}, on each pair. We consider $N+1$ possible Hamiltonians, corresponding to all pairs having Hamiltonian $H_0$, or all but one pair having Hamiltonian $H_0$ and one pair of qubits having Hamiltonian $H_1$. So the potential Gibbs states produced are of the form $\rho_0 \otimes \cdots \otimes \rho_0$ or $\rho_0 \otimes \ldots \otimes \rho_0 \otimes  \rho_1 \otimes  \rho_0 \otimes  \ldots \otimes  \rho_0$ where $\rho_1$ is the $i$th copy for $i \in [N]$. As noted, learning the Hamiltonian to $\ell_\infty$ error $\eps/2$ allows us to distinguish all these distributions. We will show that these distributions are hard to distinguish unless we have enough samples.

Consider the problem of distinguishing between $N+1$ distributions $P_0,P_1,\ldots,P_N$, where each $P_i$ is $S$ independent copies of a distribution $p_i$ over $N$ qubits. These distributions $p_i$ correspond to the $2N$-bit probability distributions that we get from the density matrix that has $\rho_0$ on all qubits and $\rho_1$ on the $i$th qubit. Since these are diagonal density matrices, we'll just think of them as probability distributions over $2SN$ bits. 

To employ Fano's lemma, we need to bound $\D(P_j\;\|\;P_0)$. Since each $P_i$ is simply $S$ copies of a distribution $p_i$, we have $\D(P_j\;\|\;P_0) = S \D(p_j\;\|\; p_0)$ due to the chain rule for KL divergence. Any $p_j$ with $j\neq 0$ and $p_0$ only differ at one site, where one distribution is $q_1$ and the other is $q_0$, so by the chain rule again we have 
$\D(p_j\;\|\; p_0) = \D(q_1\;\|\; q_0)$, which we have already computed in \Cref{lem:DKL}. Thus we have that 
\begin{equation}
    \D(P_j\;\|\;P_0) = S \D(p_j\;\|\; p_0) = S \D(q_1\;\|\; q_0) \leq 8S \beta^2 \eps^2 e^{-2\beta}.
\end{equation}
Then we can take $\alpha$ to be this value and apply Fano's lemma to get
\begin{align}
    p_\mathrm{error} &\geq 1 - \frac{\log(2) + \alpha}{\log(N)} \nonumber\\
    &= 1 - \frac{\log(2) +  8S \beta^2 \eps^2 e^{-2\beta}}{\log(N)}.
\end{align}
This error can be a small constant only if $S = \Omega\left(\frac{e^{2\beta}\log(N)}{\beta^2\eps^2}\right)$. This gives us the lower bound for constant $\delta$.

To get the $\delta$ dependence, we show a reduction to this case. Assume there is an algorithm that can solve the Hamiltonian learning problem with error $\delta$ on the above instance on $N$ qubits using $T$ samples. Let's use the same algorithm to solve the hard instance we constructed above on $N/(3\delta)$ qubits with probability $2/3$. For this problem we already have a lower bound of $\Omega \left( \frac{\exp(2\beta)}{\beta^2\eps^2} \log \left(\frac{\numQubits}{\delta}\right) \right)$, which we get by replacing $N$ by $N/(3\delta)$ in our previous lower bound.

We can split this problem up into $1/(3\delta)$ instances of size $N$, and apply the assumed algorithm that solves $N$-size instances with error $\delta$. This algorithm needs $T$ samples of the each of the $1/(3\delta)$ $N$-qubit Hamiltonians, but each sample of the $N/(3\delta)$-qubit Hamiltonian provides one sample each for the $N$-qubit Hamiltonians. So the sample complexity of our new algorithm remains $T$. Finally, this algorithm learns all $1/(3\delta)$ $N$-qubit Hamiltonians with error probability at most $\delta$ per instance. So by the union bound, it correctly learns all $1/(3\delta)$ instances with error at most $1/3$. Thus our assumed algorithm solves the Hamiltonian learning problem on $N/(3\delta)$ qubits and hence must use $\Omega \left( \frac{\exp(2\beta)}{\beta^2\eps^2} \log \left(\frac{\numQubits}{\delta}\right) \right)$ samples.
\end{proof}

\begin{remark}
The lower bound above even applies to a slightly more general learning setting where
we can choose different $\beta$ (inverse temperature) for different samples.
This scenario may arise in a physical situation where one wishes 
to examine the temperature dependence of some observable's expectation value to learn the Hamiltonian.
In this case, the probability distributions~$P_j$ that we will distinguish 
is a product of $p_j$ at possibly different temperatures.
The KL divergence can be upper bounded similarly, 
and in the application of Fano's lemma we can take the maximum of the upper bounds on the KL divergence.
If the temperatures are chosen nonadaptively,
i.e., $\beta$ are chosen beforehand and the samples are prepared for us accordingly,
then the sample complexity is $\Omega( \min_\beta \frac{e^{2\beta}}{\beta^2 \eps^2} \log \numQubits )$
for a constant probability of success,
where the minimum is taken over $\beta$ that are used in the samples. 
\end{remark}

\subsection{Lower bound for \texorpdfstring{$\ell_2$}{ell 2} error}

Our lower bound for $\ell_2$ error $\eps'$ builds on the previous construction. Let us use $\eps'$ to denote the $\ell_2$ error and reserve $\eps$ to be the parameter that appears in the definition of $H_1$ in \cref{eq:Hone}.

\begin{theorem}\label{thm:lowerboundtwo}
    For any $\beta > 0$, $\numQubits$, and $\eps' \in (0,0.01 \sqrt{N}]$, there exists a 2-local Hamiltonian on $2\numQubits$ qubits such that the sample complexity of learning its coefficients to $\ell_2$ error $\eps'$ with probability $\geq \frac{2}{3}$ is $\Omega \left( \frac{\exp(2\beta)}{\beta^2\eps'^2} \numQubits \right)$.
\end{theorem}

\begin{proof}
To get this result for $\ell_2$ error $\eps'$, we consider a different collection of Hamiltonians. Consider an error correcting code $C$ over $N$ bits that encodes $\Omega(N)$ logical bits and has code distance at least $N/10$.\footnote{In other words, let $C$ be a set of $2^{\Omega(N)}$ length-$N$ bitstrings such that any two elements of $C$ differ in $N/10$ bits.} We know such codes exist that achieve the Hamming bound and have size $|C| = \Theta\left(\frac{2^N}{\sum_{t=0}^{0.05N-1} \binom{N}{k}}\right) \gtrsim \frac{2^{cN}}{\sqrt{N}} = 2^{\Omega(N)}$.

Consider a Hamiltonian on $2N$ qubits that is specified by a codeword $x\in \{0,1\}^N$. We divide the $2N$ qubits into pairs again and each pair will have Hamiltonian either $H_0$ or $H_1$ as before. Recall that $H_1$ depends on a parameter $\eps$, which will be different from $\eps'$ and will be chosen later. 

The Hamiltonian for the first pair of qubits is $H_{x_1}$, for the next pair is $H_{x_2}$ and so on. So just as before, the Gibbs state will be $\rho_{x_1} \otimes \rho_{x_2} \otimes \cdots \otimes \rho_{x_N}$, and as before, these are diagonal states, so the resulting probability distributions will be $q_{x_1} \otimes q_{x_2} \otimes \cdots \otimes q_{x_N}$. 

Just like before, we want to show that identifying the Hamiltonian (with probability $\geq 2/3$), which is equivalent to identifying the codeword $x$ from which the Hamiltonian was constructed, requires many samples.
We claim that if we learn the Hamiltonian to $\ell_2$ error $\eps' = 0.01 \sqrt{N} \eps$, then we can exactly identify the string $x$ (with probability $\geq 2/3$). This step converts learning with $\ell_2$ error to exact identification and this conversion dictates the value of $\eps$ in our definition of $H_1$.

Consider the unknown Hamiltonian on the first pair of qubits, $H_{x_1}$. This has two unknown coefficients, which are $-1/2+\eps x_0$ and $-1/2-\eps x_0$. Let's only consider the problem of learning the first of these coefficients for all our Hamiltonians $H_{x_i}$. Now if we have learned the coefficients to $\ell_2$ error $\eps'$, it means we have a string $\lambda_i$ that satisfies $\sqrt{\sum_i (\lambda_i + 1/2 - \eps x_0)^2} \leq \eps'$. By setting $y_i = \eps(\lambda_i + 1/2)$, this means we have a string $y \in \mathbb{R}^N$ that satisfies $\sqrt{\sum_i (y_i-x_i)^2} \leq \eps'/\eps$. 

We now use the property that $x$ is a codeword of an error correcting code with large distance, so we can identify $x$ given a close enough $y$. We know that any two codewords are at least $N/10$ apart in Hamming distance. This means any two codewords are at least $\sqrt{N/10}$ apart in $\ell_2$ distance. Hence if we have a point in $\mathbb{R}^N$ (not just on the Boolean hypercube) that is $\ell_2$ distance strictly less than $\sqrt{N/10}/2$ from a codeword $x$, it can be uniquely decoded to $x$. So, if we have a string $y \in \mathbb{R}^N$ such that $\sqrt{\sum_i (y_i-x_i)^2} \leq 0.01 \sqrt{N}$, that will suffice. Thus we can choose $\eps$ to satisfy $\eps' = 0.01 \sqrt{N} \eps$.

Now that we know that solving the Hamiltonian learning task allows us to exactly distinguish this set of Hamiltonians, let's show that distinguishing the Gibbs states of this set of Hamiltonians requires many samples using Fano's lemma.

To employ Fano's lemma, we need to bound the pairwise KL divergences again. We now consider $2^{\Omega(N)}$ probability distributions $p_x$, each corresponding to the Gibbs state of the Hamiltonian constructed from a codeword $x \in C$. Without loss of generality let us assume that $x=0^N$ is part of the code, and let $p_0$ refer to the distribution corresponding to this Hamiltonian. As before, we let $P_x$ be $S$ copies of $p_x$.
For any codeword $x \neq 0^N$, $\D(P_x\;\|\;P_0) = S \D(p_x\;\|\;p_0) \leq S N \D(q_1\;\|\;q_0) \leq 8 SN \beta^2 \eps^2 e^{-2\beta}$ using the chain rule and \cref{lem:DKL}.
So we can choose the parameter $\alpha$ in Fano's lemma to be $8 SN \beta^2 \eps^2 e^{-2\beta} = O(S \beta^2 \eps'^2 e^{-2\beta})$.
Applying Fano's inequality, we get
\begin{align}
    p_\mathrm{error}  = 1 - \frac{\log(2) + O(S(e^{-2\beta}\beta^2\eps'^2))}{\log(2^{\Omega(N)})},
\end{align}
which can be a small constant only if $S =\Omega\left(\frac{\exp(2\beta) N}{\beta^2\eps'^2}\right)$.
\end{proof}

\section{Discussion}

In this paper, we have addressed the Hamiltonian learning problem in a high-temperature regime.
We have analyzed an algorithm to show that it has optimal sample complexity and time complexity.
We were able to claim time optimality because our time complexity is simply linear in the sample size,
the number of qubits in the total of all samples used in the algorithm.
The critical temperature above which our algorithm is guaranteed to work
depends only on the degree of the dual interaction graph,
which we have treated as a constant in the optimality claims for sample and time complexity.

Although our algorithm is optimal for any fixed $\degree$,
it might be possible to enlarge the temperature domain where our method works.
The critical temperature to guarantee the convergence of the Newton--Raphson method
is higher than that to ensure the convergence of the $\beta$-series expansion of $\langle E_a \rangle$;
the former is $O(\degree^{10})$ (\cref{thm:newton}) while the latter is $O(\degree^2)$ (\cref{thm:expectation-value}).
This rather large discrepancy occurred when we used the ``band-diagonal'' property of the Jacobian of $\fun$,
and it will require finer understanding of these correlations
to improve our bounds in terms of $\degree$.
It is also feasible to extend our algorithm beyond low-intersection Hamiltonians to local Hamiltonians, where $\degree$ need not be constant, since cluster expansion works in the more general setting where one-spin energy is bounded.
However, the number of monomials still scales exponentially in $\degree$, so even writing down the truncated Taylor series expansion could be computationally expensive.
\todo{Check this statement; since a reviewer asked, it might be worth asserting that, say, the strong convexity bound likely still holds in this slightly more general setting studied in \cite{kkb20})?}

The problem of finding an efficient learning algorithm in the low-temperature regime remains completely open.
Our high-temperature expansion does not converge in general for large $\beta$
since there are systems that undergo phase transitions as we lower the temperature,
where the partition function is not analytically continued from the high-temperature domain.
In fact, an efficient algorithm for all temperatures, if it exists, 
should not attempt to evaluate partition functions 
since low-temperature partition functions are generally (at least) NP-hard to compute.
The classical polynomial-time algorithms avoid evaluating partition functions using conditional independence (the Markov property), but this does not hold in general for quantum noncommuting Hamiltonians.

\section*{Acknowledgements}

R.K.\ and E.T.\ thank Marcus Silva for early discussions about this problem. R.K.\ thanks Vamsi Pritham Pingali for many helpful discussions about this problem and multivariable calculus. E.T.\ thanks Anurag Anshu for the question about strong convexity of the log-partition function and Adam Klivans for discussions about the state of the art in learning classical Hamiltonians. We also thank Hsin-Yuan Huang for raising the question of learning a Hamiltonian from its real-time evolution.

\appendix

\section{Learning Hamiltonians from real time dynamics}
\label{sec:realtime}

Suppose we are given a blackbox that implements unitary time evolution
\begin{align}
U = e^{-i t H}
\end{align}
governed by a fixed, time-independent, unknown Hamiltonian $H = \{ (a, E_a, \lambda_a) : a \in [M] \}$.
We assume that the evolution time~$t$ is known to us,
and the Hamiltonian follows the same normalization as in the main text:
$E_a$ are distinct Pauli matrices and $\lambda_a \in [-1,1]$ for all~$a$.
The blackbox converts any input state represented by a density matrix~$\rho$ to $U \rho U^\dagger$.
Now, the learning problem is to estimate~$\lambda_a$ to additive accuracy~$\varepsilon$
with as few uses of blackbox~$U$ as possible.

We consider a scenario where~$t$ is smaller than some constant $t_c$ that only depends on
the structure of Hamiltonian terms~$E_a$ but not on the coefficients~$\lambda_a$.
The learning algorithm and its analysis will be very similar to that in the main text,
so we will be brief.
We restate the theorem we will prove.

\algodynamics*

\begin{remark}\label{rem:tlinear}
In what follows below, we prove \cref{thm:algo-dynamics} with $t$ replaced by $t^2$.
We can improve this dependence from $1/t^2$ to $1/t$ by reducing to a setting where $t$ is constant: by applying $U$ $n = \lfloor t_c / t \rfloor$ times, we can produce a black box for the unitary $V = U^n = \exp(-iHt\lfloor t_c/t\rfloor)$.
Learning parameters from $V$ is the same problem as learning parameters from $U$, except the parameter $t$ becomes $t \lfloor t_c / t \rfloor = \Theta(t_c)$, which is constant (determined only by $\degree$).
So, we can use the algorithm below with the unitary $V$, requiring $O(\frac{1}{\eps^2}\log\frac{\numQubits}{\delta})$ applications of $V$, and therefore $O(\frac{1}{t\eps^2}\log\frac{\numQubits}{\delta})$ applications of $U$.
The time complexity is also inflated by $n = \Theta(1/t)$ in a similar fashion.

Note that the complexity in the time parameter is optimal;
two time-evolution operators $I$ and $e^{-it Z}$ on one qubit differ by $O(t)$ 
in operator norm and hence also in completely bounded (diamond) norm as quantum channels.
\end{remark}

\subsection{Series expansion of time-evolved operators}

Similarly to \cref{thm:expectation-value} for the $\beta$-series expansion of $\Tr( E_a e^{-\beta H} ) / \Tr e^{-\beta H}$, in this section we prove properties about the $t$-series expansion of~$U P U^\dagger$ where~$P$ is a single-qubit Pauli operator.
The relevance of this quantity to the learning problem will be evident in the next subsection.
In the following theorem, the pink text indicates where it differs from \cref{thm:expectation-value}; morally, the same properties are proven, just for a different series.
All of the quantitative bounds are at least as strong as those in \cref{thm:expectation-value}, and though we prove results for a matrix-valued polynomial, they are indeed comparable when considering their trace against the operator $Q$, as defined in $\fun(Q, P)$.

\begin{theorem}\label{thm:real-time}
    Consider a Hamiltonian $\{(a,E_a,\lambda_a):a\in[\numTerms]\}$. 
    Then, for every single-qubit Pauli operator\todo{ewin: generalize?} $P$ and $L$-qubit Pauli operator $Q$, we have a Taylor series expansion
    \begin{align}
        \textdiff{UPU^\dagger}
        &\textdiff{= \sum_{m = 1}^\infty t^mq_m(\lambda_1,\ldots,\lambda_M)} \label{eq:real-time} \\
        \textdiff{\fun(Q,P)} &\textdiff{:= \frac 1 \fulldimension \Tr( Q UPU^\dagger) = \frac{1}{\fulldimension} \sum_{m=1}^\infty \Tr(Q q_m(\lambda))}
    \end{align}
    where equality holds whenever the series converges absolutely.
    For any $m \in \ZZ_{>0}$, the following hold:
    \begin{enumerate}[ref=\thetheorem (\arabic*)]
        \item $q_m \in \textdiff{\CC^{\fulldimension \times \fulldimension}}[\lambda_1,\ldots,\lambda_\numTerms]$ is a degree $m$ homogeneous \textdiff{matrix-valued} polynomial in the Hamiltonian term coefficients.\label{thm:real-time1}
        \item Let~$\graph(P)$ denote the dual interaction graph among operators $\{E_1,\ldots,E_\numTerms, P\}$,
        i.e., $\graph(P)$ is $\graph$ with an extra node~$p$ and an extra edge $(p,a)$ if and only if $\Supp(E_a) \cap \Supp(P) \neq \varnothing$.
        Then $q_m$ involves $\lambda_a$ only if the distance between $p$ and $a$ on $\textdiff{\graph(P)}$, $\dist_{\graph(P)}(p,a)$, is at most $m$.\label{thm:real-time2}
        \item $q_m$ consists of at most $\textdiff{\max(L, \degree)e\degree(1+e(\degree-1))^{m-1}} \leq e\degree(1+e(\degree-1))^{m}$ monomials.\label{thm:real-time3}
        \item The coefficient matrix in front of any monomial of $q_m$ has \textdiff{spectral norm} at most $\textdiff{2^m}$ in magnitude.\label{thm:real-time4}
    \end{enumerate}
    Suppose further that every $E_a$ is a tensor product of Pauli matrices, supported on at most $L$ qubits.
    Then, after $O(L\numTerms\degree\log\degree)$ pre-processing time (see \cref{rm:ham-repr}), the following are true for every $m \in \ZZ_{> 0}$.
    \begin{enumerate}[label=\Alph*.,ref=\thetheorem (\Alph*)]
    \item The list of monomials that appear in $q_m$ can be enumerated 
    in time $O(m \degree C)$, where $C$ is the number of monomials (so, in particular, in time $O(m \degree^2 (1+e(\degree - 1))^m)$). \label{thm:real-timeA}
    \item \textdiff{The truncated series of $\fun(Q,P)$, $\frac{1}{\fulldimension} \sum_{\ell=1}^m \Tr(Q q_\ell(\lambda))$, can be computed exactly as a rational polynomial in $C(4^m + L)\poly(m)$ time.}\todo{ewin: todo give a more specific runtime.} \label{thm:real-timeB}
    \end{enumerate}
\end{theorem}

\noindent To understand $UPU^\dagger$, we recall the well-known formula for square matrices $A$ and $B$,
\begin{align}
	e^A B e^{-A} &= \sum_{n=0}^\infty \frac{[A,B]_n}{n!} \\
	\text{where } [A,B]_k &= 
	\begin{cases}
		B &(k = 0)\\
		A[A,B]_{k-1} - [A,B]_{k-1} A & (k \ge 1).
	\end{cases} \nonumber
\end{align}
Since the nested commutator $[A,B]_n$ has norm upper bounded by $2^n \norm{A}^n \norm{B}$,
which grows only exponentially with~$n$, 
this series always converges absolutely for any finite dimensional matrices over complex numbers.
Applying it to our case, we have
\begin{align}
	U P U^\dagger 
	&=
	\sum_{n=0}^\infty \frac{(-it)^n}{n!} [H,P]_n
	=
	P - it [H,P] - \frac{t^2}{2} [H,[H,P]] + \cdots \label{eqn:real-time-series}\\
	&=
	\sum_{\cluster V} \frac{\lambda^{\cluster V}}{\cluster V!} \mdiff_{\cluster V} (U P U^\dagger)
\end{align}
where in the second line we re-express the $t$-series as a multivariate Taylor series in~$\lambda_a$.
This is our series: \cref{thm:real-time1} follows from \cref{eqn:real-time-series} upon taking $q_m(\lambda) := \frac{(-i)^m}{m!} [H(\lambda),P]_m$.

We now examine a cluster derivative with respect to $\cluster V$, which has total weight $n = \abs{\cluster V}$.
Let us enumerate all elements of~$\cluster V$ as $a_1,a_2,\ldots,a_n$;
this is a list of nodes of~$\graph$ with no particular order 
and the elements $a_j$ are repeated as many times as their multiplicities.
In this context, the cluster derivative $\mdiff_{\cluster V} (U P U^\dagger) = [\diff_{a_1}\cdots\diff_{a_n} UPU^\dagger]|_{\lambda = (0,\ldots,0)}$ is a constant matrix which comes from evaluating a derivative at the origin of the $\lambda$-space $[-1,1]^\numTerms$.
\begin{align}
\mdiff_{\cluster V} (UPU^\dagger)
&=
\mdiff_{\cluster V} \frac{(-it)^n}{n!} [\underbrace{H,[H,\cdots[H}_n,P] \cdots ]\, ] \label{eq:diffVnc}\\
&=
\frac{(-it)^n}{n!}  \sum_{\sigma \in S_n} [\diff_{a_{\sigma(1)}} H, [\diff_{a_{\sigma(1)}} H, \cdots[ \diff_{a_{\sigma(n)}} H, P]\cdots]\,] \tag*{by the Leibniz rule}
\nonumber\\
&=
\frac{(-it)^n}{n!}  \sum_{\sigma \in S_n} [E_{a_{\sigma(1)}}, [E_{a_{\sigma(1)}}, \cdots[ E_{a_{\sigma(n)}}, P]\cdots]\,]
\nonumber
\end{align}
where $S_n$ is the permutation group on $\{1,2,\ldots,n\}$.
From \Cref{eq:diffVnc}, the rest of \cref{thm:real-time} will follow.
\begin{lemma}
For any cluster $\cluster V$ on $\graph$,
if $\cluster V \sqcup \{(p,1)\}$ is disconnected on $\graph(P)$, then $\mdiff_{\cluster V} (UPU^\dagger) = 0$.
\end{lemma}
\begin{proof}
Consider a term in \cref{eq:diffVnc}, which we can label as $[E_{a_{1}}, [E_{a_{2}}, \cdots[ E_{a_{n}}, P]\cdots]\,]$ without loss of generality.
If $\cluster V \sqcup \{(p,1)\}$ is disconnected, then there exists an $k \in [n]$ such that $a_k$ is disconnected from all of $a_{k+1},\ldots, a_n, p$ (otherwise, every $a_k$ would have a path to $p$ by strong induction, making the cluster connected).
Consequently, $E_{a_k}$ commutes with the intermediate commutator $C_{k+1} = [E_{a_{k+1}}, \cdots[ E_{a_n}, P]\cdots]$, which is supported on $R = \Supp(P) \cup \bigcup_{j=k+1}^n \Supp(E_{a_{j}})$.
This means that the next $C_k = [E_{a_{k}}, C_k]$ is zero and so the whole term $[E_{a_{1}}, [E_{a_{2}}, \cdots[ E_{a_{n}}, P]\cdots]\,]$ is zero.
This argument applies to every term, so the whole sum, and the cluster derivative, must also be zero.
\end{proof}

By this lemma, \cref{thm:real-time2} follows immediately, since $\lambda_a$ is present in $q_m$, then there must be a cluster $\cluster V$ of size $m$ such that $a \in \cluster V$ and $\cluster V \sqcup \{(p,1)\}$ is connected.
This implies that the distance between $a$ and $p$ is at most $m$.
Similarly, the number of monomials of $q_m$ can be bounded by the number of weight-$m$ connected clusters in $\graph$ neighboring $p$ in $\graph(P)$.
By \cref{statement:count-clusters}, this can be bounded by $\max(L,\degree+1)e \degree (1+e(\degree-1))^{m-1}$, where the additional factor of $\max(L,\degree+1)$ comes from needing to count clusters that start at any of the terms adjacent to $p$.
This gives \cref{thm:real-time3}.
The lemma below gives \cref{thm:real-time4}.

\begin{lemma}
For any cluster $\cluster V$ with $\abs{\cluster V} = n$, we have 
$\|\mdiff_{\cluster V} (UPU^\dagger)\| \le 2^n \abs{t}^n$.
\end{lemma}
\begin{proof}
The norm of a nested commutator in the last line of \cref{eq:diffVnc} is at most $2^n$.
\end{proof}

Finally, for the time complexity results, note the same algorithm for computing clusters works in this setting, giving \cref{thm:real-timeA}.
To compute the series $q_m$, one could use the same approach as \cref{sec:computingClusterDerivatives}, but we take a simpler and faster approach: we have an explicit form for the series, \cref{eqn:real-time-series}, so all we need to do is compute the commutators $[H, P]_n$ iteratively, for $n$ from $1$ to $m$.
We can maintain $[H, P]_n$ as a sum over clusters of monomials $\lambda^{\cluster V}$ with corresponding integer matrices (where each integer is bounded by $2^nn!$ by \cref{thm:real-time4}), of which faithful representations can be maintained as done in \cref{sec:computingClusterDerivatives}.
For each of these integer matrices $X$, one can compute the corresponding commutator $[H, X]$ in $O(4^n\poly(n))$ time, giving the matrices for the next commutator $[H, P]_{n+1}$.
This gives the specified runtime.

\subsection{Learning algorithm}

Our learning algorithm in the ``real-time dynamics'' setting will be essentially the same as that of learning from the ``Gibbs state'' setting.
For each node $a$ of $\graph$, 
choose $P_a$ to be any single-qubit Pauli that anticommutes with $E_a$, 
and let $Q_a = i[P_a,E_a] = 2 i P_a E_a$.
Define $\fun_a = \fun(Q_a,P_a)$.
Our learning algorithm consists of two parts.
\begin{enumerate}
\item 
Find estimates $\tilde \fun_a$ such that $\abs*{\tilde \fun_a - \fun_a} \le t\varepsilon$ for all $a \in [\numTerms]$.
\item
Approximately invert the function $\{\lambda_a\} \mapsto \{\fun_a\}$ given $\tilde\fun_a$.
\end{enumerate}

\begin{lemma}\label{lem:real-stepone}
    Consider a Hamiltonian $\{(a,E_a,\lambda_a): a \in [\numTerms]\}$ on $\numQubits$ qubits.
    We can find estimates $\tilde \fun_a$ such that $\abs*{\tilde \fun_a - \fun_a} \le t\varepsilon$ for all $a \in [\numTerms]$, with probability at least $1-\delta$, using only $O(\frac{\degree}{t^2\eps^2}\log(\frac{\numTerms}{\delta}))$ applications of $U$ and with time complexity $O(\frac{\numQubits\degree}{t^2\eps^2}\log(\frac{\numTerms}{\delta}))$.
\end{lemma}

\begin{proof}
Recall that in \cref{lem:stepone}, we argued that, from one copy of $\rho$, it's possible to generate a bounded random variable $Y_a \in [-1,1]$ that is an unbiased esimator of $\Tr(E_a \rho)$.
Moreover, for a set of terms $S \subset [\numTerms]$, it is possible to generate $Z_a$'s for all such $a \in S$ from one copy of $\rho$, provided the $E_a$'s are non-overlapping, or in other words, provided $S$ is an independent set in the dual interaction graph $\graph$.
We will use this again here; the main challenge is that for each $E_a$ we wish to measure against a different $P_a$, so it's not immediately clear how to use one application of $U$ to produce estimators for multiple different terms.
We resolve this by thinking of $\rho$ as a distribution over states, and then conditioning on this distribution to measure expectations over what are effectively different mixed states.

Consider the procedure of sampling a string $s \sim \{0,1,+,-,i,-i\}^\numQubits$ uniformly at random, and then preparing the state $\rho = |s_1\rangle\langle s_1| \otimes \cdots \otimes |s_\numQubits\rangle\langle s_{\numQubits}|$.
Notice that if we discard our initial string $s$, thereby averaging over $s$, then $\rho$ is the maximally mixed state; further, if we discard the entire initial string apart from one qubit $i$, then $\rho = \frac I2 \otimes \cdots \otimes \frac I2 \otimes |s_i\rangle\langle s_i| \otimes \frac I2 \otimes \cdots \otimes \frac I2$.
Note that, for $s = 0,1,+,-,i,-i$, $2|s\rangle \langle s| - I$ is a Pauli matrix $Z, -Z, X, -X, Y, -Y$, respectively.

Suppose we apply $U$ to $\rho$, measure it on the support of $E_a$ to get the unbiased estimator $Y_a$ of $\Tr(E_a U \rho U^\dagger)$, and define the following random variable.
\begin{align*}
    Z_a &= \begin{cases}
        Y_a & \text{if } s_{\Supp(P_a)} \text{ satisfies } 2|s_{\Supp(P_a)}\rangle \langle s_{\Supp(P_a)}| - I = P_a \\
        0 & \text{otherwise}
    \end{cases}
\intertext{Here, we abuse notation by using $P_a$ to refer both to the 1-qubit Pauli and the $n$-qubit tensor of that Pauli with the identity matrix.
The random variable $Z_a$ is bounded in $[-1,1]$ because $Y_a$ is, and furthermore,}
    \E[Z_a] &= \Pr_s\Big[2|s_{\Supp(P_a)}\rangle \langle s_{\Supp(P_a)}| - I = P_a\Big] \E_s\Big[Y_a \mid 2|s_{\Supp(P_a)}\rangle \langle s_{\Supp(P_a)}| - I = P_a\Big] \\
    &= \frac16\Tr\Big(E_a U \Big(\E[\rho \mid 2|s_{\Supp(P_a)}\rangle \langle s_{\Supp(P_a)}| - I = P_a]\Big) U^\dagger \Big) \\
    &= \frac16\Tr\Big(E_a U \Big(\frac I2 \otimes \cdots \otimes \frac I2 \otimes \frac{I + P_a}{2} \otimes \frac I2 \otimes \cdots \otimes \frac I2 \Big) U^\dagger \Big) \\
    &= \frac{1}{12}\Tr\Big(E_a U P_a U^\dagger \Big)
\end{align*}
So, $12 Z_a$ is an unbiased estimator for $\fun_a$, and the rest of the result follows exactly like it did in \cref{lem:stepone}: one can use one copy of $\rho$ to get multiple estimators $Z_a$, provided their corresponding terms do not overlap.
Thus, in $\degree + 1$ rounds of $O(\frac{1}{t^2\eps^2}\log\frac{\numTerms}{\delta})$ applications of $U$ each, one can get $O(\frac{1}{t^2\eps^2}\log\frac{\numTerms}{\delta})$ copies of $Z_a$ for every $a \in [\numTerms]$.
By Chernoff bound, each rescaled average $12 \bar{Z}_a$ will then satisfy $\abs*{12\bar{Z}_a - \fun_a} \leq \eps$ with probability $\geq 1-\delta/\numTerms$, and so they all satisfy $\abs*{12\bar{Z}_a - \fun_a} \leq \eps$ with probability $\geq 1-\delta$.
The time complexity is the same, since the only change to the procedure is doing one additional $O(1)$-time check per sample $Z_a$.
\end{proof}

Finally, for the second, classical part of the algorithm, note that operators $P_a$ and $Q_a$ are chosen so that the leading term of $\fun_a$ is a known constant multiple of $\lambda_a$:
\begin{align}
	\fun(Q,P) 
	&= \frac{\Tr}{\fulldimension}\left( Q P - i t Q[H,P] + \cdots \right)\label{eq:funQP}\\
	&= t \frac{\Tr}{\fulldimension} \left( [P, E_a] \sum_b \lambda_b [E_b, P] + \cdots \right)\nonumber\\
	&= 4t \lambda_a + \cdots  \nonumber
\end{align}
where the last line uses the orthonormality of the Hamiltonian terms.
This observation implies that the Jacobian is ``band-diagonal'', and suffices, along with \cref{thm:real-time}, for the full analysis of the Newton--Raphson method in \cref{sec:newton-raphson} to go through identically.
The only difference here is that $\beta$ is replaced with $-4t$, so this part takes time $O(\numTerms L / \eps)\poly(\degree \log(1/t\eps))$.
The time complexity of the quantum part dominates.

\begin{remark}
Since the series expansion in \cref{eq:funQP} is only shown to converge for $t < t_c$
where $t_c = 1/\poly(\degree)$,
we can only claim that our algorithm works for small enough $t$.
In the learning problem from Gibbs states, 
the analogous condition $\beta \le 1/\poly(\degree)$
is due to the fact that our approach cannot handle arbitrarily low temperature;
the sample complexity result~\cite{AAKS21} shows that learning is feasible for all temperatures, at least in an information-theoretic sense.
In contrast, in the learning problem from real-time evolution,
it is fundamental that we have to restrict the evolution time to be smaller
than some constant set by $\degree$;
for a certain long time, the learning is simply impossible.
Consider a Hamiltonian $H = -\lambda I + \lambda(I+Z_1)(I+Z_2)\cdots (I+Z_n)$ 
on $n$ qubits where $Z_j$ is the Pauli $Z$ on qubit $j$.
This Hamiltonian is the sum of all nonidentity products of $Z$'s with a uniform coefficient $\lambda \in [-1,1]$,
and obeys our normalization conditions for Hamiltonians.
The intersection degree $\degree$ is exponentially large in $n$.
The eigenspectrum of $H$ consists of just two values, $(2^n-1)\lambda$ and $-\lambda$.
Hence, $e^{-it H} \propto I$ if $t = 2\pi / 2^n\lambda = \Theta((\degree \lambda)^{-1})$.
Since $\lambda$ is unknown, we conclude that 
no general algorithm can determine $\lambda$ unless we restrict $t$ to be smaller than $1/\poly(\degree)$.
\end{remark}

\section{Algorithm for parameter learning of Markov random fields} \label{sec:mrfs}

In this section, we will prove a folklore result by giving a simple algorithm for parameter learning of Markov random fields.
By the Hammersley--Clifford theorem, a Markov random field $\mathcal{D}$ over $\{-1,+1\}^{\numQubits}$ can be written as
\begin{align*}
    \Pr_{Z \sim \mathcal{D}}[Z = z] \propto \exp(\sum_{S \subset [\numQubits]} \psi_S(z_S))
\end{align*}
for some functions $\psi_S: \mathbb{R}^{\abs{S}} \to \mathbb{R}$, where $z_S =(z_i)_{i \in S}$.
Typically, the sum is restricted to be over $S$ with size at most some constant.
By writing every $\psi_S$ as a sum of products of variables, this expression becomes
\begin{align*}
    \Pr_{Z \sim \mathcal{D}}[Z = z] \propto \exp(\sum_{S \subset [\numQubits]} \lambda_S z^S),
\end{align*}
where $z^S := \prod_{i \in S} z_i$.
For the parameter learning problem, we assume we already know the structure of the MRF, so suppose we are given a hypergraph $G = (V=[\numQubits], E)$ on $\numQubits$ vertices such that
\begin{align*}
    \Pr_{Z \sim \mathcal{D}}[Z = z] \propto \exp(-\beta\sum_{S \in E} \lambda_S z^S).
\end{align*}
Here, $-\beta$ is a rescaling factor so that we can assume without loss of generality that $\lambda_S \in [-1, 1]$ for all $S \in E$.
We will interpret the $\lambda$ parameters as a vector in $[-1,1]^{\abs{E}}$ and $\beta \in (0,\infty)$.
This is precisely a Gibbs state of a classical Hamiltonian, following the definitions given in \cref{sec:notation}.
Further, this is the setting where each term is is a product of Paulis, since each $z^S$ is a product of Pauli $Z$ operators.

For a vertex $i \in [\numQubits]$, let $E_i = \{S \in E \mid i \in S\}$ be the set of hyperedges containing $i$ and let $N_i = (\cup_{S \in E_i} S) \setminus \{i\}$ be the neighborhood of $i$.
Our algorithm will depend on two parameters: maximum degree $d := \max_{i \in [\numQubits]} \abs{E_i}$ and an ``average order'' parameter $L := \max_{i \in [\numQubits]} \frac{1}{d}\abs{N_i \cup \{i\}}$.
We do not consider the empty graph, so that $d, Ld \geq 1$.

We will need the following lemma.

\begin{lemma}[Version of Lemma 2.1, \cite{Bresler2015}] \label{ising-unbiased}
    For any node $u \in V$, subset $S \subset V$, and configuration $x_S \in \{\pm 1\}^{|S|}$,
    \[
        \min_{b \in \{-1,+1\}}\Pr[X_u = b \mid X_S = x_S]
        \geq \frac{1}{\exp(2\beta d) + 1}
        \geq \frac{1}{2}\exp(-2\beta d).
    \]
\end{lemma}
\begin{proof}
    First, using the Markov property
    \begin{align*}
        \abs{\E[X_u \mid X_{V \setminus \{u\}} = x_{V \setminus \{u\}}]}
        &= \abs{\frac{\sum_{x_u} x_u\exp(-\beta(\sum_{S \in E} \lambda_S x^S))}{\sum_{x_u} \exp(-\beta(\sum_{S \in E} \lambda_S x^S))}} \\
        &= \abs{\frac{\sum_{x_u} x_u\exp(-\beta(\sum_{S \in E_u} \lambda_S x^{S \setminus \{u\}}x_u))}{\sum_{x_u} \exp(-\beta(\sum_{S \in E_u} \lambda_S x^{S \setminus \{u\}}x_u))}} \\
        &= \abs{\frac{\exp(-\beta(\sum_{S \in E_u} \lambda_S x^{S \setminus \{u\}})) - \exp(\beta(\sum_{S \in E_u} \lambda_S x^{S \setminus \{u\}}))}{\exp(-\beta(\sum_{S \in E_u} \lambda_S x^{S \setminus \{u\}})) + \exp(\beta(\sum_{S \in E_u} \lambda_S x^{S \setminus \{u\}}))}} \\
        &= \frac{
        \exp(\beta\abs*{\sum_{S \in E_u} \lambda_S x^{S \setminus \{u\}}}) -
        \exp(-\beta\abs*{\sum_{S \in E_u} \lambda_S x^{S \setminus \{u\}}})}{
        \exp(\beta\abs*{\sum_{S \in E_u} \lambda_S x^{S \setminus \{u\}}}) +
        \exp(-\beta\abs*{\sum_{S \in E_u} \lambda_S x^{S \setminus \{u\}}})} \\
        &= 1 - \frac{2}{\exp(2\beta\abs*{\sum_{S \in E_u} \lambda_S x^{S \setminus \{u\}}}) + 1}
    \end{align*}
    Further, by the tower property of conditional expectation and Jensen's inequality,
    \begin{align*}
        &|\E[X_u \mid X_S = x_S]| \\
        &=|\E[\E[X_u \mid X_{V \setminus \{u\}} = x_{V \setminus \{u\}}] \mid X_S = x_S]| \\
        &\leq\E[|\E[X_u \mid X_{V \setminus \{u\}} = x_{V \setminus \{u\}}]| \mid X_S = x_S] \\
        &=\E\Big[ 1 - \frac{2}{\exp(2\beta\abs*{\sum_{S \in E_u} \lambda_S x^{S \setminus \{u\}}}) + 1} \mid X_S = x_S\Big] \\
        &= 1 - \frac{2}{\exp(2\beta\abs*{\sum_{S \in E_u} \lambda_S x^{S \setminus \{u\}}}) + 1} \\
        &\leq 1 - \frac{2}{\exp(2\beta d) + 1}
    \end{align*}
    For a $\{\pm 1\}$-valued random variable $X$, $\min\{\Pr[X = 1], \Pr[X = -1]\} = \frac12(1 - |\E[X]|)$, so
    \begin{equation*}
        \min_{b \in \{-1, +1\}}\{\Pr[X_u = b \mid X_S = x_S]\}
        = \frac12(1 - |\E[X_u \mid X_S = x_S]|)
        \geq \frac{1}{\exp(2\beta d) + 1}. \qedhere
    \end{equation*}
\end{proof}

Define the sigmoid function $\sigma(x) := \frac{e^x}{1+e^x}$.
We will need the following fact about the sigmoid:
\begin{lemma}[Claim~4.2, \cite{km17}] \label{sigmoid-lipschitz}
    For all $x, y$, $|\sigma(x) - \sigma(y)| \geq \exp(-\abs{x}-3)\min(1, \abs{x-y})$.
\end{lemma}

\begin{theorem} \label{thm:mrfs}
    Fix $\eps \in (0,1)$.
    Given samples from the MRF $X^{(1)},\ldots,X^{(T)}$ with $T = \Theta(\exp(8\beta Ld^2 + 2Ld)\frac{1}{\beta^2\eps^2}\log\frac{\numQubits}{\delta})$, we can compute an estimate~$\hat{\lambda}$ such that $\|\hat{\lambda} - \lambda\|_\infty \leq \eps$ with probability $\geq 1-\delta$.
    The algorithm takes $O(T\numQubits(Ld^2+d2^d))$ time.
\end{theorem}
For low-intersection Hamiltonians (as defined in the introduction), $L = O(1)$ and $d \leq \degree + 1 = O(1)$.
For constant $L$ and $d$, the sample complexity and time complexity of learning a classical Hamiltonian to $\ell_\infty$ error $\eps$ become
\[
    \frac{\exp(O(\beta))}{\beta^2\eps^2}\log\frac{\numQubits}{\delta}
    \text{ and }
    \frac{\exp(O(\beta))}{\beta^2\eps^2}\numQubits\log\frac{\numQubits}{\delta},
\]
respectively.

\begin{proof}
    First, fix a particular $v \in [n]$, and consider conditioning on its neighbors $N_v$.
    The distribution on $x_v$ after conditioning is
    \begin{align*}
        \Pr[X_v = x_v \mid X_{N_v} = x_{N_v}]
        = \sigma(2\beta \sum_{S \in E_v} \lambda_S x^S).
    \end{align*}
    We now show that it suffices to be able to estimate such conditional probabilities, for a particular setting of $X_{N(v)}$.
    Let $q_x^{(v)}$ be the argument inside the $\sigma$ above, so
    \begin{align*}
        q_x^{(v)} &:=
        \sigma^{-1}(\Pr[X_v = x_v \mid X_{N_v} = x_{N_v}])
        = 2\beta \sum_{S \in E_v} \lambda_S x^{S}.
    \end{align*}
    Note that $q_x^{(v)}$ only depends on those $x_u$ where $u \in N_v$.
    We will argue below that we can get an estimate of our conditional probability $\sigma(q_x^{(v)})$ to $\exp(-\abs{q_x^{(v)}}-3)\min(0.5, 2\beta\eps)$ error.
    We will invert $\sigma$ on this estimate to get an estimate $\hat{q}_x^{(v)}$ for $q_x^{(v)}$, so we denote the original estimate to be $\sigma(\hat{q}_x^{(v)})$.
    By \cref{sigmoid-lipschitz},
    \begin{align*}
        \min(1, \abs{q_x^{(v)} - \hat{q}_x^{(v)}})
        &\leq \exp(\abs{q_x^{(v)}}+3)\abs{\sigma(q_x^{(v)}) - \sigma(\hat{q}_x^{(v)})}
        \leq \min(0.5, 2\beta\eps) \\
        \text{ so }
        \abs{q_x^{(v)} - \hat{q}_x^{(v)}}
        &\leq \min(0.5, 2\beta\eps)
        \leq 2\beta\eps
    \end{align*}
    So $\hat{q}_x^{(v)}$ is an estimate for $q_x^{(v)}$ up to additive $2\beta\eps$ error.
    Now, we show that we can use these $\hat{q}_x^{(v)}$'s to get a good estimate of the parameters $\{\lambda_S\}_{S \in E}$.

    Suppose we want to know $\lambda_S$.
    Pick $v \in S$ and consider all $T \in E_v$ that are not $S$.
    If $T \nsubseteq S$, then choose some vertex $u \in T \setminus S$ and place it in a set $N_{\text{out}}$.
    Otherwise, $T \subsetneq S$, so we choose some vertex $u \in S \setminus T$ and place it in a set $N_{\text{in}}$.
    Note that $N_{\text{in}} \subset S$ and $N_{\text{out}} \subset [\numQubits] \setminus S$, so they are disjoint.\footnote{For some intuition, in the conventional setting where vertices are on a lattice and a term $S$ is a connected piece of the lattice, one can think of taking $N_{\text{in}}$ to be $S$ and $N_{\text{out}}$ to be the neighborhood of $S$ (or the boundary of $S^c$).}
    We can get an estimate for $\lambda_S$ by averaging $q_z^{(v)}$ over the coordinates of $N_{\text{in}}$ and $N_{\text{out}}$ in a particular way; for a set of indices $I \subset [n]$, let $P_I = \{z \in \{\pm 1\}^\numQubits \mid z_i = 1 \text{ for all } i \not\in I\}$ be a slice of the Hamming cube.
    \begin{multline*}
        \frac{1}{2\beta} \E_{z \sim P_{N_{\text{in}} \cup N_{\text{out}}}} [z^{N_{\text{in}}} q_z^{(v)}]
        = \Big[\E_{z_{N_{\text{in}}}} \E_{z_{N_{\text{out}}}} \Big[z^{N_{\text{in}}} \sum_{T \in E_v} \lambda_T z^T\Big]\Big]_{z = \vec{1}} \\
        = \Big[\E_{z_{N_{\text{in}}}} \Big[\sum_{\substack{T \in E_v \\ T \cap N_{\text{out}} = \varnothing}} \lambda_T z^{N_{\text{in}}} z^T\Big]\Big]_{z = \vec{1}}
        = \Big[\sum_{\substack{T \in E_v \\ T \cap N_{\text{out}} = \varnothing \\ T^c \cap N_{\text{in}} = \varnothing}} \lambda_T z^{T \setminus N_{\text{in}}}\Big]_{z = \vec{1}}
        = \Big[\lambda_S z^{S \setminus N_{\text{in}}}\Big]_{z = \vec{1}}
        = \lambda_S
    \end{multline*}
    Suppose we have an estimate of $q_z^{(v)}$, $\hat{q}_z^{(v)}$, to $2\beta \eps$ error.
    Then
    \begin{align*}
        \abs{\frac{1}{2\beta}\E_{z \sim P_{N_{\text{in}} \cup N_{\text{out}}}} z^{N_{\text{in}}} \hat{q}_z^{(v)} - \lambda_S}
        &= \frac{1}{2\beta}\abs{\E_{z \sim P_{N_{\text{in}} \cup N_{\text{out}}}} z^{N_{\text{in}}} (\hat{q}_z^{(v)} - q_z^{(v)})} \\
        &\leq \frac{1}{2\beta}\E_{z \sim P_{N_{\text{in}} \cup N_{\text{out}}}} \abs*{z^{N_{\text{in}}}} \abs*{\hat{q}_z^{(v)} - q_z^{(v)}} \\
        &\leq \frac{1}{2\beta}\E_{z \sim P_{N_{\text{in}} \cup N_{\text{out}}}} \abs*{z^{N_{\text{in}}}} 2\beta\eps \\
        &= \eps
    \end{align*}
    Note that $\abs{P_{N_{\text{in}} \cup N_{\text{out}}}} = 2^{\abs{N_{\text{in}}} + \abs{N_{\text{out}}}} \leq 2^d$.
    So, now we just need to show how to estimate $\sigma(q_z^{(v)}) = \Pr[X_v = z_v \mid X_{N_v} = z_{N_v}]$ to $\exp(-\abs*{q_x^{(v)}}-3)\min(0.5, 2\beta\eps)$ error for $\leq 2^{d+1}$ choices of $z_v, z_{N_v}$, over all $\numTerms$ choices of $S$.
    Since $\abs*{q_x^{(v)}} \leq 2\beta d$ always, it suffices to estimate to $\Theta(\exp(-2\beta d)\min(0.5, \beta\eps))$ error.

    Recall that estimating the probability $p$ of an event occurring to $\alpha$ relative error with probability $\geq 1-\delta$ requires $\Theta(\frac{1}{p\alpha^2}\log\frac{1}{\delta})$ samples.
    (The estimator we use is simply the empirical probability of the event over the samples, and the proof follows from a Chernoff bound).
    So, if we pull $\Theta(\frac{1}{p_{\min}\alpha^2}\log\frac{\numTerms 2^{d+2}}{\delta})$ samples of our Markov random field, then we can get an estimate to any particular choice of $\Pr[X_v = z_v \text{ and } X_{N_v} = z_{N_v}]$ and $\Pr[X_{N_v} = z_{N_v}]$ to $\alpha$ relative error that is correct with probability $\geq 1 - \frac{\delta}{\numTerms 2^{d+2}}$, provided that $p_{\min}$ is chosen to be smaller than this probability.
    By union bound, we can get an estimate to all the probabilities we would need to compute the $C$ conditional probabilities with probability $\geq 1-\delta$ using this number of samples, provided $p_{\min}$ is smaller than \emph{all} the probabilities we wish to estimate.
    Also recall that by \cref{ising-unbiased},
    \begin{align*}
        \Pr[X_{S} = x_{S}] = \prod_{i=1}^{\abs{S}} \Pr[X_{u_i} = x_{u_i} \mid X_{u_j} = x_{u_j} \text{ for } j < i]
        \geq 2^{-\abs{S}}\exp(-2\beta \abs{S}d).
    \end{align*}
    Since we want to compute these probabilities for $S = N_v$ or $S = N_v \cup \{v\}$, we can take $p_{\min} = \exp(-2\beta(Ld)d - Ld)$.

    With these estimates, we can get good estimates to the conditional probabilities, assuming $\alpha$ is sufficiently small:
    \begin{align*}
        \frac{\Pr[X_v = z_v \text{ and } X_{N_v} = z_{N_v}](1 \pm \alpha)}{\Pr[X_{N_v} = z_{N_v}](1 \pm \alpha)}
        &\in \frac{\Pr[X_v = z_v \text{ and } X_{N_v} = z_{N_v}]}{\Pr[X_{N_v} = z_{N_v}]}(1 \pm O(\alpha)) \\
        &= \Pr[X_v = z_v \mid X_{N_v} = z_{N_v}](1 \pm O(\alpha))
    \end{align*}
    We set $\alpha = \Theta(\exp(-2\beta d)\min(0.5, \beta\eps)) < 0.5$ to conclude.
    The number of samples we need is
    \begin{align*}
        T &= \Theta\Big(\frac{\exp(2\beta(Ld)d + Ld)}{(\exp(-2\beta d)\min(1, \beta\eps))^2}\log\frac{\numTerms 2^{d+2}}{\delta}\Big) \\
        \intertext{Using that $\frac{1}{\min(1, \beta^2\eps^2)} = \frac{1}{\beta^2\eps^2}\max(1, \beta\eps)^2 \leq \frac{1}{\beta^2\eps^2}\exp(2\beta\eps)$, this is}
        &= \Theta\Big(d\exp(2\beta(Ld)d + Ld + 4\beta d + 2\beta\eps)\frac{1}{\beta^2\eps^2}\log\frac{\numTerms}{\delta}\Big).
    \end{align*}
    The bound in the theorem statement comes from simplifying and performing rough bounds on the above expression.
    \todo[inline]{ET: The time complexity needs the most verification, I'm pretty confident about everything else.}
    We can run this algorithm in time $O(T \numTerms(Ld+2^d))$, since for each term, we can compute the $2^d$ empirical conditional probabilities for it by taking $Ld$ time per sample to sort them into the various (disjoint) events.
    Classifying each sample only requires looking at the bits corresponding to $v$ and its neighbors, so checking all takes $O(T \numTerms Ld)$ time, and we need to do this for all terms.
    The result in the statement comes from taking $\numTerms \leq \numQubits d$.
\end{proof}

\addcontentsline{toc}{section}{References}
\bibliography{main}
\bibliographystyle{alphaurl}
\end{document}